\documentclass[10pt]{article}

\usepackage{mathpazo}
\usepackage{bbm}
\usepackage{multirow}
\usepackage{amsmath,amssymb,amsthm,amsfonts,latexsym,bbm,xspace,graphicx,float,mathtools,dsfont}
\usepackage[backref,colorlinks,citecolor=blue,bookmarks=true]{hyperref}
\usepackage[dvips, paper=letterpaper, top=1in, bottom=.75in, left=1in, right=1in, nohead, includefoot, footskip=.25in]{geometry}
\usepackage{color}
\usepackage{makecell}
\usepackage{comment}
\usepackage{algorithm}
\usepackage{algorithmic}

\usepackage{tweaklist}
\usepackage{accents}
\usepackage{appendix}
\usepackage{cleveref}
\usepackage{graphicx}

\usepackage{accents}

\newtheorem{theorem}{Theorem}
\newtheorem{corollary}{Corollary}
\newtheorem{lemma}{Lemma}

\newtheorem{claim}{Claim}

\newtheorem{definition}{Definition}

\newcommand{\norm}[1]{\left\lVert#1\right\rVert}
\newcommand{\round}[1]{\left\lfloor#1\right\rfloor}

\newcommand{\rev}{\textsc{Rev}}
\newcommand{\revT}{\textsc{Rev}_T}

\newenvironment{prevproof}[2]{\noindent {\em {Proof of {#1}~\ref{#2}:}}}{$\Box$\vskip \belowdisplayskip}

\newcommand{\poly}{{\rm poly}}
\newcommand{\opt}{\text{OPT}}
\newcommand{\todo}[1]{{\bf \color{red} TODO: #1}}
\newcommand{\yangnote}[1]{{\color{blue}{#1}}}

\newcommand{\notshow}[1]{{}}

\DeclareMathAccent{\wtilde}{\mathord}{largesymbols}{"65}
\DeclareMathOperator{\Prob}{Pr}
\DeclareMathOperator{\E}{E}

\def \hM {{\widehat{M}}}

\def\supp{\text{supp}}
\def \hDD  {{\mathcal{\widehat{D}}}}
\def \DD  {{\mathcal{D}}}
\def \CC  {{\mathcal{C}}}
\def \uDD{{\overline{\mathcal{D}}}}
\def \lDD{{\underline{\mathcal{D}}}}
\def \hlDD{{\underline{\mathcal{{\widehat{D}}}}}}
\def \utDD{\widetilde{\mathcal{D}}}
\def \ltDD{\underaccent{\wtilde}{\mathcal{D}}}

\def \hFF  {{\mathcal{\widehat{F}}}}
\def \rfun{r^{(\ell,\delta)}}
\def \hb  {{{\hat{b}}}}
\def \uFF{{\overline{\mathcal{F}}}}
\def \Mli{{M_1^{(\ell)}}}
\def \Mlii{{M_2^{(\ell)}}}
\def \hMl{{\widehat{M}^{(\ell)}}}
\def \el{\varepsilon^{(\ell)}}

\def \CC {{\mathcal{C}}}

\def \LL  {{\mathcal{L}}}

\def \FF  {{\mathcal{F}}}

\def \E  {{\mathbb{E}}}

\def \L{{\mathcal{L}}}

\DeclareMathOperator{\argmax}{argmax}

\definecolor{MyGray}{rgb}{0.8,0.8,0.8}

\title{From Robustness to Learnability: a General Approach to the Sample Complexity in Multi-item Auctions}

\title{Learnability and Robustness of Multi-Item Mechanisms without Item-Independence}

\title{Multi-Item Mechanisms without Item-Independence: Robustness implies Learnability}

\title{Multi-Item Mechanisms without Item-Independence: Learnability via Robustness}

{\author{Johannes Brustle\footnote{Supported by the NSERC Discovery Award RGPIN-2015-06127 and FRQNT Award 2017-NC-198956.}\\
McGill University, Canada\\ johannes.brustle@mail.mcgill.ca
 \and Yang Cai\footnote{Supported by the NSF Award CCF-1942583 (CAREER) and a Sloan Foundation Research Fellowship. Part of Cai's work was done under the support of the NSERC Discovery Award RGPIN-2015-06127 and FRQNT Award 2017-NC-198956.}\\
Yale University, USA\\
yang.cai@yale.edu \and Constantinos Daskalakis\footnote{Supported by NSF Awards IIS-1741137, CCF-1617730 and CCF-1901292, by a Simons Investigator Award, by the DOE PhILMs project (No. DE-AC05-76RL01830), by the DARPA award HR00111990021, by a Google Faculty award, by the MIT Frank Quick Faculty Research and Innovation Fellowship, and by the MIT-IBM Watson AI Lab.}\\ MIT, USA\\ costis@csail.mit.edu}}

\begin{document}
\maketitle
 \begin{abstract}
We study the sample complexity of learning  revenue-optimal multi-item auctions. We obtain the first set of positive results that go beyond the standard but unrealistic setting of item-independence. In particular, we consider settings where bidders' valuations are drawn from correlated distributions that can be captured by Markov Random Fields or Bayesian Networks --  two of the most prominent graphical models. We establish parametrized sample complexity bounds for learning an up-to-$\varepsilon$ optimal mechanism in both models, which scale polynomially in the size of the model, i.e.~the number of items and bidders, and only exponential in the natural complexity measure of the model, namely either the largest in-degree (for Bayesian Networks) or the size of the largest hyper-edge (for Markov Random Fields).

We obtain our learnability results through a novel and modular framework that involves first proving a robustness theorem. We show that, given only ``approximate distributions'' for bidder valuations, we can learn a mechanism whose revenue is nearly optimal simultaneously for all ``true distributions'' that are close to the ones we were given in Prokhorov distance. Thus, to learn a good mechanism, it suffices to learn approximate distributions. When item values are independent, learning in Prokhorov distance is immediate, hence our framework directly implies the main result of Gonczarowski and Weinberg~\cite{GonczarowskiW18}.
When item values are sampled from more general graphical models, we combine our robustness theorem with novel sample complexity results for learning Markov Random Fields or Bayesian Networks in Prokhorov distance, which may be of independent interest. Finally, in the single-item case, our robustness result can be strengthened to hold under an even weaker distribution distance, the L\'evy distance.
\end{abstract}
\thispagestyle{empty}
\addtocounter{page}{-1}
\newpage

\section{Introduction}
A central problem in Economics and Computer Science is the design of revenue-optimal auctions.  The problem involves a seller who wants to sell one or more items to one or more strategic bidders. As bidders' valuation functions are private, no meaningful revenue guarantee can be achieved without any information about these functions. To remove this impossibility, it is standard to make a {\em Bayesian assumption,} whereby a joint distribution from which bidders' valuations are drawn is assumed common knowledge, and the goal is to design an auction that maximizes  expected revenue with respect to this distribution. 

In the {\em single-item setting}, a celebrated result by Myerson characterizes the optimal auction when bidder values are independent~\cite{Myerson81}. The quest for optimal {\em multi-item auctions} has been quite more challenging. It has been recognized that revenue-optimal multi-item auctions can be really complex and may exhibit counter-intuitive properties~\cite{HartN13, HartR12, BriestCKW10, DaskalakisDT13, DaskalakisDT14}. As such, it is doubtful that there is a clean characterization similar to Myerson's for the optimal multi-item auction. On the other hand, there has been significant recent progress in efficient computation of revenue-optimal auctions~\cite{ChawlaHK07,ChawlaHMS10,Alaei11,CaiD11b,AlaeiFHHM12,CaiDW12a,CaiDW12b,CaiH13,CaiDW13b,AlaeiFHH13,BhalgatGM13,DaskalakisDW15}. This progress has enabled the identification of {\em simple auctions} (mostly variants of sequential posted pricing mechanisms) that achieve constant factor approximations to the optimum revenue~\cite{BabaioffILW14,Yao15,CaiDW16,ChawlaM16, CaiZ17}, under {\em item-independence} assumptions.\footnote{Intuitively, item independence means that each bidder's value for each item is independently distributed, and this definition has been suitably generalized to set value functions such as submodular or subadditive functions~\cite{RubinsteinW15}.} 

Making Bayesian assumptions in the study of revenue-optimal auctions is both crucial and fruitful. However, to apply the theory to practice, we would need to know the underlying distributions. Where does such knowledge come from? A common answer is that we estimate the distributions through market research or observation of  bidder behavior in previously run auctions. Unavoidably, errors will creep in to the estimation, and a priori it seems possible that the performance of our mechanisms may be fragile to such errors. This has motivated a quest for optimal or approximately optimal mechanisms under imperfect knowledge of the underlying distributions.

This problem has received lots of attention from Theory of Computation recently. The focus has been on whether  optimal or  approximately optimal mechanisms are learnable given sample access to the true distributions. In single-item settings, where Myerson's characterization result applies, it is possible to learn 
up-to-$\varepsilon$ optimal auctions~\cite{Elkind07,ColeR14,MohriM14,HuangMR15,MorgensternR15,DevanurHP16,RoughgardenS16,GonczarowskiN16}.\footnote{The term ``up-to-$\varepsilon$ optimal'' introduced in~\cite{GonczarowskiW18} means an additive $\varepsilon\cdot H$ approximation for distributions supported on $[0,H]$. Under tail assumption on the distribution, it is also possible to obtain $(1-\varepsilon)$-multiplicative approximations.} A recent paper by Guo et al.~\cite{GuoHZ19} provides upper and lower bounds on the sample complexity, which are tight up to logarithmic factors, thereby rendering a nearly complete picture for the single-item case. 

In multi-item settings, largely due to the lack of simple characterizations of optimal mechanisms, results have been sparser. Recent work~\cite{MorgensternR15,GoldnerK16,CaiD17,Syrgkanis17} has shown how to learn simple mechanisms which attain a constant factor of the optimum revenue using polynomially many samples in the number of bidders and items. Last year, a surprising result by Gonczarowski and Weinberg~\cite{GonczarowskiW18} shows that the sample complexity of learning an up-to-$\varepsilon$ optimal mechanism is also polynomial.\footnote{In particular, they learn a mechanism that is $O(\varepsilon)$-truthful and has up-to-$\varepsilon$ optimal revenue.} However, all these results  rely on the \emph{item-independence} assumption mentioned earlier, which limits their applicability. A main goal of our work is the following:

\bigskip \noindent \hspace{0.7cm}\begin{minipage}{0.92\textwidth}
\begin{enumerate}
\item[{\bf Goal I:}] Push the boundary of learning (approximately) optimal multi-item auctions  to the important setting of \textbf{item dependence}.
\end{enumerate}
\end{minipage}

\bigskip Unfortunately, it is impossible to learn approximately optimal auctions from polynomially many samples under general item dependence. Indeed, an exponential sample complexity lower bound has been established by Dughmi et al.~\cite{DughmiLN14} for even a single unit-demand buyer. Arguably, however, in auction settings, as well as virtually any high-dimensional setting, the distributions that arise are not arbitrary. Arbitrary high-dimensional distributions cannot be represented efficiently, and are known to require exponentially many samples to learn or even perform the most basic statistical tests on them; see e.g.~\cite{Daskalakis18} for a discussion. Accordingly a large focus of Statistics and Machine Learning has been on identifying structural properties of high-dimensional distributions, which enable succinct representation, efficient learning, and efficient statistical inference. 
In line with this literature, 
%
%
we propose {\em learning multi-item auctions under the assumption that item values are jointly sampled from a high-dimensional distribution with structure.} 


There are several widely-studied probabilistic frameworks which allow modeling structure in a high-dimensional distribution. In this work we consider two of the most prominent ones: {\em Markov Random Fields (MRFs)} and {\em Bayesian Networks, a.k.a.~Bayesnets,} which are the two most common types of {\em graphical models}. Both MRFs and Bayesnets have been studied in Machine Learning and Statistics  for decades. Both frameworks can be used to express arbitrary high-dimensional distributions. Their advantage, however, is that they are associated with natural complexity parameters which allow tuning the dependence structure in the distributions they model, from product measures all the way up to arbitrary distributions. {In Figure~\ref{fig:costas}, we show a very simple example illustrating how naturally these models  express dependence structure in a distribution. The figure shows a Bayesnet, which samples the values of a buyer for four items. The structure of the Bayesnet implies (see Definition~\ref{def:Bayesnet}) that these values are sampled conditionally independently, conditioning on the value of the variable at the root of the Bayesnet which is the state of the buyer's residence. The node is shaded because we assume that the corresponding variable is not observable.}
The pertinent question  is how we might exploit the structure of the distribution, as captured by the natural complexity parameter of an MRF or a Bayesnet, to efficiently learn a good mechanism. At a high level, there are two components to the problem of learning approximately optimal auctions. One is \emph{inference from samples}, i.e.~extracting  information about the distribution using samples. The other is \emph{mechanism design}, i.e.~constructing a good mechanism using the information extracted. A main goal of our work is:
%

\bigskip \noindent \hspace{0.7cm}\begin{minipage}{0.92\textwidth}
\begin{enumerate}
\item[{\bf Goal II:}] Provide a modular approach for learning multi-item auctions which decouples the Inference and Mechanism Design components, so that one may leverage all techniques from Machine Learning and Statistics to tackle the first and, independently, leverage all techniques from Mechanism Design to address the second.
\end{enumerate}
\end{minipage}

\bigskip Unfortunately, the Statistical and  Mechanism design components are complexly intertwined in  prior work on learning multi-item auctions.  Specifically,~\cite{MorgensternR16,CaiD17,Syrgkanis17, GonczarowskiW18} are PAC-learning  approaches, which require a fine balance between (i) selecting a class of mechanisms that is rich enough to contain an approximately optimal one for a class of distributions; and (ii) having small enough statistical complexity so that the performance of all mechanisms in the class on a small sample is representative of their performance with respect to the whole distribution, and so that a small sample suffices to select a good mechanism in the class. See the related work section for a detailed discussion of these works and their natural limitations. Our goal in this work is to avoid a joint consideration of (i) and (ii). Rather we want to obtain a learning framework that separates Mechanism Design from Statistical Inference, based on the following:
\begin{enumerate}
\item[(i)'] find an algorithm ${\mathcal {M}}$, which given a distribution $F$ in some family of distributions $\mathcal{F}$, computes an (approximately) optimal mechanism ${\mathcal{M}}(F)$ when bidders' valuations are drawn from $F$;

\item[(ii)'] find an algorithm ${\cal L}$, which given sample access to a distribution $F$ from the family of distributions ${\cal F}$ learns a distribution ${\cal L}(F)$ that is close to $F$ in some distribution distance $d$.
\end{enumerate}
Achieving (i)' and (ii)' is of course not enough, unless we also guarantee the following:
\begin{enumerate}
\item[(iii)'] 
given an (approximately) optimal mechanism $M$ for some $F$ there is a way to transform $M$ to some $M'$ that is approximately optimal for \emph{any distribution} $F'$ that is close to $F$ under distribution~distance~$d$. 
\end{enumerate}
Given (i)'--(iii)', the learnability of (approximately) optimal mechanisms for a family of distributions ${\cal F}$ can be established as follows: {(a)} Given sample access to some distribution $F \in {\cal F}$ we use ${\cal L}$ to learn some distribution $F'$ that is close to $F$ under $d$;~ {(b)} we then use $\cal M$ to compute an (approximately) optimal mechanism $M'$ for $F'$; and {(c)} finally, we use (iii)' to argue that $M'$ can be converted to a mechanism $M$ that is (approximately) optimal for $F$ because $M$ is (approximately) optimal for any distribution that~is~close~to~$F'$. 

Clearly, (iii)' is important for decoupling (i)'---i.e.~computing (approximately) optimal mechanisms for a family of distributions ${\cal F}$, and (ii)'---i.e.~learning distributions in ${\cal F}$. At the same time, it is important in its own right:

\bigskip \noindent \hspace{0.7cm}\begin{minipage}{0.92\textwidth}
\begin{enumerate}
\item[{\bf Goal III:}] Develop robust mechanism design tools, allowing to transform a mechanism $M$ designed for some distribution $F$ into a mechanism $M_{\rm robust}$ which attains similar performance simultaneously for all distributions that are close to $F$ in some distribution distance of interest.
\end{enumerate}
\end{minipage}

\bigskip \noindent The reason Goal III is interesting in its own right is that oftentimes we actually have no sample access to the underlying distribution over valuations. It is possible that we estimate that distribution through market research or econometric analysis in related settings, so we only know some approximate distribution. In other settings, we may have sample access to the true distribution but there might be errors in measuring or recording those samples. In both cases, we would know some approximate distribution $F$ that is close to the true distribution under some distribution distance, and we would want to use $F$ to identify a good mechanism for~the~unknown distribution that is close to $F$. Clearly, outputting a mechanism $M$ that attains good performance under $F$ might be a terrible idea as this mechanism might very well overfit the details of $F$. So we need to ``robustify'' $M$. A similar goal was pursued in the work of Bergemann and Schlag~\cite{BergemannS11}, for single-item and single-bidder settings, and in the work of Cai and Daskalakis~\cite{CaiD17}, for robustifying a specific class of mechanisms under item-independence. Our goal here is to provide a very general robustification result.



\subsection{Our Results}  \label{sec:results}

We discuss our contributions in the setting of additive bidders, whose values for the items are not necessarily independent. Our results hold for quite more general valuations, including constrained additive and any family of Lipschitz valuations (Definition~\ref{def:Lipschitz}), but we do not discuss these here to avoid overloading our notation. We will denote by $n$ the number of bidders, and by $m$ the number of items. We will also assume that the bidders' values for the items lie in some bounded interval $[0,H]$.
 

\paragraph{Our Robustness Results (cf.~Goal III above).}  The setting we consider is the following. We are given a collection of model distributions $\DD=\{\DD_i\}_{i\in[n]}$, one for each bidder $i=1,\ldots,n$. We do not know the true distributions $\hDD=\{\hDD_i\}_i$ sampling the valuations of the bidders, and the only information we have about each $\hDD_i$ is that $d(\DD_i,\hDD_i)<\varepsilon$, under some distribution distance $d(\cdot,\cdot)$---we will discuss distances shortly.

Our goal is to design a mechanism 
that performs well under any possible collection of true  distributions $\{\hDD_i\}_{i}$ that are close to their corresponding distributions $\{\DD_i\}_i$ under $d$. 
We show that there are robustification algorithms, which transform a mechanism $M$ 
into a robust mechanism $\hM$ that attains similar revenue to that of $M$ under $\DD$, except that $\hM$'s revenue  guarantee holds {\em simultaneously} for any collection $\hDD$ that is close to $\DD$. Applying our robustification algorithm to the optimum mechanism for $\DD$ allows us to obtain the results reported in the first three columns of Table~\ref{table 1}. DSIC and BIC refer to the standard properties of Dominant Strategy and Bayesian Incentive Compatibility of mechanisms, IR refers to the standard notion of Individual Rationality, and $\eta$-BIC is the standard notion of approximate Bayesien Incentive Compatibility. For completeness these notions are reviewed in Appendix~\ref{sec:add_prelim}. 

Some remarks are in order. First, in multi-item settings, it is unavoidable that our robustified mechanism is only approximately BIC, as we do not know the true distributions. In single-item settings, the optimal mechanism is DSIC, and we can indeed robustify it into a mechanism $\hM$ that is  DSIC. In the multi-item case, however, it is known that DSIC mechanisms sometimes can extract at most a constant fraction of the optimal  revenue~\cite{Yao17a}, so  it is necessary to consider BIC mechanisms and the BIC property is fragile to errors in the distributions. 

Second, we consider several natural distribution distances. In multi-item settings, we consider both the Prokhorov  and the Total Variation distance. In single-item settings, we consider both the L\'evy  and the Kolmogorov distance. Please see Section~\ref{sec:prelim} for formal definitions of these distances and a discussion of their relationships, and their relationship to other standard distribution distances. We note that the L\'evy distance for single-dimensional distributions, and the Prokhorov distance for multi-dimensional distributions are quite permissive notions of distribution distance. This makes our robustness results for these distances stronger, automatically implying robustness results under several other common distribution distances.

Finally, en route to proving our robustness results, we show a result of independent interest, namely that {\em the optimal revenue is continuous with respect to the distribution distances that we consider}. Our continuity results are summarized in the last column of Table~\ref{table 1}. {Note that the continuity results are substantially easier to establish than the robustness results, please see Section~\ref{sec:roadmap} for details.}



\begin{table}[t]
	\centering
\resizebox{0.9\textwidth}{!}{\begin{minipage}{\textwidth}	\centering
		\hspace*{-8pt}\makebox[\linewidth][c]{
		\begin{tabular}{c || c | c| c| }
		\hline \hline
		Setting & Distance $d$ & Robustness& Continuity\\
\hline
\multirow{2}{*}{\rotatebox[origin=c]{0}{\parbox[c]{1.3cm}{\centering Single Item}}}
 & Kolmogrov& \thead{ $ \rev\left (\hM,\hDD \right ) \geq \opt\left (\hDD \right )-O\left(nH\varepsilon\right)$ \\ $\hM$ is IR and DSIC \\(Theorem~\ref{thm:single-kolmogorov})} & \thead{$ \left \lvert \opt \left(\hDD\right )-\opt\left (\DD\right ) \right \rvert \leq O(nH\varepsilon)$\\ (Theorem~\ref{thm:single-kolmogorov})} \\  \cline{2-4}

		&L\'evy& \thead{$ \rev\left (\hM,\hDD \right ) \geq \opt\left (\hDD \right )-O\left(nH\varepsilon\right)$ \\ $\hM$ is IR and DSIC  \\(Theorem~\ref{thm:single-levy})}  & \thead{$ \left | \opt \left(\hDD\right )-\opt\left (\DD\right ) \right | \leq O(nH\varepsilon)$\\ (Corollary~\ref{cor:single continuity})}\\  \cline{2-4}\hline\hline
\multirow{2}{*}{\rotatebox[origin=c]{0}{\parbox[c]{1.3cm}{\centering Multiple Items}}} 

& TV&\thead{$ \rev \left (\hM,\hDD \right ) \geq \opt_{\eta}\left (\hDD \right )-O\left(n^2mH\varepsilon+nmH\sqrt{n\varepsilon}\right)$\\ $\hM$ is IR and $\eta$- BIC w.r.t. $\hDD$, where $\eta = O\left( n^2mH\varepsilon \right)$ \\(Theorem~\ref{thm:approximation preserving})}  &  \thead{$ \left \lvert \opt \left(\hDD\right )-\opt\left (\DD\right ) \right \rvert \leq O\left(n^2mH\varepsilon+nmH\sqrt{n\varepsilon}\right)$\\ (Theorem~\ref{thm:continuous of OPT under Prokhorov})}\\   \cline{2-4}
		& Prokhorov& \thead{$\rev\left(\hM,\hDD\right)\geq \opt_{\eta}\left(\hDD\right)-O\left({ n\eta}+n\sqrt{{m H\eta}}\right)$\\
		$\hM$ is IR and $\eta$- BIC w.r.t. $\hDD$, where $\eta=O\left(n m  H \varepsilon+  m  \sqrt{nH\varepsilon}\right)$ \\
		(Theorem~\ref{thm:approximation preserving})} &  \thead{$\left \lvert\opt(\DD)-\opt(\hDD)\right\rvert\leq O\left(n\xi+n\sqrt{mH\xi}\right)$ \\
		where $\xi=O\left(n m  H \varepsilon+  m  \sqrt{nH\varepsilon}\right)$\\
		(Theorem~\ref{thm:continuous of OPT under Prokhorov})}
\\ \cline{2-4}
\hline
		\end{tabular}}
				\end{minipage}}
				\caption{\fontsize{7}{9}\selectfont Summary of Our Robustness and Revenue Continuity Results. Recall that the true bidder distributions $\hDD$ are unknown, and that $\hM$ is the robustified mechanism returned by our algorithm given an optimal mechanism $M$ for a collection of bidder distributions $\DD$ that are $\varepsilon$-close to $\hDD$ under distribution distance $d$. $\rev(\hM,\hDD)$ denotes the revenue of $\hM$ when the bidder distributions are $\hDD$. For a collection of bidder distributions $\cal F$, $\opt({\cal F})$ is the optimal revenue attainable by any BIC and IR mechanism under distributions $\cal F$, and $\opt_\eta(\FF)$ denotes the optimum revenue attainable by any $\eta$-BIC and IR mechanism under $\FF$. Not included in the table are approximation preserving robustification results under TV and Prokhorov closeness. We show that we can transform any $c$-approximation $M$ w.r.t. $\DD$ to a robust mechanism $\hM$, so that $\hM$ is almost a $c$-approximation w.r.t.~$\hDD$. The results included in the table are  corollaries of this more general result when $c=1$. See our theorem statements for the complete details. Moreover, if there is only a single bidder, we can strengthen our robustness results in multi-item settings so that $\hM$ is IC instead of $\eta$-IC (see Theorem~\ref{thm:single approximation preserving}). Our continuity results hold for any $\DD$ and $\hDD$ as long as $d(\DD_i,\hDD_i)\leq \varepsilon$ for each~bidder~$i$.}
				\label{table 1}
					\vspace{-.15in}
				\end{table}

\paragraph{Learning Multi-Item Auctions Under Item Dependence (cf.~Goal I above).}
In view of our robustness results, presented above, the challenge of learning near-optimal auctions given sample access to the bidders' valuation distributions, becomes a matter of estimating these distributions in the required distribution distance, depending on which robustification result we want to apply. 

When the item values are independent, learning bidders' type distributions in our desired distribution distances is immediate. So we easily recover the guarantees of the main theorem of~\cite{GonczarowskiW18}. These guarantees are summarized in the second row of Table~\ref{tab:sample-based results}, and are expanded upon in Theorem~\ref{thm:multi-item auction sample}.

But a main goal of our work (namely Goal I from earlier) is to push the learnability of auctions well beyond item-independence. As stated earlier, it is impossible to attain learnability from polynomially many samples for arbitrary joint distributions over item values so we consider the well-studied  frameworks of MRFs and Bayesnets. These frameworks are flexible and can model {\em any} distribution, but they have a tunable complexity parameter whose value controls the dependence structure. This  parameter is the maximum clique size of an MRF and  maximum in-degree of a Bayesnet. We will denote this complexity parameter $d$ in both cases. Recall that  we also used $d(\cdot,\cdot)$ to denote distribution distances. To disambiguate, whenever we study MRFs or Bayesnets, we make sure to use $d(\cdot,\cdot)$, {\em with  parentheses}, to denote distribution distances. {Note that a small value of the complexity parameter $d$ {\em does not mean} that the corresponding MRF or Bayesnet does not have correlations among every pair of item values. Many natural MRF structures, with $d=2$, and Bayesnet structures, with $d=1$, permit distributions where all the variables are correlated, and indeed any pair of variables remain correlated even after conditioning on the values of all the other variables. In Figure~\ref{fig:costas}, we show a simple such example where the values of a bidder on four items are sampled from a Naive Bayes Model, which is a very simple type of Bayesnet with $d=1$. While even small values of $d$  allow all pairs of variables to be correlated even conditioning on everything else, the complexity parameter $d$ forbids \emph{arbitrary} dependence structures. Indeed, this is the reason why MRFs and Bayesnets are so prevalent. They allow rich dependent structures but not arbitrary ones, unless their complexity parameter  $d$ is tuned up to its maximum possible value, i.e.~equal to the total number of variables, in which case they can express any dependence structure. In particular, a model of complexity $d$ can express arbitrary dependence on subsets of $d$ (for MRFs) or $d+1$ (for Bayesnets) variables, and it allows some dependence structures on larger subsets of variables depending on the graphical structure of the model.
\begin{figure}[htbp]
\begin{center}
\includegraphics[scale=0.36]{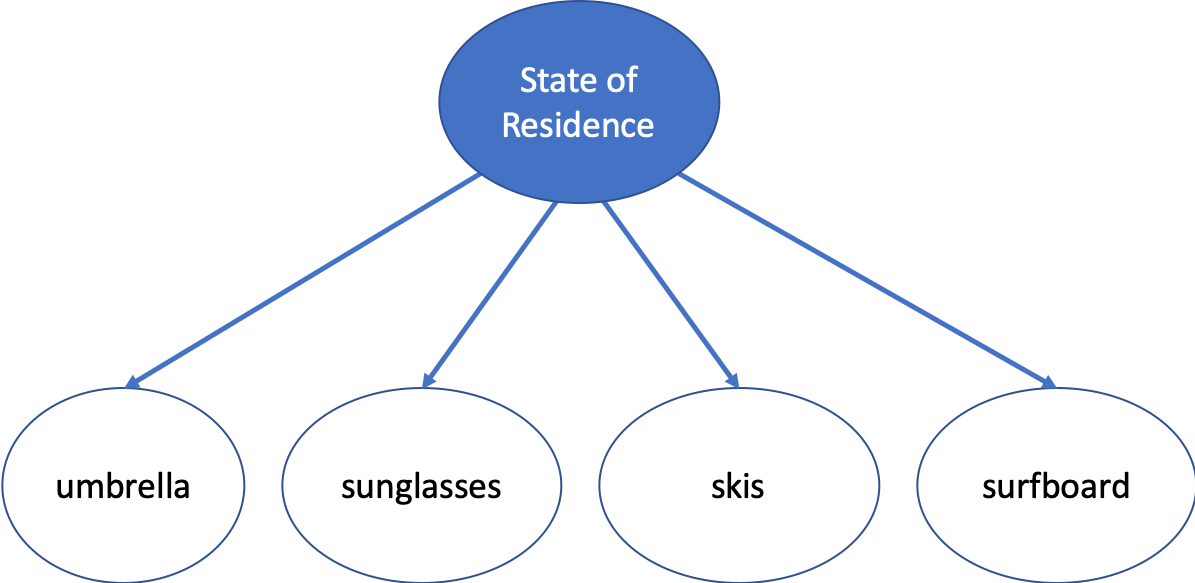}
\caption{\fontsize{7}{9}\selectfont The values of a buyer for an umbrella, a  pair of sunglasses, a pair of skis, and a surfboard are sampled from a Naive Bayes model. These values are sampled conditionally independently conditioning on the value of the variable at the root of the network, which is the state of the buyer's residence. This variable is latent, i.e.~non-observable, and this is why the corresponding node of the network is shaded blue. The distribution over $(v_{\rm umbrella}, v_{\rm sunglasses}, v_{\rm skis}, v_{\rm surfboard})$ has the property that any pair of values remain correlated even conditioning on all the other values, unless the conditional distributions in the Bayesnet have special structure.}
\label{fig:costas}
\end{center}
\vspace{-.2in}
\end{figure}

}

Now, in order to learn near-optimal mechanisms when item values for each bidder are sampled from an MRF or a Bayesnet of certain complexity $d$, our robustness results reassure us that it suffices to learn MRFs and Bayesnets under Total Variation or Prokhorov distance, depending on which multi-item robustenss theorem we seek to apply. So we need an upper bound on the sample complexity necessary to learn MRFs and Bayesnets. One of the contributions of our paper is to provide very general sample complexity bounds for learning these distributions, as summarized in Theorems~\ref{thm:sample complexity MRF finite alphabet} and~\ref{thm:sample complexity Bayesnets finite alphabet} for MRFs and Bayesnets respectively. In both theorems, $V$ is the set of variables, $d$ is the complexity measure of the underlying distribution, and~$\varepsilon$ is the distance within which we are seeking to learn the distribution. Each theorem has a version when the variables take values in a finite alphabet $\Sigma$, and a version when the variables take values in some interval $\Sigma=[0,H]$. In the first case, we provide bounds for learning in the stronger notion of Total Variation distance. In the second case, since we are learning from finitely many samples, we need to settle for the weaker notion of Prokhorov distance. For the same reason, we need to make some Lipschitzness assumption on the density, so our sample bounds depend on the Lipschitzness $\cal C$ of the MRF's potential functions and the Bayesnet's conditional distributions. 

The sample bounds we obtain for learning MRFs and Bayesnets are directly reflected in the sample bounds we obtain for learning multi-item auctions when the item-values are sampled from an MRF or a Bayesnet respectively, as summarized in the last two rows of Table~\ref{tab:sample-based results}. Indeed, the sample complexity for learning auctions is entirely due to the sample complexity needed to learn the underlying item-distribution. In all cases we consider, the complexity is polynomial in number of variables $n=|V|$ and only depends exponentially in $d$, the complexity of the distribution, and~this~is~unavoidable.~\footnote{Note that the example by Dughmi et al.~\cite{DughmiLN14} can be captured by an MRF or Bayesnet with $d=O(m)$, and it is shown in~\cite{DughmiLN14} that the sample complexity for learning a mechanism that is a constant factor approximation to the optimal revenue in this example is at least $2^{\Omega(m)}$.}



%
%

\begin{table}[h]
	\centering
\resizebox{0.9\textwidth}{!}{\begin{minipage}{\textwidth}	\centering
			\hspace*{-8pt}\makebox[\linewidth][c]{
	\begin{tabular}{c || c | c| c| }
			
			\hline\hline
			Setting & \ Revenue Guarantee and Sample Complexity  & Prior Result  & Technique \\
						\hline
												&&&\\

								\thead{Item \\Independence} &  \thead{up-to-$\varepsilon$ optimal, $\eta$-BIC~(Theorem~\ref{thm:multi-item auction sample})\\
									$\poly\left(n,m,H,1/\varepsilon, 1/\eta, \log (1/\delta)\right)$}  & \thead{recovers main\\ result of~\cite{GonczarowskiW18}} & \thead{Prokhorov Robustness + \\
									Learnability of Product Dist. (Folklore)}$$ \\
									&&&\\				\hline
									
									&&&\\

		\thead{MRF} &   \thead{up-to-$\varepsilon$ optimal, $\eta$-BIC (Theorem~\ref{thm:auction MRF}) \\
		$\poly\left(n,m^d,|\Sigma|^d, H, 1/\eta, 1/\varepsilon, \log (1/\delta)\right)$ (Finite $\Sigma$)\\ 
		$\poly\left(n,m^{d^2}, ({\CC H\over \varepsilon})^d, 1/\eta,\log (1/\delta)\right)$ ($\Sigma=[0,H]$)
		}& \thead{unknown} & \thead{Prokhorov Robustness +\\ Learnability of MRFs (Theorem~\ref{thm:sample complexity MRF finite alphabet}) }\\
			 
				\hline

						&&&\\

			\thead{Bayesnet} &   \thead{up-to-$\varepsilon$ optimal, $\eta$-BIC (Theorem~\ref{thm:auction Bayesnet}) \\
			 $\poly\left(n,d, m, |\Sigma|^{d+1}, H, 1/\eta, 1/\varepsilon, \log (1/\delta) \right)$ (Finite $\Sigma$)\\
			 $\poly\left( n, d^{d+1}, m^{d+1}, ({\CC H\over \varepsilon})^{d+1}, 1/\eta, \log (1/\delta)\right)$ ($\Sigma=[0,H]$)
			 } & \thead{unknown} & \thead{Prokhorov Robustness +\\
			Learnability of Bayesnets (Theorem~\ref{thm:sample complexity Bayesnets finite alphabet})} \\
			 				\hline
			 		\end{tabular}}
						\end{minipage}}

		
	\caption{\fontsize{7}{8}\selectfont Summary of Our Sample-based Results. We denote by $\Sigma$ the support of each item-marginal, taken to equal the interval $[0,H]$ in the continuous case, we use $\delta$ for the failure probability, and use $d$ to denote the standard complexity measure of the graphical model used to model item dependence, namely the size of the maximum hyperedge in MRFs and the largest in-degree in Bayesnets. For both MRFs and Bayesnets we allow latent variables and we also do not need to know the underlying graphical structure. Moreover, for continuous distributions, our results require Lipschitzness of potential functions in MRFs and conditional distributions in Bayesnets, which we denote with $\CC$. Finally, if there is only a single bidder, the mechanism we learnt is strengthened to be IC instead of $\eta$-IC. See our theorem statements for~our~complete~results.}
	\label{tab:sample-based results}
\end{table}

Our sample bounds improve if the underlying graph of the MRF or Bayesnet are known and, importantly, without any essential modifications {\em our sample bounds hold even when there are latent, i.e.~unobserved, variables in the distribution}. This makes both our auction and our distribution learning results much more richly applicable. As a simple example of the modeling power of latent variables, situations can be captured where an unobserved random variable determines the type of a bidder, and conditioning on this type the observable values of the bidder for~different~items~are~sampled.  

Finally, it is worth noting that our sample bounds for learning MRFs (i.e.~Theorem~\ref{thm:sample complexity MRF finite alphabet}) provide broad generalizations of the bounds for learning Ising models and Gaussian MRFs presented in recent work of Devroye et al~\cite{devroye2018minimax}. Their bounds are obtained by bounding the VC complexity of the Yatracos class induced by the distributions of interest, while our bounds are obtained by constructing $\varepsilon$-nets of the distributions of interest, and running a tournament-style hypothesis selection algorithm~\cite{devroye2012combinatorial,DaskalakisK14,AcharyaJOS14} to select one distribution from the net. Since the distribution families we consider are non-parametric, our main technical contribution is to bound the size of an $\varepsilon$-net sufficient to cover the distributions of interest. Interestingly, we use properties of linear programs to argue through a sequence of transformations that the net size can be upper bounded in terms of the bit complexity of solutions to a linear program that we construct.

\subsection{Roadmap and Technical Ideas}\label{sec:roadmap}
In this section, we provide a roadmap to the paper and survey some of our technical ideas.
 
\paragraph{Single-item Robustness (Appendix~\ref{sec:single})} We consider first the setting where the model distribution $\DD$ is $\varepsilon$-close to the true, but unknown distribution $\hDD$ in {Kolmogorov distance}. In this case, we argue directly that Myerson's optimal mechanism~\cite{Myerson81} for $\DD$ is approximately optimal for any distribution that is in the $\varepsilon$-Kolmogorov-ball around $\DD$, which includes $\hDD$ (Theorem~\ref{thm:single-kolmogorov}). The idea is that the revenue of the optimal mechanism can be written as an  integral over probabilities of events of the form: does $v_i$ lie in a certain interval $[a,b]$? Since $\DD$ and $\hDD$ are $\varepsilon$-close in  Kolmogorov distance, the probabilities of all such events are within $\varepsilon$ of each other, which implies that the revenues under $\DD$ and $\hDD$ are also close. Finally, note that Myerson's optimal mechanism is DSIC and IR, so it is truthful and IR w.r.t.~any distribution.

Unfortunately, the same idea fails for L\'evy distance, as the difference in the probabilities of the event that a certain $v_i$ lies in some interval $[a,b]$ under $\DD$ and $\hDD$ can be as large as $1$ even when $\DD$ and $\hDD$ are $\varepsilon$-close in  L\'evy distance. (Indeed, consider two single point distributions: a point mass at $A$ and a point mass at $A-\varepsilon$; their probabilities of falling in the interval $[A-\varepsilon/2,A+\varepsilon/2]$ are respectively $1$ and $0$.)  We thus prove our robustness result for  L\'evy distance via a different route. Given any model distribution $\DD$, we first construct the ``worst'' distribution $\lDD$ and the ``best'' distribution $\uDD$ %
 in the $\varepsilon$-L\'evy ball around $\DD$: this means that, for any  $\hDD$ that lies in the $\varepsilon$-L\'evy ball around $\DD$, $\hDD$ first-order stochastically dominates $\lDD$ and is dominated by $\uDD$ (see Definition~\ref{def:best-worst-distributions}). We choose our robust mechanism $\hM$ to be  Myerson's optimal mechanism for $\lDD$. It is not hard to argue that $\hM$'s revenue under $\hDD$ is at least $\opt(\lDD)$, the optimal revenue under the ``worst'' distribution (Lemma~\ref{lem:LB for the revenue of our mechanism}), due to the revenue monotonicity lemma (Lemma~\ref{lem:rev-monotonicity}) shown in~\cite{DevanurHP16}. The statement provides a lower bound of $\hM$'s revenue under the unknown true distribution $\hDD$. To complete the argument, we need to argue that $\opt(\hDD)$ cannot be too much larger than $\opt(\lDD)$. Indeed, we relax $\opt(\hDD)$ to $\opt(\uDD)$, and show that even the optimal revenue under the ``best'' distribution $\opt(\uDD)\approx \opt(\lDD)$. To do so, we construct two auxiliary distributions $P$ and $Q$, such that (i) $\opt(P)\approx \opt(Q)$; and (ii) $P$ and $\lDD$ are $\varepsilon$-close in  {\em Kolmogorov} distance, and $Q$ and $\uDD$ are $\varepsilon$-close also in {\em Kolmogorov} distance. Our robustness theorem under Kolmogorov distance (Theorem~\ref{thm:single-kolmogorov}) implies then that $\opt(P)\approx\opt(\lDD)$ and  $\opt(Q)\approx\opt(\uDD)$. Hence,  $\opt(\uDD)\approx\opt(\lDD)$, which completes our proof.
 \paragraph{Multi-item Robustness  (Section~\ref{sec:prokhorov multi})}
 We first discuss our result for total variation distance. Unfortunately, our approach for L\'evy distance---of simply choosing the optimal mechanism for the ``worst,'' in the first-order stochastic dominance sense, distribution in the $\varepsilon$-TV-ball around $\DD$ to be our robust mechanism---no longer applies. Indeed, it is known that the optimal revenue in multi-item auctions may be non-monotone with respect to first-order stochastic dominance~\cite{HartR12}, i.e.~a distribution may be stochastically dominated by another but result in higher revenue.
%
However, if $\DD$ and $\hDD$ are $\varepsilon$-close in total variation distance, this means that there is a coupling between $\DD$ and $\hDD$ under which the valuation profiles are almost always sampled the same. If we take the optimal mechanism $M$ for $\DD$, and apply to bidders from $\hDD$, it will  produce almost the same revenue under $\hDD$, and vice versa. Indeed, the only event under which $M$ may generate different revenue under the two distributions is when the coupling samples different profiles, but this happens with small probability. Similarly, the BIC and IR properties of $M$ under $\DD$ become slightly approximate under~$\hDD$. We claim that we can massage $M$, in a way \emph{oblivious to $\hDD$}, to produce a $\left(\poly(n,m,H)\cdot\varepsilon\right)$-truthful and {\em exactly} IR mechanism $\hM$ for $\hDD$, which achieves an up-to-$\left(\poly(n,m,H)\cdot\varepsilon\right)$ revenue (Theorem~\ref{lem:transformation under TV}).

 The main challenge is when $\DD$ and $\hDD$ are only $\varepsilon$-close in Prokhorov distance. Note that two distributions within Prokhorov distance $\varepsilon$ may have total variation distance $1$. Just imagine two point masses: one at $A$ and another at $A-\varepsilon$. So   Prokhorov robustness~is not directly implied by~TV~robustness. 
 
 \paragraph{\textbf{Why Standard Discretization Arguments are Insufficient?}} Unlike standard algorithmic problems, discretization is subtle in mechanism design. Due to the presence of incentives, a small change in the bidders' value distributions may change the distribution of outcomes of the mechanism dramatically. To perform discretization in mechanism design, a standard procedure goes as follows~\cite{BabaioffGN17, GonczarowskiW18, Kothari0MSW19}: let $\hDD$ be the true distribution, and $\DD$ be the distribution after discretization; design the optimal mechanism $M$ for $\DD$; to run $M$ on a bid vector $b$ from $\hDD$, discretize it to $\gamma(b)=({\gamma}_1(b_1),\ldots, {\gamma}_n(b_n))$ and apply mechanism $M$ on $\gamma(b)$. This procedure can be generalized to any pair of distributions $\DD$ and $\hDD$ as long as, we are given a coupling $\gamma(\cdot)$ between $\DD$ and $\hDD$ that maps any bid vector $b$ in the support of distribution $\hDD$ to a bit vector $\gamma(b)$ in the support of $\DD$. If for every bidder $i$, $b_i$ and ${\gamma}_i(b_i)$ are close with all but small probability, we can apply similar arguments as in the total variation robustness result to massage the mechanism above to be nearly-truthful and exactly IR for $\hDD$, and argue it is approximately revenue optimal. Clearly, in the context of discretization, $b_i$ and $\gamma_i(b_i)$ are guaranteed to be close if the discretization is sufficiently fine.


At first glance, this procedure may seem applicable to our problem. A characterization of Prokhorov distance due to Strassen (Theorem~\ref{thm:prokhorov characterization}) shows that: two distributions $P$ and $Q$ are $\varepsilon$-close in Prokhorov distance if and only if there exists a (potentially randomized) coupling $\gamma$ such that if random variable $s$ is distributed according to $P$, then $\gamma(s)$ is distributed according to $Q$ and $\Pr\left[ \norm{s-\gamma(s)}_{1}> \varepsilon \right]\leq \varepsilon$. If $M$ is the optimal mechanism for the model distribution $\DD$, and $\hDD$ is the true distribution that is $\varepsilon$-close to $\DD$, why can't we combine the procedure above with the coupling $\gamma$ to establish our Prokhorov robustness result?

Unfortunately, this approach is insufficient due to the following two issues: (i) The procedure relies on knowing the coupling  ${\gamma}$. As we do not even know $\hDD$, how can we know the coupling? 
(ii) Even if we can identify the coupling $\gamma$ between $\DD$ and a specific $\hDD$, the procedure above constructs a mechanism that depends on the coupling $\gamma$. However, ${\gamma}$ may change for every different $\hDD$ in the $\varepsilon$-Prokhorov-ball around $\DD$, so the procedure generates a different mechanism for every possible true distribution.~\footnote{It is worth noting that the procedure can indeed be employed to prove the Prokhorov continuity, as the the pure existence of a good coupling $\gamma$ between $\DD$ and $\hDD$ suffices.}
 
 To satisfy our requirement for a robust mechanism in Goal III, we need to construct a \emph
{single mechanism} that is nearly truthful, IR, and near-optimal simultaneously for every distribution in the $\varepsilon$-Prokhorov-ball around $\DD$. Our proof relies on a novel way to \emph{``simultaneously couple''} $\DD$ with every distribution $\hDD$ in the $\varepsilon$-Prokhorov-ball around $\DD$. If we round both $\DD$ and any $\hDD$ to a random grid $G$ with width $\sqrt{\varepsilon}$, we can argue that the \emph{expected total variation distance} (over the randomness of the grid) between the two rounded distributions $\DD_G$ and $\hDD_G$ is $O(\sqrt{\varepsilon})$ (Lemma~\ref{lem:prokhorov to TV}). Now consider the following mechanism: choose a random grid $G$, round the bids to the random grid, apply the optimal mechanism $M_G$ that is designed for $\DD_G$. Our robustness result under the total variation distance implies that for every realization of the random grid $G$, $M_G$ is $O\left(\poly(n,m,H)\cdot\norm{\DD_G-\hDD_G}_{TV}\right)$-truthful and up-to-$O\left(\poly(n,m,H)\cdot\norm{\DD_G-\hDD_G}_{TV}\right)$ revenue optimal for any $\hDD_G$. 
 Since the expected value (over the randomness of the grid) of $\norm{\DD_G-\hDD_G}_{TV}$ is $O(\sqrt{\varepsilon})$ for any $\hDD$ in the $\varepsilon$-Prokhorov-ball of $\DD$, our randomized mechanism is simultaneously $O\left(\poly(n,m,H)\cdot\sqrt{\varepsilon}\right)$-truthful and up-to-$O\left(\poly(n,m,H)\cdot\sqrt{\varepsilon}\right)$ revenue optimal for all distributions in the $\varepsilon$-Prokhorov-ball around $\DD$.\footnote{Since we round the bids to a random grid, we will also need to accommodate the rounding error. Please see Theorem~\ref{thm:multi-item p-robustness} for details.}

 \paragraph{Sample Complexity Results}In Section~\ref{sec:old sampling}, we apply our robustness theorem to obtain sample bounds for learning multi-item auctions under the item-independence assumption (Theorem~\ref{thm:multi-item auction sample}). Our result provides an alternative proof of the main result of~\cite{GonczarowskiW18}. In Section~\ref{sec:new sampling}, we combine our robustness theorem with our sample bounds for learning Markov Random Fields and Bayesian Networks discussed earlier to derive new polynomial sample complexity results for learning multi-item auctions when the distributions have structured correlation over the items. Theorem~\ref{thm:auction MRF} summarizes our results when  item values are generated by an MRF, and Theorem~\ref{thm:auction Bayesnet} our results when  item values are generated by a Bayesenet.
\section{Preliminaries}\label{sec:prelim}
We first define a series statistical distances that we will use in the paper and discuss their relationships.
\begin{definition}[Statistical Distance]\label{def:statistical distances}
	Let $P$ and $Q$ be two probability measures.
We use $\norm{P-Q}_{TV}$, $\norm{P-Q}_K$, and $\norm{P-Q}_L$ to denote the \textbf{total variational distance}, the \textbf{Kolmogorov distance}, and the \textbf{L\'evy distance} between $P$ and $Q$, respectively. See Appendix~\ref{sec:add_prelim} for more details. \textbf{Prokhorov Distance} is a generalization of the L\'evy Distance to high dimensional distributions. Let $(U,d)$ be a metric space and $\cal{B}$ be a $\sigma$-algebra on $U$. For $A \in \cal{B}$, let $A^{\varepsilon} = \{x : \exists y \in A \ \  s.t \ \ d(x,y)<\varepsilon \}$.
Then two measures $P$ and $Q$ on $\cal{B}$ have Prokhorov distance 
$$ \inf \left \{\varepsilon>0 : P(A) \leq Q(A^{\varepsilon}) + \varepsilon, \ Q(A)\leq P(A^{\varepsilon})+\varepsilon  \ \forall A \in \cal{B}\right \}$$
We consider distributions supported on $\mathbb{R}^k$ for some $k \in \mathbb{N}$, so $U$ will be the $k$-dimensional Euclidean Space, and we choose \textbf{$d$ to be the $\ell_1$-distance}. We denote the Prokhorov distance between distributions $\FF$, $\widehat{\FF}$ by $\norm{\FF-\widehat{\FF}}_P$.
\end{definition}

\vspace{-.1in}
\paragraph{Relationships between the Statistical Distances:} Among the four metrics, the L\'evy distance and the Kolmogorov distance are only defined for single dimensional distributions, while the Prokhorov distance and the total variation distance are defined for general distributions. In the single dimensional case, the L\'evy distance is a very liberal metric. In particular, for any two single dimensional distributions $P$ and $Q$,  $$\norm{P-Q}_L\leq \norm{P-Q}_K\leq \norm{P-Q}_{TV}.$$ Note that a robustness result for a more liberal metric is more general. For example, the robustness result for single-item auctions under the L\'evy metric implies the robustness under the total variation and Kolmogorov metric, because the $\varepsilon$-ball in L\'evy distance contains the $\varepsilon$-ball in total variation and Kolmogorov distance. An astute reader may wonder whether one can find a more liberal metric in the single dimensional case. Interestingly, for the most common metrics studied probability theory, including the Wasserstein distance, the Hellinger distance, and the relative entropy, the L\'evy distance is the most liberal up to a polynomial factor. That is, if the  L\'evy distance is $\varepsilon$, the distance under any of these metrics is at least $\poly(\varepsilon)$. Indeed, the polynomial is simply the identity function or the quadratic function $\varepsilon^2$ in most cases. Please see the survey by Gibbs and Su~\cite{GibbsS02} and the references therein for more details.

The Prokhorov distance, also known as L\'evy-Prokhorov Distance, is the generalization of the L\'evy distance to multi-dimensional distributions. It is also the standard metric in robust statistical decision theory, see Huber~\cite{Huber11} and Hampel et al.~\cite{HampelRRS11}. The Prokhorov distance is almost as liberal as the L\'evy distance.~\footnote{Note that for single dimensional distributions, the Prokhorov distance is not equivalent to L\'evy distance. In particular, $\norm{P-Q}_L\leq \norm{P-Q}_P$ for any single dimensional distributions $P$ and $Q$.}  First, for any two distributions $P$ and $Q$, $$\norm{P-Q}_P\leq \norm{P-Q}_{TV}.$$ Second, if we consider other well studied metrics such as the Wasserstein distance, the Hellinger distance, and the relative entropy, the Prokhorov distance is again the most liberal up to a polynomial factor. 

\vspace{-.1in}
\paragraph{Multi-item Auctions:} We focus on revenue maximization in the combinatorial auction with \textbf{$n$ bidders} and \textbf{$m$ heterogenous items}. We use $X$ to denote the set of possible allocations, and each bidder $i\in [n]$ has a valuation function/type  $v_i(\cdot): X \mapsto \mathbb{R}_{\geq 0}$. In this paper, we assume the function $v_i(\cdot)$ is parametrized by $(v_{i,1},\ldots,v_{i,m})$, where $v_{i,j}$ is bidder $i$'s value for item $j$. We assume that \emph{bidder's types are distributed independently}. Throughout this paper, we assume all bidders types lie in $[0,H]^m$.  We adopt the valuation model in Gonczarowski and Weinberg~\cite{GonczarowskiW18} and consider valuations  that satisfy the following Lipschitz property. 

\begin{definition}[Lipschitz Valuations]\label{def:Lipschitz}
    There exists an absolute constant $\L$ such that if type $\bold{v_i}=(v_{i,1},\ldots,v_{i,m})$ and $\bold{v'_i}=(v'_{i,1},\ldots,v'_{i,m})$ are within $\ell_1$ distance $\varepsilon$, then for the corresponding valuations $v_i(\cdot)$ and $v'_i(\cdot)$, $|v_i(x)-v'_i(x)| \leq \L\cdot \varepsilon$ for all $x \in X$. 
\end{definition}
This for example includes common settings such as additive and unit demand with Lipschitz constant $\L=1$. More generally, $\L=1$ holds for constrained additive valuations~\footnote{$v_i(\cdot)$ is constrained additive if  $v_i(X)=\max_{R\subseteq S, R\in \mathcal{I}}\sum_{j\in R} v_{i,j}$, for some downward closed set system ${\cal I} \subseteq 2^{[m]}$~and~$S=\{j: x_{i,j}=1\}$.} and even in some settings with complementarities. Please see~\cite{GonczarowskiW18} for further discussion. 

A mechanism $M$ consists of an allocation rule $x(\cdot)$ and a payment rule $p(\cdot)$. For any input bids $b=(b_1,\ldots,b_n)$, the allocation rule outputs a distribution over allocations $x(b)\in \Delta(X)$ and payments $p(b) = \left(p_1(b),\ldots, p_n(b) \right)$. 
If bidder $i$'s type is $v_i$, her utility under input $b$ is $u_i\left(v_i, M(b)\right)=\E\left[v_i\left(x(b)\right)-p_i(b)\right]$. 

\vspace{-.05in}
\paragraph{Truthfulness and Revenue:}
We use the standard notion $\varepsilon$-BIC and IR (see Appendix~\ref{sec:add_prelim} for details). If $M$ is a $\varepsilon$-BIC mechanism w.r.t. some distribution $\DD$, we use $\revT(M,\DD)$ to denote the revenue of mechanism $M$ under distribution $\DD$ assuming bidders are bidding truthfully. Clearly, $\revT(M,\DD)=\rev(M,\DD)$ when $M$ is BIC w.r.t. $\DD$. We denote the optimal revenue achievable by any $\varepsilon$-BIC (or BIC) mechanism by $\text{OPT}_{\varepsilon}(\DD)$ (or $\opt(\DD)$). Although it is conceivable that permitting mechanisms to be $\varepsilon$-BIC allows for much greater expected revenue than if they were restricted to be BIC, past results show that this is not the case.


\begin{lemma}~\cite{DaskalakisW12,RubinsteinW15}\label{lem:eps-BIC to BIC}
In any $n$-bidder $m$-item auction, let $\DD$ be any joint distribution over arbitrary $\LL$-Lipschitz valuations, where the valuations of different bidders are independent. The maximum revenue attainable by any IR and $\varepsilon$-BIC auction for a given product distribution is at most $2n\sqrt{m\L H\varepsilon}$ greater than the maximum revenue attainable by any IR and BIC auction for that distribution.
\end{lemma}


%

\vspace{-.15in}
\paragraph{Notations:} We allow the bidders to submit a special type $\perp$, which represents not participating the auction. If anyone submits $\perp$, the mechanism {terminates immediately, and does not allocate any item to any bidder or charge any bidder.} 
A bidder's utility for submitting type $\perp$ is $0$. We will sometimes refer to $\perp$ as the \textbf{IR type}.Throughout the paper, we use $\hDD=\bigtimes_{i=1}^n \hDD_i$ to denote the true type distributions of the bidders. We use $\DD=\bigtimes_{i=1}^n \DD_i$ to denote the model type distributions or our learned type distributions from samples. We use $\DD_i \mid \bigtimes_{j=1}^m [w_{ij},w_{ij}+\delta)$ to denote the distribution induced by $\DD_i$ conditioned on being in the $m$-dimensional cube $\bigtimes_{j=1}^m [w_{ij},w_{ij}+\delta)$, and \textbf{$\supp(\FF)$} to denote the support of distribution $\FF$.

\notshow{
\paragraph{Support of Mechanisms.}
At times we will want to run mechanisms on bidders whose values the mechanism potentially does not support because it is not defined on all of $[0,H]^{n\cdot m}$. In these cases we would still like to guarantee every bidder the option of essentially not participating in the auction and hence not having to make any payment. Denote this type by $\perp$ i.e a bidder of this type draws a value from their distribution and submits it to the mechanism, but instead of getting the usual allocation, always receives no items and makes no payment.
Also, when a distribution is only supported on a finite set because it's values were previously rounded down, we may recover the original distribution while insuring we stay close to the sample of the rounded distribution. For some bidder $i\in [n]$, some small $\delta>0$ and distribution $\DD_i$, let $\round{\DD_i}$ be a distribution that was rounded down with respect to some grid of size $\delta$. Let $w_i \sim \round{\DD_i}$. Then  $\DD_i \mid \bigtimes_{j=1}^m [w_{ij},w_{ij}+\delta)$ is the distribution $\DD_i$ conditioned on being in $\bigtimes_{j=1}^m [w_{ij},w_{ij}+\delta)$.
For such rounded distributions and also others which are not supported on $[0,H]^{n\cdot m}$, we will denote by supp$(\cdot)$ the support of a distribution.
}




\section{L\'evy-Robustness for Single-Item Auctions}\label{sec:single}


In this section, we show the robustness result under  the L\'evy distance in the single-item setting.  If we are given a model distribution $\DD_i$ that is $\varepsilon$-close to the true distribution $\hDD_i$, in L\'{e}vy distance, for every bidder $i\in [n]$, we show how to design a mechanism $M^*$ only based on $\DD=\bigtimes_{i=1}^n \DD_i $ and extracts revenue that is at most $O(nH\cdot \varepsilon)$ less than the optimal revenue under any possible true distribution $\hDD = \bigtimes_{i=1}^n \hDD_i$.

\begin{theorem}[L\'evy-Robustness for Single-item Auctions]\label{thm:single-levy}
	Given $\DD=\bigtimes_{i=1}^n \DD_i $, where $\DD_i$ is an arbitrary distributions supported on $[0,H]$ for all $i\in[n]$. We can design a DSIC and IR mechanism $M^*$ based on $\DD$ such that for any product distribution $\hDD = \bigtimes_{i=1}^n \hDD_i$ satisfying $\norm{\DD_i-\hDD_i}_L\leq \varepsilon$ for all $i\in[n]$, we have:	\begin{equation*}
		\rev(M^*,\hDD)\geq \opt(\hDD)-O(nH\cdot \varepsilon).
	\end{equation*}
\end{theorem}

Let us sketch the proof of Theorem~\ref{thm:single-levy}. We prove our statement in three steps.\begin{itemize}
\item	 \textbf{Step (i):} We first identify the ``best'' and ``worst'' distributions (Definition~\ref{def:best-worst-distributions}), in terms of the first-order stochastic dominance (Definition~\ref{def:stochastic dominance}), among all distributions in the $\varepsilon$-L\'evy-ball around the model distribution $\DD$. We construct the optimal mechanism $M^*$ w.r.t. the ``worst'' distribution, and show that its revenue under any possible true distribution is at least $M^*$'s revenue under the ``worst'' distribution (Lemma~\ref{lem:LB for the revenue of our mechanism}). The statement provides a lower bound of $M^*$'s revenue under the unknown true distribution. Its proof follows from the revenue monotonicity lemma (Lemma~\ref{lem:rev-monotonicity}) shown in~\cite{DevanurHP16}.  
\item \textbf{Step (ii):} We use the revenue monotonicity lemma again to show the optimal revenue under the true distribution $\hDD$ is upper bounded by the optimal revenue under the ``best'' distribution(Lemma~\ref{lem:UB for the true revenue}).
\item \textbf{Step (iii):} We complete the proof by argueing that $M^*$'s  revenue under the ``worst'' distribution can be at most $O(nH\cdot \varepsilon)$ worst than the optimal revenue under the ``best'' distribution (Lemma~\ref{lem:LB and UB are close}). The statement follows from a robustness theorem for single-item auctions under the Kolmogorov distance (Theorem~\ref{thm:single-kolmogorov}).
\end{itemize}
 
 We show Step (i) and (ii) in Section~\ref{sec:best and worst distributions Levy} and Step (iii) in Section~\ref{sec:bounding the gap between best and worst Levy}.
 \subsection{Best and Worst Distributions in the $\varepsilon$-L\'evy-Ball}\label{sec:best and worst distributions Levy}
We formally define the  ``best'' and ``worst'' distributions in the $\varepsilon$-L\'evy-ball around the model distribution. 
	\begin{definition}\label{def:best-worst-distributions}
		For every $i\in [n]$, we define $\uDD_i$ and $\lDD_i$ based on $\DD_i$. $\uDD_i$ is supported on $[0,H+\varepsilon]$, and its CDF is defined as $F_{\uDD_i}(x) = \max\left\{F_{\DD_i}(x-\varepsilon)-\varepsilon,0\right\}$. $\lDD_i$ is supported on $[-\varepsilon,H]$, and its CDF is defined as $F_{\lDD_i}(x) = \min\left\{F_{\DD_i}(x+\varepsilon)+\varepsilon,1\right\}$.
	\end{definition}

We provide a more intuitive interpretation of $\uDD_i$ and $\lDD_i$ here. To obtain $\uDD_i$, we first shift all values in $\DD_i$ to the right by $\varepsilon$, then we move the bottom $\varepsilon$ probability mass to $H+\varepsilon$. To obtain $\lDD_i$, we first shift all values in $\DD_i$ to the left by $\varepsilon$, then we move the top $\varepsilon$ probability mass to $-\varepsilon$. It is not hard to see that both $\uDD_i$ and $\lDD_i$ are still in the $\varepsilon$-ball around $\DD_i$ in L\'{e}vy distance. More importantly, $\uDD_i$ and $\lDD_i$ are the ``best'' and ``worst'' distributions in the $\varepsilon$-L\'evy-ball under first-order-stochastic-dominance. 

\begin{definition}[First-Order Stochastic Dominance]\label{def:stochastic dominance}
	We say distribution $B$ \emph{first-order stochastically dominates} $A$ iff $F_B(x)\leq F_A(x)$ for all $x\in \mathbb{R}$. We use $A\preccurlyeq B$ to denote that distribution $B$ first-order stochastically dominates distribution $A$. If $\mathbf{A}=\times_{i=1}^n A_i$ and $\mathbf{B}=\times_{i=1}^n B_i$ are two product distributions, and $A_i\preccurlyeq B_i$ for all $i\in[n]$, we slightly abuse the notation $\preccurlyeq$ to write $\mathbf{A}\preccurlyeq \mathbf{B}$.
\end{definition}

\begin{lemma}\label{lem:best-worst-stochastic-dominance}
	For any $\hDD_i$, such that $\norm{\hDD_i-\DD_i}_L\leq \varepsilon$, we have $\lDD_i \preccurlyeq\hDD_i\preccurlyeq \uDD_i$. \end{lemma}
\begin{proof}
	It follows from the definition of L\'{e}vy distance and Definition~\ref{def:best-worst-distributions}. For any $x$, \begin{equation*}
		F_{\hDD_i}(x) \in [F_{\DD_i}(x-\varepsilon)-\varepsilon, F_{\DD_i}(x+\varepsilon)+\varepsilon].
	\end{equation*}
	Clearly, $0\leq F_{\hDD_i}(x) \leq 1$, so we have  $	F_{\uDD_i}(x) \leq F_{\hDD_i}(x) \leq F_{\lDD_i}(x)$
for all $x$.
\end{proof}

The plan is to construct the optimal mechanism for $\lDD=\bigtimes_{i=1}^n \lDD_i$ and show that this mechanism achieves up-to-$O(nH\cdot\varepsilon)$ optimal revenue under any possible true distribution $\DD$.

Next, we state a revenue monotonicity lemma that will be useful. We first need the following definition. \begin{definition}[Extension of a Mechanism to All Values]\label{def:mechanism-extension}
 Suppose a mechanism $M=(x,p)$ is defined for all value profiles in $T=\times_{i=1}^n T_i$. Define its extension $M'=(x',p')$ to all values. We only specify $x'$, as $p'$ can be determined by the payment identity given $x'$. $x'$ first rounds the bid of each bidder $i$ down to the closest value in $T_i$, and then apply allocation rule $x$ on the rounded bids. If some bidder $i$'s bid is smaller than the lowest value in $T_i$, $x'$ does not allocate the item to any bidder.
\end{definition}
Observe that the extension provides a DSIC and IR mechanism for all values if the original mechanism is DSIC and IR.

\begin{lemma}[Strong Revenue Monotonicity~\cite{DevanurHP16}]\label{lem:rev-monotonicity}
Let $\FF=\bigtimes_{i=1}^{n} \FF_i$ be a product distributions. There exists an optimal DSIC and IR mechanism
$M$ for $\FF$ such that, for any product distribution $\FF'=\bigtimes_{i=1}^{n} \FF'_i \succcurlyeq \FF$, 
\begin{equation*}
	\rev(M',\FF')\geq \rev(M,\FF) = \opt(\FF).
\end{equation*}
$M'$ is  the extension of $M$. In particular, this implies $\opt(\FF')\geq \opt(\FF)$.
\end{lemma}

Combining Lemma~\ref{lem:best-worst-stochastic-dominance} and~\ref{lem:rev-monotonicity}, we show that if $M^*$ is the extension of the optimal mechanism for  $\lDD$, it achieves at least $\opt(\lDD)$ under any distribution $\hDD$ with $\norm{\hDD_i-\DD_i}_L\leq \varepsilon$.
\begin{lemma}\label{lem:LB for the revenue of our mechanism}
	Let $M^*$ be the extension of the optimal DSIC and IR mechanism for $\lDD$. For any product distribution $\hDD=\bigtimes_{i=1}^n \hDD_i$ with $\norm{\hDD_i-\DD_i}_L\leq \varepsilon$ for all $i\in[n]$, we have the following:
	\begin{equation*}
	\rev(M^*,\hDD)\geq \opt(\lDD).
\end{equation*}
\end{lemma}
\begin{proof}
	Since  $\hDD\succcurlyeq \lDD$ (Lemma~\ref{lem:best-worst-stochastic-dominance}), the claim follows from Lemma~\ref{lem:rev-monotonicity}.
\end{proof}

Lemma~\ref{lem:LB for the revenue of our mechanism} shows that with only knowledge of the model distribution $\DD$, we can design a mechanism whose revenue under any possible true distribution $\hDD$ is at least $\opt(\lDD)$. Next, we upper bound the optimal revenue under $\hDD$ with the optimal revenue under $\uDD$.
\begin{lemma}\label{lem:UB for the true revenue}
	For any product distribution $\hDD$ with $\norm{\hDD_i-\DD_i}_L\leq \varepsilon$ for all $i\in[n]$, we have the following:
	\begin{equation*}
	\opt(\uDD)\geq \opt(\hDD).
\end{equation*}
\end{lemma}
\begin{proof}
		Since  $\uDD\succcurlyeq \hDD$ (Lemma~\ref{lem:best-worst-stochastic-dominance}), the claim follows from Lemma~\ref{lem:rev-monotonicity}.
\end{proof}

\subsection{Comparing the Revenue of the Best and Worst Distributions}\label{sec:bounding the gap between best and worst Levy}


In this section, we show that our lower bound of $M^*$'s revenue under the true distribution $\hDD$ and our upper bound of the optimal revenue under $\hDD$ are at most $O(nH\cdot \varepsilon)$ away.

\begin{lemma}\label{lem:LB and UB are close}
	\begin{equation*}
		\opt(\lDD)\geq \opt(\uDD)-O(nH\cdot \varepsilon).
	\end{equation*}
\end{lemma}

It is a priori not clear why Lemma~\ref{lem:LB and UB are close} should be true, as $\uDD$ is the ``best'' distribution and $\lDD$ is the ``worst'' distribution in the $\varepsilon$-L\'evy-ball around $\DD$. We prove Lemma~\ref{lem:LB and UB are close} by introducing another two auxiliary distributions $\ltDD$ and $\utDD$. In particular, we construct $\utDD_i$ by shifting all values in $\DD_i$ to the right by $\varepsilon$, and construct $\ltDD_i$ by shifting all values in $\DD_i$ to the left by $\varepsilon$. There are two important properties of these two new distributions: (i) one can couple $\ltDD_i$ with $\utDD_i$ so that the two random variables are always exactly $2\varepsilon$ away from each other; (ii) $\ltDD_i$ and $\lDD_i$ are within Kolmogorov distance $\varepsilon$, and $\utDD_i$ and $\uDD_i$ are also within Kolmogorov distance $\varepsilon$. Property (i) allows us to prove that $\left \lvert \opt(\ltDD)-\opt(\utDD)\right \rvert \leq 2\varepsilon$ (see Claim~\ref{clm:two shifted distributions have similar revenues}). To make use of property (ii), we prove the following robustness theorem w.r.t. the Kolmogorov distance.

\begin{theorem}\label{thm:single-kolmogorov}
	For any buyer $i\in[n]$, let $\DD_i$ and $\hDD_i$ be two arbitrary distributions supported on $(-\infty,H]$ such that $\norm{\DD_i-\hDD_i}_K\leq \varepsilon$. We have the following:	\begin{equation*}
		\opt\left(\hDD\right)\geq \opt(\DD)-3nH\cdot \varepsilon.
	\end{equation*}
	 where $\DD= \bigtimes_{i=1}^n \DD_i$ and $\hDD = \bigtimes_{i=1}^n \hDD_i$.
\end{theorem}
The proof of Theorem~\ref{thm:single-kolmogorov} is postponed to Appendix~\ref{appx:single}. Equipped with Theorem~\ref{thm:single-kolmogorov}, we can immediately show that $\left \lvert \opt(\ltDD)-\opt(\lDD)\right \rvert \leq  O(nH\cdot \varepsilon)$ and $\left \lvert \opt(\utDD)-\opt(\uDD)\right \rvert \leq  O(nH\cdot \varepsilon)$. Lemma~\ref{lem:LB and UB are close} follows quite easily from Claim~\ref{clm:two shifted distributions have similar revenues} and the two inequalities above. The complete proof of Lemma~\ref{lem:LB and UB are close} can be found in Appendix~\ref{appx:single}.

We are now ready to prove Theorem~\ref{thm:single-levy}.

\begin{prevproof}{Theorem}{thm:single-levy}
We first construct $\lDD$ based on $\DD$ and let $M^*$ be the extension of the optimal mechanism for $\lDD$. By Lemma~\ref{lem:LB for the revenue of our mechanism}, we know $\rev(M^*,\hDD)$ is at least $\opt(\lDD)$ for any $\hDD$. We also know that the optimal revenue under  $\hDD$ is at most $\opt(\uDD)$ by Lemma~\ref{lem:UB for the true revenue}, and $\opt(\uDD)\leq \opt(\lDD)+O(nH\cdot \varepsilon)$ by Lemma~\ref{lem:LB and UB are close}. Therefore, $$\rev(M^*,\hDD)\geq \opt(\uDD)-O(nH\cdot \varepsilon)\geq \opt(\hDD)-O(nH\cdot\varepsilon).$$

\end{prevproof}

A simple corollary of Theorem~\ref{thm:single-levy} is the continuity of the optimal revenue under L\'evy distance in single-item settings.
\begin{corollary}\label{cor:single continuity}
If $\DD_i$ and $\hDD_i$ are supported on $[0,H]$, and $\norm{\DD_i-\hDD_i}_L\leq \varepsilon$ for all $i\in[n]$, then $$\left \lvert\opt(\DD)-\opt(\hDD)\right\rvert\leq O\left( nH\cdot \varepsilon\right),$$ where $\DD= \bigtimes_{i=1}^n \DD_i$ and $\hDD = \bigtimes_{i=1}^n \hDD_i$.
\end{corollary}
\section{Robustness for Multi-item Auctions}\label{sec:prokhorov multi}
In this section, we prove our robustness results under the total variation distance and the Prokhorov distance in multi-item settings. As discussed in Section~\ref{sec:roadmap}, the proof strategy for single-item auctions fails miserably in multi-item settings due to the lack of structure of the optimal mechanism. In particular, one of the crucial tools we relied on in single-item settings, the revenue monotonicity, no longer holds in multi-item settings~\cite{HartR12}. Nevertheless, we still manage to provide robustness guarantees in multi-item auctions. The plan is to first prove the robustness result under the total variation distance in Section~\ref{sec:tv-robustness}, then we show show to relate the Prokhorov distance with the total variation distance using randomized rounding in Section~\ref{sec:Prokhorov to TV}, and reduce the robustness under the Prokhorov distance to the robustness under the total variation distance in Section~\ref{sec:robust multi}.

\subsection{TV-Robustness for Multi-item Auctions}\label{sec:tv-robustness}
\notshow{ 

Our TV-robustness result for multi-item auctions is stated in Instead of directly proving Theorem~\ref{thm:multi-item tv-robustness}, we prove a stronger version of the statement where the mechanism $M$ is only approximately BIC.
\begin{lemma}\label{lem:transformation under TV}
	 	Given any distribution $\FF=\bigtimes_{i=1}^n \FF_i$, where each $\FF_i$ is a distribution supported on $[0,H]^m$, and a $\eta$-BIC and IR mechanism $M_1$ w.r.t. $\FF$, we construct 
	 	 a mechanism $M_2$. We show that for any distribution $\hFF=\bigtimes_{i=1}^n \hFF_i \in [0,H]^{nm}$,  if we let $\varepsilon_i=\norm{\hFF_i-\FF_i}_{TV}$ for all $i\in[n]$ and $\rho= \sum_{i\in [n]} \varepsilon_i$, $M_2$ is $2m\LL H\rho +\eta$-BIC w.r.t. $\hFF$ and IR. Moreover, $\revT(M_2, \hFF)\geq \revT(M_1, \FF)-nm\LL H\rho.$ Note that our construction of $M_2$ only depends on $\FF$ and does not require any knowledge of $\hFF$.
	 \end{lemma}
	 
	 }

\begin{theorem}[TV-Robustness for Multi-item Auctions]\label{lem:transformation under TV}
	Given any distribution $\DD=\bigtimes_{i=1}^n \DD_i$, where each $\DD_i$ is a distribution supported on $[0,H]^m$, and a $\eta$-BIC and IR mechanism $M$ w.r.t. $\DD$, we can construct 
	 	 a mechanism $\hM$ such that  for any distribution $\hDD=\bigtimes_{i=1}^n \hDD_i \in [0,H]^{nm}$,  if we let $\varepsilon_i=\norm{\hDD_i-\DD_i}_{TV}$ for all $i\in[n]$ and $\rho= \sum_{i\in [n]} \varepsilon_i$, then $\hM$ is $2m\LL H\rho +\eta$-BIC w.r.t. $\hDD$ and IR. Moreover, $\revT(\hM, \hDD)\geq \revT(M, \DD)-nm\LL H\rho.$ Note that our construction of $\hM$ only depends on $\DD$ and does not require any knowledge of $\hDD$.	
\end{theorem}

We briefly describe the ideas behind the proof. If $\hDD$ and $\DD$ share the same support, it is not hard to see that $M$ is already $(2m\LL H\rho +\eta)$-BIC w.r.t. $\hDD$. The reason is that for any bidder $i$ and any type $v_i$, her expected utility under any report can change by at most $m\LL H\rho$ when the other bidders' bids are drawn from $\hDD_{-i}$ rather than $\DD_{-i}$, as $\norm{\hDD_j-\DD_j}_{TV}=\varepsilon_j$ for all $j\in [n]$. The bulk of the proof is dedicated to the case, where $\hDD$ and $\DD$ have different supports. We construct mechanism $\hM$, which first takes each bidder $i$'s report and maps it to the ``best'' possible report from $\supp(\DD_i)$, then runs essentially $M$ on the transformed reports. We show that $\hM$ is $2m\LL H\rho +\eta$-BIC w.r.t. $\hDD$ and generates at most $nm\LL H\rho$ less revenue. The proof of Theorem~\ref{lem:transformation under TV} is postponed to Appendix~\ref{sec:TV robustness}.

\subsection{Connecting the Prokhorov Distance with the Total Variation Distance}\label{sec:Prokhorov to TV}

In this section, we provide a randomized rounding scheme that relates the Prokhorov distance to the total variation distance.
We first state a characterization of the Prokhorov distance due to Strassen~\cite{Strassen65} that is useful for our analysis.

\begin{theorem}[Characterization of the Prokhorov Metric \cite{Strassen65}]\label{thm:prokhorov characterization}
	Let $\FF$ and $\hFF$ be two distributions supported on $\mathbb{R}^k$. $\norm{\FF-\hFF}_P\leq \varepsilon$ if and only if there exists a coupling $\gamma$ of $\FF$ and $\hFF$, such that $\Pr_{(x,y)\sim \gamma}\left[d(x,y)>\varepsilon \right]\leq \varepsilon,$
	where $d(\cdot,\cdot)$ is the $\ell_1$ distance.
\end{theorem}

Theorem~\ref{thm:prokhorov characterization} states that $\FF$ and $\hFF$ are within Prokhorov distance $\varepsilon$ of each other if and only if there exists a coupling between the two distributions such that the two random variables are within $\varepsilon$ of each other with probability at least $1-\varepsilon$. Next, we show that if $\FF$ and $\hFF$ are close to each other in Prokhorov distance, then one can use a randomized rounding scheme to round both $\FF$ and $\hFF$ to discrete distributions so that the two rounded distributions are close in total variation distance with high probability.

First, let us fix some notations.

\begin{definition}[Rounded Distribution]\label{def:rounded dist}
	Let $\FF$ be a distribution supported on $\mathbb{R}_{\geq 0}^k$. For any $\delta > 0$ and $\ell\in [0,\delta]^k$, we define function $\rfun: \mathbb{R}_{\geq 0}^k\mapsto \mathbb{R}^k$  as follows:
	$	\rfun_i(x) = \max \left\{ \left\lfloor \frac{x_i-\ell_i}{\delta} \right\rfloor \cdot \delta+\ell_i , 0\right\}$
	for all $i\in [k]$. Let $X$ be a random variable sampled from distribution $\FF$. We define $\round{\FF}_{\ell,\delta}$ as the distribution for the random variable $\rfun(X)$, and we call $\round{\FF}_{\ell,\delta}$ as the rounded distribution of $\FF$.
\end{definition}

\begin{lemma}\label{lem:prokhorov to TV}
	Let $\FF$ and $\hFF$ be two distributions supported on  $\mathbb{R}^k$, and $\norm{\FF-\hFF}_P\leq \varepsilon$. For any $\delta>0$, sample $\ell$ from the uniform distribution over $[0,\delta]^k$, $\E_{\ell\sim U[0,\delta]^k}\left[\norm{\round{\FF}_{\ell,\delta}-\round{\hFF}_{\ell,\delta}}_{TV}\right]\leq \left(1+\frac{1}{\delta}\right)\varepsilon.$
	\notshow{then for any constant $\CC\geq 1$ \begin{equation*}
		\Pr_{\ell\sim U[0,\delta]^k}\left[\norm{\round{\FF}_{\ell,\delta}-\round{\hFF}_{\ell,\delta}}_{TV}\leq \CC\left(1+\frac{1}{\delta}\right)\varepsilon \right]\geq 1- 1/\CC.
	\end{equation*}}
\end{lemma}
	
	We only sketch the idea and postpone the formal proof to Appendix~\ref{appx:randomized rounding}. Let $x$ be a random variable sampled from $\FF$ and $y$ be a random variable sampled from $\hFF$. Since $\FF$ and $\hFF$ are close in Prokhorov distance, we can couple $x$ and $y$ according to Theorem~\ref{thm:prokhorov characterization} such that they are within $\varepsilon$ of each other with probability at least $1-\varepsilon$. The rounding scheme chooses a random origin $\ell$ from $[0,\delta]^k$ and rounds $\FF$ and $\hFF$ to the corresponding random grid with width $\delta$. More specifically, we round $\FF$ and $\hFF$ to $\round{\FF}_{\ell,\delta}$ and $\round{\hFF}_{\ell,\delta}$ respectively. For simplicity, consider $\delta=\Theta(\sqrt{\varepsilon})$. The key observation is that when $x$ and $y$ are within $\ell_1$-distance $\varepsilon$ of each other, they lie in the same grid with probability at least $1-O(\sqrt{\varepsilon})$ over the randomness of $\ell$. If $x$ and $y$ are in the same grid, they will be rounded to the same point. In other words, the coupling between $x$ and $y$ induces a coupling between $\round{\FF}_{\ell,\delta}$ and $\round{\hFF}_{\ell,\delta}$ such that, in expectation over the choice of $\ell$, the event that the corresponding two rounded random variables have different values happens with probability at most $\varepsilon+(1-\varepsilon)\cdot O(\sqrt{\varepsilon}) = O(\sqrt{\varepsilon})$. By the definition of total variation distance, this implies that the expected total variation distance between $\round{\FF}_{\ell,\delta}$ and $\round{\hFF}_{\ell,\delta}$ is also at most $ O(\sqrt{\varepsilon})$. A similar argument applies to  other choices of $\delta$.
	
	 \subsection{Prokhorov-Robustness for Multi-item Auctions}\label{sec:robust multi}
	 
	 In this section, we show that even in multi-item settings, if every bidder's approximate type distribution $\DD_i$ is within Prokhorov distance $\varepsilon$ of her true type distribution $\hDD_i$, given any BIC and IR mechanism $M$ for $\DD=\bigtimes_{i=1}^n \DD_i$, we can construct a mechanism $\hM$ that is {$O(\poly(n,m,\LL,H,\varepsilon))$}-BIC w.r.t. $\hDD=\bigtimes_{i=1}^n \hDD_i$, IR, and its revenue under truthful bidding $\revT(\hM,\hDD)$ is at most $O(\poly(n,m,\LL,H,\varepsilon))$ worse than $\rev(M,\DD)$. 
	 
 \begin{theorem}\label{thm:multi-item p-robustness}
	 	Suppose we are given $\DD=\bigtimes_{i=1}^n \DD_i$, where $\DD_i $ is an $m$-dimensional distribution for each $i\in[n]$, and a BIC and IR mechanism $M$ w.r.t. $\DD$. Suppose $\hDD=\bigtimes_{i=1}^n \hDD_i$ is the true but unknown type distribution such that $\norm{\DD_i-\hDD_i}_P\leq \varepsilon$ for all $i\in[n]$. We can construct a randomized mechanism $\hM$, oblivious to the true distribution $\hDD$, such that for any $\hDD$ the followings hold: 
	 	\begin{enumerate}
	 		\item $\hM$ is $\kappa$-BIC w.r.t. $\hDD$ and IR, where  $\kappa=O\left( nm \LL H \varepsilon+   m \LL \sqrt{nH\varepsilon}\right)$;
	 		\item  the expected revenue of $\hM$ under truthful bidding is $\revT\left(\hM,\hDD\right)\geq \rev(M,\DD)-O\left(n\kappa\right).$
	 	\end{enumerate}
	 	
	\notshow{ 	If $\norm{\DD_i-\hDD_i}_{TV}\leq \varepsilon$ for all $i\in[n]$, our algorithm becomes deterministic and provides stronger guarantees. More specifically, \begin{enumerate}
	 		\item $\hM$ is a $O(n m \LL H \varepsilon)$-BIC and IR mechanism for $\hDD$;
	 		\item  the expected revenue of $\hM$ under truthful bidding is $$\revT(\hM,\hDD)\geq \rev(M,\DD)-O\left(n^2 m \LL H \varepsilon\right).$$
	 	\end{enumerate}}
	 	 \end{theorem}
	 
	\notshow { \begin{theorem}\label{thm:multi-item p-robustness}
	 	Given $\DD=\bigtimes_{i=1}^n \DD_i$, where $\DD_i $ is a $m$-dimensional distribution for all $i\in[n]$, and a BIC mechanism $M$ w.r.t. $\DD$. Let $\alpha$ be an arbitrary real number in $(0,1)$. We have randomized algorithm that constructs a mechanism $\hM$ \yangnote{in polynomial time given access to the oracle that provides the best bid} such that for any possible true type distribution $\hDD=\bigtimes_{i=1}^n \hDD_i$ satisfying $\norm{\DD_i-\hDD_i}_P\leq \varepsilon$ for all $i\in[n]$, with probability at least $1-\alpha$ over the randomness in the construction of $\hM$, $\hM$ is a $\frac{\kappa}{\alpha}$-BIC and IR mechanism for $\hDD$, where  $\kappa=O\left( n^2 m \LL H \varepsilon+  n m \LL \sqrt{H\varepsilon}\right)$. Moreover, the expected revenue of $\hM$ under truthful bidding is {$$\revT(\hM,\hDD)\geq \rev(M,\DD)-O\left(\frac{n\kappa}{\alpha}\right).$$} 
	 \end{theorem}}
	 
	 We postpone the formal proof of Theorem~\ref{thm:multi-item p-robustness} to Appendix~\ref{appx:proof_multi}.
We provide a complete sketch here. Our construction consist of the following five steps. 
	 \begin{itemize}
	
	 	\item \textbf{Step (1):} After receiving the bid profile, first sample $\ell$ from $U[0,\delta]^m$. 
	 	 For every realization of $\ell$, we construct a mechanism $\hMl$ and execute $\hMl$ on the reported bids. In the next several steps, we show how to construct $\hMl$ via two intermediate mechanisms $\Mli$ and $\Mlii$ for every realization of $\ell$ based on $M$. Since $\ell$ is a random variable, $\hM$ is a randomized mechanism.
	 	 
	 	\item \textbf{Step (2):} Round $\DD_i$ to $\round{\DD_i}_{\ell,\delta}$ for every bidder $i$. We construct mechanism $\Mli$ based on $M$ and show that $\Mli$ is $O(m\L \delta)$-BIC w.r.t. $\bigtimes_{i=1}^n \round{\DD_i}_{\ell,\delta}$ and IR. Moreover, $$\revT\left(\Mli, \bigtimes_{i=1}^n \round{\DD_i}_{\ell,\delta}\right)\geq \rev(M,\DD)-O(nm\L \delta).$$ 
	 	 Here is the idea behind the construction: for any bidder $i$ and  type $w_i$ drawn from $\round{\DD_i}_{\ell,\delta}$, we resample a type from $\DD_i \mid \bigtimes_{j=1}^m [w_{ij},w_{ij}+\delta)$, which is the distribution induced by $\DD_i$ conditioned on being in the cube $\bigtimes_{j=1}^m [w_{ij},w_{ij}+\delta)$. We use the allocation rule of $M$ and a slightly modified payment rule on the resampled type profile. This guarantees that the new mechanism is $O(m\L \delta)$-BIC w.r.t. $\bigtimes_{i=1}^n \round{\DD_i}_{\ell,\delta}$ and IR. The formal statement and analysis are shown in Lemma~\ref{lem:original to rounded}.
	 	
	 	\item \textbf{Step (3):} We use $\el_i$ to denote $\norm{\round{\DD_i}_{\ell,\delta}-\round{\hDD_i}_{\ell,\delta}}_{TV}$ for our sample $\ell$ and every $i\in [n]$, and $\rho^{(\ell)}$ to denote $\sum_{i\in [n]}\el_i$. We transform $\Mli$ into a new mechanism $\Mlii$ using Theorem~\ref{lem:transformation under TV}. In particular, $\Mlii$ is $O\left(m\L\delta+m\L H\cdot \rho^{(\ell)}\right)$-BIC w.r.t.  	 $\bigtimes_{i=1}^n \round{\hDD_i}_{\ell,\delta}$ and IR. Importantly, the construction of $\Mlii$ is oblivious to $\bigtimes_{i=1}^n \round{\hDD_i}_{\ell,\delta}$ and $\left\{\el_i\right\}_{i\in[n]}$. Moreover, $\revT\left(\Mlii, \bigtimes_{i=1}^n \round{\hDD_i}_{\ell,\delta}\right)\geq \revT\left(\Mli, \bigtimes_{i=1}^n \round{\DD_i}_{\ell,\delta}\right)-O\left(nm\L H\cdot \rho^{(\ell)}\right).$
	 		 	
	 	\item \textbf{Step (4):} We convert $\Mlii$ to $\hMl$ so that it is $O\left(m\L\delta+m\L H\cdot\rho^{(\ell)}\right)$-BIC w.r.t. $\hDD$, IR and $$\revT(\hMl, \hDD)\geq \revT\left(\Mlii, \bigtimes_{i=1}^n \round{\hDD_i}_{\ell,\delta}\right)-nm\L\delta.$$ Here is the idea behind the construction of $\hMl$: for every bidder $i$ and her type $w_i$ drawn from $\hDD_i$, round it to $\rfun_i(w_i)$ (see Definition~\ref{def:rounded dist}). We use the allocation rule of $\Mlii$ and a slightly modified payment rule on the rounded type profile. This guarantees that the new mechanism is $O\left(m\L\delta+m\L H\cdot \rho^{(\ell)}\right)$-BIC w.r.t. $\hDD$ and IR. Note that our construction only requires knowledge of $\Mlii$, $\ell$, and $\delta$, and is completely oblivious to $\hDD$ and $\bigtimes_{i=1}^n \round{\hDD_i}_{\ell,\delta}$. The formal statement and analysis are shown in Lemma~\ref{lem:discrete to continuous}.
	 	\item \textbf{Step (5):} Since for every realization of $\ell$, $\hMl$ is $O\left(m\L\delta+m\L H\cdot \rho^{(\ell)}\right)$-BIC w.r.t. $\hDD$ and IR, $\hM$ must be  $O\left(m\L\delta+m\L H\cdot \E_{\ell\sim U[0,\delta]^m}\left[\rho^{(\ell)}\right]\right)$-BIC w.r.t. $\hDD$ and IR. According to Lemma~\ref{lem:prokhorov to TV}, $\E_{\ell\sim U[0,\delta]^m}\left[\rho^{(\ell)}\right]=\sum_{i\in[n]} \E_{\ell\sim U[0,\delta]^m}\left[\el_i\right]=n\cdot \left(1+{1\over\delta}\right)\varepsilon.$ Therefore, $\hM$ is \\$O\left(m\L\delta+nm\L H\left(1+{1\over\delta}\right)\epsilon\right)$-BIC w.r.t. $\hDD$ and IR.  Moreover,
	 	$$\revT\left(\hM, \hDD\right)\geq \rev(M,\DD)-O\left(nm\LL\delta+ n^2m\L H\left(1+{1\over\delta}\right)\varepsilon\right).$$
	 		 \end{itemize} 	
	 
	 \begin{lemma}\label{lem:original to rounded}
	Given any $\delta>0$, $\ell\in [0,\delta]^m$, and a BIC and IR mechanism $M$ w.r.t. $\DD$, we can construct 
	 a $\xi_1$ = $O(m\LL  \delta)$-BIC w.r.t. $\lDD=\bigtimes_{i=1}^n \round{\DD_i}_{\ell,\delta}$~and~IR~mechanism $\Mli$,~such~that~$$\revT\left(\Mli, \lDD\right)\geq \rev(M,\DD)-nm\LL\delta.$$
\end{lemma}	 
The proof of Lemma~\ref{lem:original to rounded} can be found in Appendix~\ref{appx:proof_multi}. In the next Lemma, we make \textbf{Step (4)} formal.

	 
\begin{lemma}\label{lem:discrete to continuous}
	For any $\delta>0$, $\ell \in [0,\delta]^m$, and distribution $\hDD$, if $\Mlii$ is a $\xi_2$-BIC  w.r.t. $\hlDD=\bigtimes_{i=1}^n \round{\hDD_i}_{\ell,\delta}$ and IR mechanism, we can transform $\Mlii$ into a mechanism $\hMl$, so that $\hM$ is $(\xi_2+3m\LL \delta)$-BIC w.r.t. $\hDD$, IR, and has revenue under truthful bidding $\revT\left(\hMl,\hDD\right)\geq \revT\left(\Mlii,\hlDD\right)-nm\LL\delta.$ Moreover, the transformation does not rely on any knowledge of $\hDD$ or $\hlDD$.
\end{lemma} 

	 The proof of Lemma~\ref{lem:discrete to continuous} is postpone to Appendix~\ref{appx:proof_multi}. 

\subsection{Applications of Multi-Item Robustness}
\paragraph{Lipschitz Continuity of the Optimal Revenue in Multi-item Auctions.} Equipped with Theorem~\ref{lem:transformation under TV} and \ref{thm:multi-item p-robustness}, we can easily argue the Lipschitz continuity of the optimal revenue in multi-item auctions (Theorem~\ref{thm:continuous of OPT under Prokhorov}) as stated in the last column of the second half of Table~\ref{table 1}.  Due to Theorem~~\ref{lem:transformation under TV} and \ref{thm:multi-item p-robustness}, we know that the optimal revenue of a {$O(\poly(n,m,\LL,H,\varepsilon))$}-BIC and IR mechanism w.r.t. distribution $\FF=\bigtimes_{i\in [n]} \FF_i$ is at least as large as the optimal revenue of a BIC and IR mechanism w.r.t. distribution $\hFF=\bigtimes_{i\in [n]} \hFF_i$, if $\norm{\FF_i-\hFF_i}_{TV}\leq \varepsilon, \forall i$ or $\norm{\FF_i-\hFF_i}_{P}\leq \varepsilon, \forall i$. According to Lemma~\ref{lem:eps-BIC to BIC}, the optimal revenue of a {$O(\poly(n,m,\LL,H,\varepsilon))$}-BIC and IR mechanism is at most $O(\poly(n,m,\LL,H,\varepsilon))$ larger than the optimal revenue of a BIC and IR mechanism. Hence, $\opt(\FF)\approx \opt(\hFF)$. Please see Appendix~\ref{sec:Lipschitz Continuity of Opt} for the formal statement and the proof of Theorem~\ref{thm:continuous of OPT under Prokhorov}.

\paragraph{Approximation Preserving Transformation.} One interesting implication of Theorem~\ref{thm:continuous of OPT under Prokhorov} is that the transformations of Theorems~\ref{lem:transformation under TV}~and~\ref{thm:multi-item p-robustness} are also approximation preserving. Given a a $c$-approximation mechanism $M$ to the optimal revenue under distribution $\DD$, applying the transformation in Theorem~\ref{thm:multi-item p-robustness} (or Theorem~\ref{lem:transformation under TV}) to $M$, we obtain a new mechanism $\hM$ that is $O(\poly(n,m,\LL,H,\varepsilon))$-BIC w.r.t. $\hDD$ and IR if $\norm{\DD_i-\hDD_i}_P\leq \varepsilon, \forall i$ (or if $\norm{\DD_i-\hDD_i}_{TV}\leq \varepsilon, \forall i$ ). Moreover, its revenue under truthful bidding is at least $c$ fraction of the optimal $O(\poly(n,m,\LL,H,\varepsilon))$-BIC revenue under $\hDD$ less a small additive term. The result is formally stated as Theorem~\ref{thm:approximation preserving} in Appendix~\ref{sec:approximation preserving guarantee}. Note that the third column of the second half of Table~\ref{table 1} is simply Theorem~\ref{thm:approximation preserving} with $c=1$. Furthermore, if there is only a single bidder, the mechanism $\hM$ becomes exactly IC instead of approximately IC (Theorem~\ref{thm:single approximation preserving}).

\paragraph{Learning Multi-item Auctions under Item Independence.} Since independent distributions are straightforward to learn within Prokhorov distance $\varepsilon$ with polynomially many samples, the result of Gonczarowski and Weinberg~\cite{GonczarowskiW18} follows easily from our robustness result (see Theorem~\ref{thm:multi-item auction sample} in Appendix~\ref{sec:old sampling}). 

\paragraph{Learning Multi-item Auctions under Structured Item Dependence.} Going beyond product measures, we initiate the study of learning multi-item auctions when every bidder's item-values are dependent, but sampled from a joint distribution with structure. As we have already noted, arbitrary joint distributions are both unnatural from a modeling perspective, as they require exponentially many bits to describe, and are also known to require exponentially many samples to even learn  approximately optimal auctions~\cite{DughmiLN14}. We thus propose studying the learnability of auctions under the assumption that each bidder's item values are sampled from a Markov Random Field (MRF) or a Bayesian network (a.k.a.~Bayeset). In fact, this is not really an assumption. These well-studied probabilistic frameworks, defined formally in Definitions~\ref{def:MRF} and~\ref{def:Bayesnet} of Appendix~\ref{sec:new sampling} due to lack of space, are very flexible in that they can represent {\em any distribution}. The reason they are attractive from a modeling perspective is that they have a natural complexity parameter that controls how expressive they are, namely the maximum hyperedge size of an MRF and the maximum in-degree of a Bayesnet. Under the assumption that each bidder's item-values are drawn from an MRF or a Bayesnet of complexity $d$, we establish the results summarized in the last two rows of Table~\ref{tab:sample-based results}, whose main feature is that the sample complexity to learn an up-to-$\epsilon$ optimal auction is polynomial in the number of bidders $n$, the number of items $m$, the inverse approximation parameter $1/\varepsilon$, and other relevant parameters, and is only exponential in the complexity parameter $d$ of the bidders' MRFs or Bayesian networks, as it should given the known lower bounds~\cite{DughmiLN14}. 

Our results for learning near-optimal auctions under MRF and Bayesnet assumptions are stated in more detail as Theorems~\ref{thm:auction MRF} and~\ref{thm:auction Bayesnet} of Appendix~\ref{sec:new sampling}, and {\em can also accommodate unobservable variables} which makes their applicability very broad. In turn, these results are proven by combining our robustness result (Theorem~\ref{thm:approximation preserving}) with new learnability results for MRFs and Bayesnets that we also establish, namely Theorems~\ref{thm:sample complexity MRF finite alphabet} and~\ref{thm:sample complexity Bayesnets finite alphabet} of Appendix~\ref{sec:new sampling} respectively. These results are of independent interest and provide broad generalizations of the recent upper bounds of~\cite{devroye2018minimax} for Gaussian MRFs and Ising models. While this recent work bounds the VC dimension of the Yatracos class of these families of distributions, for our more general families of non-parametric distributions we construct instead covers under either total variation distance or Prokhorov distance, and combine our cover-size upper bounds with generic tournament-style algorithms; see~e.g.~\cite{devroye2012combinatorial,DaskalakisK14,AcharyaJOS14}.  The details are provided in Appendix~\ref{appx:sample}. While there are many details, we illustrate one snippet of an idea used in constructing a $\varepsilon$-cover, in total variation distance, of the set of all MRFs with hyper-edges $E$ of size at most $d$ and a discrete alphabet $\Sigma$ on every node. The proof argues that (i) the (appropriately normalized) log-potential functions of the MRF can be discretized to take values in the negative integers at a cost of $\varepsilon$ in total variation distance; (ii) using properties of linear programming, it argues that using negative integers of bit complexity polynomial in $|E|$, $|\Sigma|^d$ and $\log(1/\varepsilon)$ suffices at another cost of $\varepsilon$ in total variation distance. It thus argues that all MRFs can be covered by a set of MRFs of size exponential in ${\rm poly}\left(|E|, |\Sigma|^d, \log({1\over \varepsilon})\right)$, which is sufficient to yield the required sample bounds using the tournament algorithm.

\section{Learning Multi-item Auctions under Item Independence}\label{sec:old sampling}

In this section, we show how to derive one of the state-of-the-art learnability results for learning multi-item auctions via our robustness results. We consider the case where every bidder's type distribution is a $m$-dimensional product distribution. We will show that a generalization of the main result by Gonczarowski and Weinberg~\cite{GonczarowskiW18} follows easily from our robustness result. The main idea is that it suffices to learn the distribution $\FF_i$ within small Prokhorov distance for every bidder $i$, and it only requires polynomial many samples when each $\FF_i$ is a product distribution.

\begin{theorem}\label{thm:multi-item auction sample}
Consider the general mechanism design setting of Section 2. Recall that $\L$ is the Lipschitz constant of the valuations.	For every $\varepsilon,\delta>0$, and for every $\eta\leq \poly(n,m,\L,H,\varepsilon)$, we can learn a distribution $\DD= \bigtimes_{i\in[n],j\in[m]}\DD_{ij}$ with $\poly\left( n,m,\L,H,1/\varepsilon, 1/\eta, \log(1/\delta)\right)$ 
	 samples from $\hDD=\bigtimes_{i\in[n],j\in[m]}\hDD_{ij}$, such that, with probability $1-\delta$, we can transform any BIC w.r.t. $\DD$, IR, and $c$-approximation mechanism $M$ to an $\eta$-BIC w.r.t. $\hDD$ and IR mechanism $\hM$, whose revenue under truthful bidding satisfies
	$$\revT\left(\hM,\hDD\right)\geq c\cdot \opt_{\eta}\left(\hDD\right)-\varepsilon.$$
	If $n=1$, the mechanism $\hM$ will be IR and IC, and $$\rev\left( \hM, \hDD\right)\geq c\cdot (1-\sqrt{\eta})\cdot \opt\left(\hDD\right)-\varepsilon-\sqrt{\eta}.$$
\end{theorem}

In particular, Gonczarowski and Weinberg~\cite{GonczarowskiW18} proved the $c=1$ case, and our result applies to any $c\in(0,1]$. The proof is given in Appendix~\ref{appx:sample independent}. We provide a proof sketch here. We first prove Lemma~\ref{lem:discretized product TV}, which shows that polynomially many samples suffice to learn a distribution $\DD$ that is close to $\hDD$ in Prokhorov distance. Now the statement simply follows from Theorem~\ref{thm:approximation preserving}.

\section{Missing Proofs from Section~\ref{sec:old sampling}}\label{appx:sample independent}

We first show that for any product distribution $\FF$, we can learn the rounded distribution of $\FF$ within small TV distance with polynomially many samples.

\begin{lemma}\label{lem:discretized product TV}
	Let $\FF=\bigtimes_{j=1}^m \FF_j$, where $\FF_j$ is an arbitrary distribution supported on $[0,H]$ for every $j\in[m]$. Given $N=O\left({m^3H\over \eta^3}\cdot(\log 1/\delta+\log m)\right)$ samples, we can learn a product distribution $\hFF=\bigtimes_{j=1}^m \hFF_j$ such that $$\norm{{\FF}-\hFF}_{P}\leq \eta$$ with probability at least $1-\delta$.\end{lemma}
	
\begin{proof}
	We denote the samples as $s^1,\ldots,s^N$. Round each sample to multiples of $\eta'=\eta/m$. More specifically, let $\hat{s}^i= \left(\round{s^i_1/\eta'}\cdot \eta',\ldots, \round{s^i_m/\eta'}\cdot \eta' \right )$ for every sample $i\in[N]$. Let $\hFF_j$ be the uniform distribution over $\hat{s}^1_j,\ldots, \hat{s}^N_j$. Let $\uFF_j=\round{\FF_j}_{0,\eta'}$. Note that $\hFF_j$ is the empirical distribution of $N$ samples from $\uFF_j$. As $\left\lvert \supp(\uFF_j)\right\rvert = \round{\frac{H}{\eta'}}={mH\over \eta}$, with $N=O\left({ \lvert \supp(\uFF_j)\rvert\over \eta'^2}\cdot(\log 1/\delta+\log m)\right)$ samples, the empirical distribution $\hFF_j$ should satisfy $\norm{\hFF_j-\uFF_j}_{TV}\leq \eta'$ with probability at least $1-\delta/m$. By the union bound $\norm{\hFF_j-\uFF_j}_{TV}\leq \eta'$ for all $j\in[m]$ with probability at least $1-\delta$, which implies {$\norm{\hFF-\uFF}_{TV}\leq \eta$} with probability at least $1-\delta$. Observe that $\uFF$ and $\FF$ can be coupled so that the two samples are always within $\eta$ in $\ell_1$ distance. When {$\norm{\hFF-\uFF}_{TV}\leq \eta$}, consider the coupling between $\hFF$ and $\FF$ by composing the optimal coupling between $\hFF$ and $\uFF$ and the coupling between $\uFF$ and $\FF$. Clearly, the two samples from $\hFF$ and $\FF$ are within $\ell_1$ distance $\eta$ with probability at least $1-\eta$. Due to Theorem~\ref{thm:prokhorov characterization}, the existence of this coupling implies that {$\norm{\hFF-\FF}_{P}\leq \eta$}. 
\end{proof}

\begin{prevproof}{Theorem}{thm:multi-item auction sample}
We only consider the case, where $\eta\leq \alpha\cdot \min\left\{{\varepsilon\over n}, {\varepsilon^2\over n^2m\LL H}\right\}$. $\alpha$ is an absolute constant and we will specify its choice in the end of the proof.

In light of Lemma~\ref{lem:discretized product TV}, we take $N=O\left({m^3H\over \sigma^3}\cdot(\log {n\over \delta}+\log m)\right)$ from $\hDD$ and learn a distribution $\DD$ so that, with probability at least $1-\delta$, $\norm{\DD_i-\hDD_i}_P\leq \sigma$ for all $i\in[n]$. According to Theorem~\ref{thm:approximation preserving}, we can transform $M$ into mechanism $\hM$ that is $O\left( nm\LL H \sigma+m\LL\sqrt{n H\sigma} \right)$-BIC w.r.t. $\hDD$ and IR. Choose $\sigma$ in a way so that $\hM$ is $\eta$-BIC w.r.t. $\hDD$. Moreover, $\hM$'s revenue under truthful bidding satisfies 	
$$\revT\left(\hM,\hDD\right)\geq c\cdot \opt_{\eta}\left(\hDD\right)-O\left(n\eta+n\sqrt{m\LL H\eta}\right).$$ If we choose $\alpha$ to be sufficiently small, then $$\revT\left(\hM,\hDD\right)\geq c\cdot \opt_{\eta}\left(\hDD\right)-\varepsilon.$$

When there is only a single-bidder, we can apply Lemma~\ref{lem:eps-IC to IC} to transform $\hM$ to an IC and IR mechanism, whose revenue satisfies the guarantee in the statement.\end{prevproof}

\notshow{
	We take $N=O\left({m^3H\over \sigma^3}\cdot(\log 1/\delta+\log m)\right)$ samples from $\round{\hDD_i}_{0,\sigma/m}$, and let $\DD_i$ be the empirical distribution over the samples. By Lemma~\ref{lem:discretized product TV}, $\norm{\DD_i-\round{\hDD_i}_{0,\sigma/m}}_{TV} \leq \sigma$ with probability $1-\delta$. We use $\hlDD$ to denote $\bigtimes_{i=1}^n \round{\hDD_i}_{0,\sigma/m}$ and  choose $\sigma = \frac{\tau^2}{n^3m^2L^2H^2}$. By Theorem~\ref{thm:how close do we need to learn}, we can construct a mechanism $M_1$ based on $M$ that is $\alpha=O(\frac{\tau^2}{nmLH})$-BIC and IR for $\hlDD$. Since $M$ is 
	c-approximate to $\text{OPT}(\DD)$, we also have revenue bound $$\revT(M_1,\hlDD ) \geq c\cdot\opt_{\alpha}(\hlDD) - \tau .$$
	
Next, we construct $\hM$ based Lemma~\ref{lem:discrete to continuous}. Since $M_1$ is $\alpha$-BIC and IR with respect to $\hlDD$, $\hM$ that is $(\alpha+3\L{\sigma})$-BIC and IR with respect to $\hDD$ and 
$$\revT(\hM,\hDD) \geq \revT(M_1,\hlDD)-n\L{\sigma} \geq c \cdot \opt_\alpha(\hlDD) - \tau-n\L{\sigma}.$$

Let $M^*$ be the optimal BIC and IR mechanism w.r.t. $\hDD$. According to Lemma~\ref{lem:original to rounded}, 
there exists a $O(\L \sigma)$-BIC and IR mechanism $M'$ such that $$\revT(M',\hlDD)\geq \rev(M^*,\hDD)-n\L\sigma=\opt(\hDD)-n\L\sigma.$$ Note that $\alpha>\L\sigma$, so $\revT(M',\hlDD)\leq \opt_\alpha(\hlDD)$. Hence, $$\revT(\hM,\hDD) \geq c \cdot \opt_\alpha(\hlDD) - \tau-n\L{\sigma}\geq  c \cdot \opt(\hDD) - \tau-(c+1)n\L{\sigma}.$$ 

Finally, $\opt(\hDD)\geq \opt_{\alpha+3\L\sigma}(\hDD)-2n\sqrt{m\L H(\alpha+3\L\sigma)}$ by Lemma~\ref{lem:eps-BIC to BIC}. Let $\eta=\alpha+3\L\sigma$. To sum up, $\hM$ is a $\eta$-BIC mechanism and 
$$\revT(\hM,\hDD)\geq c\cdot \opt_{\eta}(\hDD)-\tau-(c+1)n\L{\sigma}-2n\sqrt{m\L H\eta}.$$

Choose $\tau$ appropriately so that $\varepsilon=\tau+(c+1)n\L{\sigma}+2n\sqrt{m\L H\eta}$. In particular, $\eta=\frac{\varepsilon^2}{\poly(n,m,\L,H)}$ and $N=O(\frac{\poly(n,m,\L,H,\log(1/\delta))}{\varepsilon^6})$.
}

\section{Optimal Mechanism Design under Structural Item Dependence}\label{sec:new sampling}

In this section, 
we go beyond the standard assumption of item-independence, which has been employed in most of  prior literature, to consider settings where, as is commonly the case in practice,  item values are correlated. Of course, once we embark onto a study of correlated distributions, we should not go all the way to full generality, since exponential sample-size lower bounds are known, even for learning approximately optimal mechanisms in single-bidder unit-demand settings~\cite{DughmiLN14}. Besides those sample complexity lower bounds, however, fully general distributions are also not very natural. In practice, high-dimensional distributions are not arbitrary, but have structure, which allows us to perform inference on them and learn them more efficiently. We thus propose the study of optimal mechanism design under the assumption that item values are jointly sampled from a high-dimensional distribution with structure. 

There are many probabilistic frameworks that allow modeling structure in a high-dimensional distribution. In this work we consider one of the most prominent ones: {\em graphical models}, and in particular consider the two most common types of graphical models: {\em Markov Random Fields} and {\em Bayesian Networks}.

\begin{definition} \label{def:MRF} A {\em Markov Random Field (MRF)} is a distribution defined by a hypergraph $G=(V,E)$. Associated with every vertex $v \in V$ is a random variable $X_v$ taking values in some alphabet $\Sigma$, as well as a potential function $\psi_v: \Sigma \rightarrow [0,1]$. Associated with every hyperedge $e \subseteq V$ is a potential function $\psi_e: \Sigma^e \rightarrow [0,1]$. In terms of these potentials, we define a probability distribution $p$ associating to each vector $x \in \Sigma^V$ probability $p(x)$ satisfying:
\begin{align}
p(x) = {1 \over Z} \prod_{v\in V} \psi_v(x_v) \prod_{e\in E} \psi_e(x_e), \label{eq:graphical model}
\end{align}
where for a set of nodes $e$ and a vector $x$ we denote by $x_e$ the restriction of $x$ to the nodes in $e$, and $Z$ is a normalization constant making sure that $p$, as defined above, is a distribution. In the degenerate case where the products on the RHS of~\eqref{eq:graphical model} always evaluate to $0$, we assume that $p$ is the uniform distribution over $\Sigma^V$. In that case, we get the same distribution by assuming that all potential functions are identically $1$. Hence, we can in fact assume that the products on the RHS of~\eqref{eq:graphical model} cannot always evaluate to $0$.
\end{definition}

\begin{definition} \label{def:Bayesnet}
A {\em Bayesian network}, or {\em Bayesnet}, specifies a probability distribution in terms of a directed acyclic graph $G$ whose nodes $V$ are random variables taking values in some alphabet $\Sigma$. To describe the probability distribution, one specifies conditional probabilities $p_{X_v|X_{\Pi_v}}(x_v|x_{\Pi_v})$, for all vertices $v$ in $G$, and configurations $x_v\in \Sigma$ and $x_{\Pi_v} \in \Sigma^{\Pi_v}$, where $\Pi_v$ represents the set of parents of $v$ in $G$, taken to be $\emptyset$ if $v$ has no parents. In terms of these conditional probabilities, a probability distribution over $\Sigma^V$ is defined as follows:
$$p(x)  = \prod_{v} p_{X_v | X_{\Pi_v}} (x_v | x_{\Pi_v}), \text{for all }x \in \Sigma^V.$$
\end{definition}

It is important to note that both MRFs and Bayesnets allow the study of distributions in their {\em full generality}, as long as the graphs on which they are defined are sufficiently dense. In particular, the graph (hypergraph and DAG respectively) underlying these models captures conditional independence relations, and is sufficiently flexible to capture the structure of intricate dependencies in the data. As such these models have found myriad applications; see e.g.~\cite{jensen1996introduction,nielsen2009bayesian,pearl2009causality,kindermann1980markov} and their references. A common way to control the expressiveness of MRFs and Bayesnets is to vary the maximum size of hyperedges in an MRF and indegree in a Bayesnet. Our sample complexity results presented below will be parametrized according to this measure of complexity in the distributions.


\smallskip In our results, presented below, we exploit our modular framework to disentangle the identification of good mechanisms for these settings from the intricacies of learning a good model of the underlying distribution from samples. In particular, we are able to pair our mechanism design framework presented in earlier sections with learning results for MRFs and Bayesnets to characterize the sample complexity of learning good mechanisms when the item distributions  are MRFs and Bayesnets. Below, we first present our results on the sample complexity of learning good mechanisms in these settings, followed by the learning results for MRFs and Bayesnets that these are modularly dependent on.



\subsection{Learning Multi-item Auctions under Structural Item Dependence}\label{sec:sample complexity multi-item auctions correlated}

In this section, we state our results for learning multi-item auctions when each bidder's values correlated. In particular, we consider two cases: (i) every bidder's type is sampled from an MRF, or (ii) every bidder's type is sampled from a Bayesnet. Our results can accommodate latent variables, that is, some of the variables/nodes of the MRF or Bayesnet are not observable in the samples. We show that the sample complexity for learning an  $\eta$-BIC and IR mechanism, whose revenue is at most $\varepsilon$ less than the optimal revenue achievable by any $\eta$-BIC and IR mechanisms, is polynomial in the size of the problem and scales gracefully with the parameters of the graphical models that generate the type distributions. If there is only a single bidder, the mechanism we learn will be exactly IC rather than approximately IC. We derive the sample complexity by combining our robustness result (Theorem~\ref{thm:approximation preserving}) with learnability results for MRFs and Bayesnets (Theorem~\ref{thm:sample complexity MRF finite alphabet} and~\ref{thm:sample complexity Bayesnets finite alphabet}). 

\begin{theorem}[Optimal Mechanism Design under MRF Item Distributions]\label{thm:auction MRF}
Consider the general mechanism design setting of Section 2. Recall that $\L$ is the Lipschitz constant of the valuations. Let $\hDD=\bigtimes_{i\in[n]}\hDD_{i}$, where each $\hDD_i$ is a $m$-dimensional distribution generated by an MRF $p_i$, as in Definition~\ref{def:MRF}, defined on a graph {with $N_i\geq m$ nodes,}  hyper-edges of size at most $d$, and $\supp(\hDD_i)\subseteq \Sigma^m\subseteq [0,H]^m$. When $N_i>m$, we say $\hDD_i$ is generated by an MRF with $N_i-m$ latent variables. We use $N$ to denote $\max_{i\in [n]}\{N_i\}$. 

For every $\varepsilon$, $\delta>0$,  and $\eta\leq \poly(n,m,\L,H,\varepsilon)$, we can learn, with probability at least $1-\delta$, an $\eta$-BIC w.r.t. $\hDD$ and IR mechanism $\hM$, whose revenue under truthful bidding is at most $\varepsilon$ smaller than the optimal revenue achievable by any $\eta$-BIC w.r.t. $\hDD$ and IR mechanism, given \begin{itemize}
	\item $\frac{poly\left(n,N^d, |\Sigma|^d, \LL, H, 1/\eta, \log(1/\delta)\right)}{\varepsilon^4}$ samples if  \textbf{ the alphabet $\Sigma$ is finite}; when the graph on which $p_i$ is defined is known for each bidder $i$, then $\frac{poly\left(n,N, \kappa, |\Sigma|^d, \L, H,1/\eta, \log(1/\delta)\right)}{\varepsilon^4}$-many samples suffice, where $\kappa$ is an upper bound on the number of edges in all the graphs;
	\item  ${{\rm poly}\left(n,N^{d^2}, {\left ( H\over \varepsilon\right)}^d,\CC^d, \LL,1/\eta, \log (1/\delta) \right) }$ samples if \textbf{the alphabet $\Sigma=[0,H]$}, and the log potentials $\phi_v^{p_i}(\cdot) \equiv \log \left(\psi_v^{p_i}(\cdot) \right)$ and $\phi_e^{p_i}(\cdot) \equiv \log \left(\psi_e^{p_i}(\cdot) \right)$ for every node $v$ and every edge $e$ are $\CC$-Lipschitz w.r.t. the $\ell_1$-norm, for every bidder $i$; when the graph on which $p_i$ is defined is known for each bidder $i$, then ${poly\left(n,N, \kappa^d, {\left ( H\over \varepsilon\right)}^d, \CC^d,\L, 1/\eta, \log(1/\delta)\right)}$-many samples suffice, where $\kappa$ is an upper bound on the number of edges in all the graphs.
\end{itemize}

If $n=1$, the mechanism $\hM$ will be IR and IC, and $$\rev\left( \hM, \hDD\right)\geq  \left(1-\sqrt{\eta}\right)\cdot \opt\left(\hDD\right)-\varepsilon-\sqrt{\eta}.$$

\end{theorem}

\begin{theorem}[Optimal Mechanism Design under Bayesnet Item Distributions]\label{thm:auction Bayesnet}
Consider the general mechanism design setting of Section 2. Recall that $\L$ is the Lipschitz constant of the valuations. Let $\hDD=\bigtimes_{i\in[n]}\hDD_{i}$, where each $\hDD_i$ is a $m$-dimensional distribution generated by a Bayesnet $p_i$, as in Definition~\ref{def:Bayesnet}, defined on a DAG with $N_i\geq m$ nodes, in-degree at most $d$, and $\supp(\hDD_i)\subseteq \Sigma^m\subseteq [0,H]^m$.  When $N_i>m$, we say $\hDD_i$ is generated by an MRF with $N_i-m$ latent variables. We use $N$ to denote $\max_{i\in [n]}\{N_i\}$.  

For every $\varepsilon$, $\delta>0$,  and $\eta\leq \poly(n,m,\L,H,\varepsilon)$, we can learn, with probability at least $1-\delta$, an $\eta$-BIC w.r.t. $\hDD$ and IR mechanism $\hM$, whose revenue under truthful bidding is at most $\varepsilon$ smaller than the optimal revenue achievable by any $\eta$-BIC w.r.t. $\hDD$ and IR mechanism, with
\begin{itemize}
	\item 
	$\poly\left(n,d, N, |\Sigma|^{d+1}, \LL, H, 1/\eta, 1/\varepsilon, \log (1/\delta) \right)$ samples if  \textbf{ the alphabet $\Sigma$ is finite};
	\item 
	${\poly\left( n, d^{d+1}, N^{d+1}, ({H\CC\over \varepsilon})^{d+1}, \LL, 1/\eta, \log (1/\delta)\right)}$ samples if \textbf{the alphabet $\Sigma=[0,H]$}, and for every $p_i$, the conditional probability of every node $v$ is $\CC$-Lipschitz in the $\ell_1$-norm (see Theorem~\ref{thm:sample complexity Bayesnets finite alphabet} for the definition). 

\end{itemize}

If $n=1$, the mechanism $\hM$ will be IR and IC, and $$\rev\left( \hM, \hDD\right)\geq  \left(1-\sqrt{\eta}\right)\cdot \opt\left(\hDD\right)-\varepsilon-\sqrt{\eta}.$$

\end{theorem}

\subsection{Sample Complexity for Learning MRFs and Bayesnets}\label{sec:sample complexity MRF and Bayesnet}
In this section, we present the sample complexity of learning an MRF or a Bayesnet. Our sample complexity scales gracefully with the maximum size of hyperedges in an MRF and indegree in a Bayesnet. Furthermore, our results hold even in the presence of latent variables, where we can only observe the values of $k$ variables, out of the total $|V|$ variables, in a sample.
\begin{theorem}[Learnability of MRFs in Total Variation and Prokhorov Distance] \label{thm:sample complexity MRF finite alphabet}
Suppose we are given sample access to an MRF $p$, as in Definition~\ref{def:MRF}, defined on an unknown graph with hyper-edges of size at most $d$.
\begin{itemize}
	\item \textbf{Finite alphabet $\Sigma$}: Given ${{\rm poly}\left(|V|^{d}, |\Sigma|^d, \log({1\over \varepsilon})\right) \over \varepsilon^2}$ samples from $p$ we can learn some MRF~$q$ whose hyper-edges also have size at most $d$ such that $\norm{p-q}_{TV} \le \varepsilon$. If the graph on which $p$ is defined is known, then ${{\rm poly}\left(|V|,|E|, |\Sigma|^d, \log({1\over \varepsilon})\right) \over \varepsilon^2}$-many samples suffice. Moreover, the polynomial dependence of the sample complexity on $|\Sigma|^d$ cannot be improved, and the dependence on $\varepsilon$ is tight up to ${\rm poly}(\log{1 \over \varepsilon})$ factors.

	\item \textbf{Alphabet $\Sigma=[0,H]$:} If the log potentials $\phi_v(\cdot) \equiv \log \left(\psi_v(\cdot) \right)$ and $\phi_e(\cdot) \equiv \log \left(\psi_e(\cdot) \right)$ for every node~$v$ and every edge $e$ are $\CC$-Lipschitz w.r.t. the $\ell_1$-norm, then given ${{\rm poly}\left(|V|^{d^2}, {\left ( H\over \varepsilon\right)}^d,\CC^d \right) }$  samples from $p$ we can learn some MRF~$q$ whose hyper-edges also have size at most $d$ such that $\norm{p-q}_{P} \le \varepsilon$. If the graph on which $p$ is defined is known, then ${{\rm poly}\left(|V|, |E|^{d}, {\left ( H\over \varepsilon\right)}^d,\CC^d \right) }$-many samples suffice.
\end{itemize} 
Our sample complexity bounds can be easily extended to MRFs with {\em latent variables}, i.e.~to the case where some subset $V' \subseteq V$ of the variables are observable in each sample we draw from $p$. Suppose $k=|V'| \le |V|$ is the number of observable variables. In this case,
for all settings discussed above, our sample complexity bound only increases by a $k\cdot \log |V|$ multiplicative factor. 
\end{theorem}

\begin{theorem}[Learnability of Bayesnets in Total Variation and Prokhorov Distance] \label{thm:sample complexity Bayesnets finite alphabet}
Suppose we are given sample access to a Bayesnet $p$, as in Definition~\ref{def:Bayesnet}, defined on an unknown DAG with in-degree at most $d$.
\begin{itemize}
	\item  \textbf{Finite alphabet $\Sigma$}: Given $O\left({ d |V|\log |V| + |V| \cdot |\Sigma|^{d+1} \log\left({|V||\Sigma| \over \varepsilon}\right)\over \varepsilon^2}\right)$-many samples from $p$ we can learn some Bayesnet~$q$ defined on a DAG whose in-degree is also bounded by $d$ such that $\norm{p-q}_{TV} \le \varepsilon$. If the graph on which $p$ is defined is known, then $O\left({ |V| \cdot |\Sigma|^{d+1} \log\left({|V||\Sigma| \over \varepsilon}\right)\over \varepsilon^2}\right)$-many samples suffice. Moreover, the dependence of the sample complexity on $|\Sigma|^{d+1}$ and $1\over \varepsilon$  is tight up to logarithmic factors.
	\item \textbf{Alphabet $\Sigma=[0,H]$:} Suppose that the conditional probability distribution of every node $v$ is $\CC$-Lipschitz in the $\ell_1$-norm, that is, $\norm{p_{X_v | X_{\Pi_v}=\sigma} - p_{X_v | X_{\Pi_v}=\sigma'}}_{TV} \le\CC\cdot\norm{\sigma-\sigma'}_1$, $\forall v$ and $\sigma,\sigma'\in \Sigma^{\Pi_v}$. Then, given $O\left({{ d |V|\log |V|  + |V| \cdot \left({H|V|d\CC\over \varepsilon}\right)^{d+1} \log\left({|V|{Hd\CC} \over \varepsilon}\right)} \over \varepsilon^2} \right)$ -many samples from $p$, we can learn some Bayesnet~$q$ defined on a DAG whose in-degree is also bounded by $d$ such that $\norm{p-q}_{P} \le \varepsilon$. If the graph on which $p$ is defined is known, then $O\left({|V| \cdot \left({H|V|d\CC\over \varepsilon}\right)^{d+1} \log\left({|V|{Hd\CC} \over \varepsilon}\right) \over \varepsilon^2}\right)$ -many samples suffice. 
	 \end{itemize}
	Our sample complexity bounds can be easily extended to Bayesnets with {\em latent variables}, i.e.~to the case where some subset $V' \subseteq V$ of the variables are observable in each sample we draw from $p$. Suppose $k=|V'| \le |V|$ is the number of observable variables. In this case,
for all settings discussed above, our sample complexity bound only increases by a $k\cdot \log |V|$ multiplicative factor. 
\end{theorem}

In our proof of Theorem~\ref{thm:sample complexity MRF finite alphabet}, we first carefully construct an $\varepsilon$-net over all MRFs with hyperedges of size at most $d$  in either total variation distance or Prokhorov distance, then apply a tournament-style  density estimation algorithm~\cite{devroye2012combinatorial,DaskalakisK14,AcharyaJOS14} to learn a distribution from the $\varepsilon$-net that is at most $O(\varepsilon)$ away from the true distribution using polynomially many samples. Our proof of Theorem~\ref{thm:sample complexity Bayesnets finite alphabet} follows  a similar recipe. The main difference is how we construct the $\varepsilon$-net over all Bayesnets with in-degree at most $d$. Both proofs are presented in Appendix~\ref{appx:sample}.

\bibliographystyle{plain}	
\bibliography{Yang}
\newpage
\appendix
\section{Further Related Work} \label{sec:related}
As described earlier, most prior work on learning multi-item auctions follows a PAC-learning approach, bounding the statistical complexity of classes of mechanisms that are (approximately) optimal for the setting of interest. The statistical complexity measures that are used for this purpose are the standard notions of pseudodimension, which generalizes VC dimension to real valued functions, and Rademacher complexity. In particular, Morgenstern and Roughgarden~\cite{MorgensternR16} and Syrgkanis~\cite{Syrgkanis17} bound respectively the pseudodimension and Rademacher complexity of simple classes of mechanisms that have been shown in the literature to contain approximately optimal mechanisms 
in  multi-item multi-bidder settings satisfying item-independence~\cite{ChawlaHMS10, BabaioffILW14, Yao15, CaiDW16, CaiZ17}. 
The classes of mechanisms studied by these works
contain
approximately 
optimal mechanisms in multi-item settings with item-independence and either multiple unit-demand/additive bidders,
or a single subadditive bidder.
More powerful classes of simple mechanisms are also known in the literature. The state-of-the-art is the \emph{sequential two-part tariff mechanism}  considered by Cai and Zhao~\cite{CaiZ17}, which is shown to approximate the optimal revenue in multi-item settings even with multiple bidders whose valuations are fractionally subadditive, again under item-independence. Unfortunately, both the pseudodimension and the empirical Rademacher complexity of sequential two-part tariff mechanisms are already exponential even in two bidder settings, making these measures unsuitable tools for showing the learnability of two-part tariff mechanisms. 

An important feature of the afore-described works is that bounding the pseudo-dimension or empirical Rademacher complexity of mechanism classes is oblivious to the structure in the distribution. Hence, while the mechanisms considered in these works are only approximately optimal under item-independence, the independence cannot be exploited. In contrast to empirical Rademacher complexity, Rademacher complexity {\em is} sensitive to the underlying distribution, but bounds exploiting the structure of the  distribution are not easy to obtain.
%
%
This observation motivated another line of work which
	heavily exploits the structure of the distributions of interest to choose both the class of mechanisms {\em and} the statistical complexity measure to bound their learnability. 
	So far, this approach has only been applied to settings satisfying item-independence.
	Indeed, Cai and Daskalakis~\cite{CaiD17} propose a statistical complexity measure that is tailored to product distributions, and use their new measure to establish  learnability of sequential two-part tariff mechanisms under item-independence. Gonczarowski and Weinberg~\cite{GonczarowskiW18}  choose a finite class of mechanisms so that an up-to-$\varepsilon$ optimal mechanism is guaranteed to exist in the class. 
	For item-independent distributions, the size of this class is only singly exponential implying polynomial sample learnability. Unfortunately, the size becomes doubly exponential for correlated items turning the sample complexity exponential. 
	
Finally,
Goldner and Karlin~\cite{GoldnerK16} do not use a PAC-learning based approach. They show how to learn approximately optimal auctions in the multi-item multi-bidder setting with additive bidders using only one sample from each bidder's distribution, assuming that it is regular and independent across items. Their approach is tailored for a mechanism designed by Yao~\cite{Yao15} and does not apply to broader settings.

\section{Additional Preliminaries}\label{sec:add_prelim}
\begin{definition}[Total Variation Distance]\label{def:TV}
The {\em total variation distance} between two probability measures $P$ and $Q$ on a $\sigma$-algebra $\cal F$ of subsets of some sample space $\Omega$, denoted $|| P-Q||_{TV}$, is defined as 
$$\sup_{E \in {\cal F}}\left\lvert P({E})-Q({E})\right\rvert.$$
\end{definition}

\begin{definition}[Kolmogorov Distance] \label{def:Kolm}
The {\em Kolmogorov distance} between two distributions $P$ and $Q$ over $\mathbb{R}$, denoted $\norm{P-Q}_K$, is defined as 
$$\sup_{x \in \mathbb{R}}\left\lvert\Pr_{X \sim P}[X \le x]-\Pr_{X \sim Q}[X \le x]\right\rvert.$$
\end{definition}

\begin{definition}[L\'evy Distance]\label{def:levy}
Let $\DD_1$ and $\DD_2$ be two probability distributions on $\mathbb{R}$ with cumulative distribution functions $F$ and $G$ respectively. Then we denote their L\'evy distance by

$$\norm{\DD_1-\DD_2}_L = \inf\left\{\varepsilon>0 : F(x-\varepsilon)-\varepsilon  \leq G(x) \leq F(x+ \varepsilon)+\varepsilon,  \ \forall x \in \mathbb{R} \right\}$$
\end{definition}

\paragraph{Multi-item Auctions:} We focus on revenue maximization in the combinatorial auction with \textbf{$n$ bidders} and \textbf{$m$ heterogenous items}. 

The outcomes of the auction lie in $X \subseteq \{0,1\}^{n \cdot m}$ such that for any allocation $x \in X$, $x_{i,j}$ is the probability that bidder $i$ receives item $j$. Formally, $X=\left\{ (x_{i,j})_{i\in[n], j\in [m]} \in \{0,1\}^{nm} \mid \forall j:\sum_{i=1}^n x_{i,j} \leq 1\right\}$. Each bidder $i\in [n]$ has a valuation function  $v_i(\cdot): X \rightarrow \mathbb{R}$ that maps an allocations of items to a real number. In this paper, we assume the function $v_i(\cdot)$ is parametrized by $(v_{i,1},\ldots,v_{i,m})$, where $v_{i,j}$ is bidder $i$'s value for item $j$. We will refer to the vector $(v_{i,1},\ldots,v_{i,m})$ as bidder $i$'s type, and we assume that \emph{each bidder's type is drawn independently} from some distribution.~\footnote{We will not explicitly write bidder $i$'s valuation as $v_{i,\bold{v_i}}(\cdot)$ where $\bold{v_i}=(v_{i,1},\ldots,v_{i,m})$.} Throughout this paper, we assume all bidders types lie in $[0,H]^m$.

\paragraph{Mechanisms, Payments, and Utility:} We use $p=(p_1,\ldots,p_n)$ to specify the payments for the bidders. Given some prices $p=(p_1,\ldots,p_n)$, allocation $x$ and type $v_i$, denote the quasilinear utility of bidder $i\in [n]$ by $u_i(v_i,(x,p)) = v_i(x)-p_i$.  Let $M=\left(x(\cdot),p(\cdot)\right)$ be a mechanism with allocation rule $x(\cdot)$ and payment rule $p(\cdot)$. For any input bid vector $b=(b_1,\ldots,b_n)$, the allocation rule outputs a distribution over allocations $x(b)\in \Delta(X)$ and payments $p(b) = \left(p_1(b),\ldots, p_n(b) \right)$. Then $u_i(v_i,M(b)) = v_i(x(b)) - p_i(b)$. If bidder $i$'s type is $v_i$, then her utility under input bid vector $b$ is $u_i\left(v_i, M(b)\right)=\E\left[v_i\left(x(b)\right)-p_i(b)\right]$, where the expectation is over the randomness of the allocation and payment rule. 
\paragraph{$\varepsilon$-Incentive Compatible and Individually Rational:}
\begin{itemize}
	\item \emph{Ex-post Individually Rational (IR):} $M$ is \emph{IR} if for all types $v \in [0,H]^{n \cdot m}$ and all bidders $i\in [n]$, $$u_i(v_i,M(v_i,v_{-i})) \geq 0.$$
	\item \emph{$\varepsilon$-Dominant Strategy Incentive Compatible ($\varepsilon$-DSIC):}  if for all $i\in [n]$, $v\in [0,H]^{n\cdot m}$ and potential misreports $v'_i \in [0,H]^m$ of bidder $i$, $u_i(v_i,M(v_i,v_{-i}))\geq u_i(v_i,M((v'_i,v_{-i}))) -\varepsilon$.	A mechanism is DSIC if it is $0$-DSIC.
	\item  \emph{$\varepsilon$-Bayesian Incentive Compatible ($\varepsilon$-BIC):} if bidders draw their values from some distribution $\FF = (\FF_1,\ldots,\FF_n)$, then 
define $M$ to be \emph{$\varepsilon$-BIC with respect to $\FF$} if 
$$\mathbb{E}_{v_{-i} \sim \FF_{-i}}[u_i(v_i,M(v_i,v_{-i}))] \geq \mathbb{E}_{v_{-i} \sim \FF_{-i}}[u_i(v_i,M(v'_i,v_{-i}))] -\varepsilon,$$ for all potential misreports $v'_i$, in expectation over all other bidders bid $v_{-i}$. A mechanism is BIC if it is $0$-BIC.
\end{itemize}

If there is only one bidder, the definition of DSIC coincides with the definition of BIC, and we simply use $\varepsilon$-IC to describe the incentive compatibility of single bidder mechanisms.

In single-bidder case, there is a well known transformation, Lemma~\ref{lem:eps-IC to IC}, that maps any $\varepsilon$-IC mechanism to an IC mechanism with negligible revenue loss. To the best of our knowledge, the result is attributed Nisan in~\cite{ChawlaHK07, HartN13, GonczarowskiW18} and many other papers.

\begin{lemma}[Nisan, circa 2005]\label{lem:eps-IC to IC}
	Let $M$ be an IR and $\varepsilon$-IC mechanism for a single bidder, and $\DD$ be the bidder's type distribution. Modifying each possible allocation and payment pair by multiplying the payment by $1-\sqrt{\varepsilon}$ and letting the bidder choose the (modified) allocation and payment pair that maximizes her utility yields an IR and IC mechanism $M'$ with expected revenue at least $(1-\sqrt{\varepsilon})(\revT(M,\DD)-\sqrt{\varepsilon})$. Importantly, the modification does not require any knowledge of $\DD$. 
\end{lemma}
Interested readers can find a proof of Lemma~\ref{lem:eps-IC to IC} in~\cite{GonczarowskiW18}.

\paragraph{Up-to-$\varepsilon$ Optimal Mechanisms:} We say a mechanism $M$ is up-to-$\varepsilon$ optimal under distribution $\DD$, if $$\rev_T(M,\DD)\geq \opt(\DD)-\varepsilon.$$

\section{Missing Proofs from Section~\ref{sec:single}}\label{appx:single}

\begin{prevproof}{Theorem}{thm:single-kolmogorov}
	We prove the claim using a hybrid argument. We construct a collection of distributions, where $\DD^{(0)} = \DD$, $\DD^{(i)}= \hDD_1\times\cdots\times \hDD_i\times\DD_{i+1}\times \cdots\times \DD_n$ for all $1\leq i < n$, and $\DD^{(n)}=\hDD$. We first show the following claim
	\begin{claim}\label{clm:hybrid-single-kolmogorov}
		\begin{equation*}
		\opt\left (\DD^{(i)}\right)\geq \opt\left(\DD^{(i-1)}\right)-3H\cdot \varepsilon,
	\end{equation*}	for all $i\in[n]$.
	\end{claim}
	 	\begin{proof}
	 	W.l.o.g, we can assume the optimal mechanism for $\DD^{(i-1)}$ is a deterministic. We use $M=(x,p)$ to denote it. In particular, there exists a collection of monotone non-decreasing functions $\{\mu_j(\cdot)\}_{\{j\in[n]\}}$ such that $\mu_j: \supp\left(\DD_j^{(i-1)}\right)\mapsto \mathbb{R}$.  We extend the function $\mu_j(\cdot)$ to the whole interval $(-\infty,H]$. We slightly abuse notation and still call the extended function $\mu_j(\cdot)$. For any $z\in \supp\left(\DD_j^{(i-1)}\right)$, $\mu_j(x)$ remains the same. For any $z> \inf\supp\left(\DD_j^{(i-1)}\right) $, let $$\mu_j(z)=\sup\left\{\mu_j(w)\mid w\leq z \text{ and } w\in \supp\left(\DD_j^{(i-1)}\right)\right\}.$$ If $z\leq \inf\supp\left(\DD_j^{(i-1)}\right)$ and $\notin \supp\left(\DD_j^{(i-1)}\right)$, let $\mu_j(z)=-\infty$. 

	 	 Now we define a mechanism $M'=(x',p')$ for $\DD^{(i)}$ based on the extended $\{\mu_j(\cdot)\}_{\{j\in[n]\}}$. For every profile $v$, let the bidder $j^*$ with the highest positive $\mu_j(v_j)$ be the winner. If no bidder $j$ has positive $\mu_j(v_j)$, the item is unallocated. When there are ties, break the tie in alphabetical order. Clearly, the allocation rule is monotone. According to Myerson's payment identity, if a bidder wins the item, she should pay $\inf\{z \mid \text{$z$ is a winning bid}\}$.

	 	To prove the claim, we demonstrate the following two statements: for every fixed $v_{-i} $ \textbf{(A1:)} bidder $i$'s expected payments under $\DD^{(i)}$ and $\DD^{(i-1)}$ are within $O(H\cdot \varepsilon)$; \textbf{(A2:)} the total expected payments of all bidders except $i$ under $\DD^{(i)}$ and $\DD^{(i-1)}$ are within $O(H\cdot \varepsilon)$. We first prove A1.
	 	
	 	\paragraph{Proof of A1:} For every fixed $v_{-i}$, let $\ell^* = \argmax_{\ell\neq i } \mu_\ell(v_\ell)$. For bidder $i$ to win the item, $\mu_i(v_i)$ needs to be greater than $\mu_{\ell^*}(v_{\ell^*})$. Therefore, there exists a threshold $\theta(v_{-i})$ for every fixed $v_{-i}$, such that bidder $i$ wins the item iff $v_i\geq \theta(v_{-i})$.  
	 	Clearly, \begin{equation*}
 	\E_{v_i\sim \hDD_i}[p'_i(v_i,v_{-i})]=\theta(v_{-i})\cdot \Pr_{v_i\sim \hDD_i}\left[v_i\geq \theta(v_{-i})\right],
 	 \end{equation*}
 	 and
\begin{equation*}
 		\E_{v_i\sim \DD_i}[p_i(v_i,v_{-i})]=\theta(v_{-i})\cdot \Pr_{v_i\sim \DD_i}\left[v_i\geq \theta(v_{-i})\right].
 \end{equation*}
 Since $\norm{\DD_i-\hDD_i}_K\leq \varepsilon$, $\left\lvert \Pr_{v_i\sim \hDD_i}\left[v_i\geq \theta(v_{-i})\right]-\Pr_{v_i\sim \DD_i}\left[v_i\geq \theta(v_{-i})\right]\right\rvert\leq \varepsilon$, which implies that
\begin{align*}
	&\left\lvert\E_{v\sim \DD^{(i)}}[p'_i(v)]-\E_{v\sim \DD^{(i-1)}}[p_i(v)]\right\rvert\\
	\leq &~ \E_{v_{-i}\sim \DD^{(i)}_{-i}}\left[\left\lvert \E_{v_i\sim \hDD_i}[p'_i(v_i,v_{-i})]-\E_{v_i\sim \DD_i}[p_i(v_i,v_{-i})]\right\rvert \right]\\
	\leq &~ \E_{v_{-i}\sim \DD^{(i)}_{-i}}\left[\theta(v_{-i})\cdot \varepsilon\right]\\
	\leq &~ H\cdot\varepsilon
\end{align*}
This completes the argument for statement A1. Next, we prove statement A2.

	\paragraph{Proof of A2:}

Since there is only one item, only the winner $\ell^*$ has non-zero payment and $\sum_{\ell\neq i}p_\ell(\textbf{v})= p_{\ell^*}(v)$ for any $v_i$.  Our goal now is to bound the difference between $\E_{v_i\sim \DD_i}\left[p_{\ell^*}(v)\right]$ and $\E_{v_i\sim \hDD_i}\left[p'_{\ell^*}(v)\right]$. Note that \begin{equation*}
 	\E_{v_i\sim \DD_i}\left[p_{\ell^*}(v)\right]=\int_0^H \Pr_{v_i\sim \DD_i}[p_{\ell^*}(v) > t] dt.  
 \end{equation*}
When $\mu_{\ell^*}(t) \geq \mu_{\ell^*}(v_{\ell^*})$, $\Pr_{v_i\sim \DD_i}[p_{\ell^*}(v) > t] =0$, so we only consider the case where $\mu_{\ell^*}(t) < \mu_{\ell^*}(v_{\ell^*})$. Let $\alpha = \max_{\ell\neq i \text{ or } \ell^*}\mu_{\ell}(v_{\ell})$. $p_{\ell^*}(v) > t$ is equivalent to having  $\max\{\alpha, \mu_{i}(v_{i})\} > \mu_{\ell^*}(t)$ and $\mu_{i}(v_{i}) < \mu_{\ell^*}(v_{\ell^*})$ if $\ell^*> i$ (or $\mu_{i}(v_{i}) \leq \mu_{\ell^*}(v_{\ell^*})$ if $\ell^*< i$). Since $\mu_i(\cdot)$ is monotone, it is not hard to observe that this is equivalent to having $v_i$ lying in some interval that only depends on $v_{-i}$. Let the lower bound of the interval be $a(v_{-i})$ and the upper bound be $b(v_{-i})$. Similarly, we know \begin{equation*}
 	\E_{v_i\sim \hDD_i}\left[p_{\ell^*}(v)\right]=\int_0^H \Pr_{v_i\sim \hDD_i}[p'_{\ell^*}(v) > t] dt,
 \end{equation*}
 and $\Pr_{v_i\sim \hDD_i}[p'_{\ell^*}(v) > t] $ is also the probability that $v_i$ lies between $a(v_{-i})$ and $b(v_{-i})$. Since $\norm{\DD_i-\hDD_i}_K\leq \varepsilon$, $\left \lvert\Pr_{v_i\sim \DD_i}[p_{\ell^*}(v) > t] -\Pr_{v_i\sim \hDD_i}[p'_{\ell^*}(v) > t]\right\rvert\leq 2\varepsilon$ for all $t\in[0,H]$, and 
 \begin{equation*}
 	\left\lvert \E_{v_i\sim \DD_i}\left[p_{\ell^*}(v)\right]-\E_{v_i\sim \hDD_i}\left[p'_{\ell^*}(v)\right]\right\rvert\leq H\cdot 2\varepsilon.
 	 \end{equation*}

 	Combining statement (i) and (ii), we complete the proof.
 	\end{proof}
 	By Claim~\ref{clm:hybrid-single-kolmogorov}, it is clear that \begin{equation*}
 		\opt\left(\hDD\right)=\opt\left(\DD^{(n)}\right)\geq \opt\left(\DD^{(0)}\right)-3nH\cdot \varepsilon=\opt\left(\DD\right)-3nH\cdot \varepsilon
 	\end{equation*}
	\end{prevproof}
	
\notshow{

\begin{theorem}\label{thm:single-kolmogorov}
	For any buyer $i$, let $\DD_i$ and $\hDD_i$ be two arbitrary distributions supported on $(-\infty,H]$ such that $\norm{\DD_i-\hDD_i}_K\leq \varepsilon$. If $M$ is the optimal DSIC and IR mechanism w.r.t. $\DD$, then there exists a DSIC and IR mechanism $\widehat{M}$ such that
	\begin{equation*}
		\rev(\widehat{M},\hDD)\geq \rev(M,\DD)-2nH\cdot \varepsilon.
	\end{equation*}
	 where $\hDD = \bigtimes_{i=1}^n \hDD_i$.
\end{theorem}
\begin{proof}
	We prove the claim using a hybrid argument. We construct a collection of distributions, where $\DD^{(0)} = \DD$, $\DD^{(i)}= \hDD_1\times\cdots\times \hDD_i\times\DD_{i+1}\times \cdots\times \DD_n$ for all $1\leq i < n$, and $\DD^{(n)}=\hDD$. We first show the following claim
	\begin{claim}\label{clm:hybrid-single-kolmogorov}
		\begin{equation*}
		\rev(M,\DD^{(i)})\geq \rev(M,\DD^{(i-1)})-2H\cdot \varepsilon,
	\end{equation*}	for all $i\in[n]$.
	\end{claim}
	 	\begin{proof}
	 	To prove the claim, we demonstrate the following two statements: (i) bidder $i$'s expected payments under $\DD^{(i)}$ and $\DD^{(i-1)}$ are within $O(H\cdot \varepsilon)$; (ii) the total expected payments of all bidders except $i$ under $\DD^{(i)}$ and $\DD^{(i-1)}$ are within $O(H\cdot \varepsilon)$.
	 	
	 	\paragraph{Statement (i):}W.l.o.g. we can assume that  $M$ is a deterministic. Therefore, there exists a threshold $\theta(v_{-i})$ for every fixed $v_{-i}$, such that bidder $i$ wins the item w.p. $1$ iff $v_i\geq \theta(v_{-i})$.  
	 	Clearly, \begin{equation*}
 	\E_{v_i\sim \hDD_i}[p_i(v_i,v_{-i})]=\theta(v_{-i})\cdot \Pr_{v_i\sim \hDD_i}\left[v_i\geq \theta(v_{-i})\right],
 	 \end{equation*}
 	 and
\begin{equation*}
 		\E_{v_i\sim \DD_i}[p_i(v_i,v_{-i})]=\theta(v_{-i})\cdot \Pr_{v_i\sim \DD_i}\left[v_i\geq \theta(v_{-i})\right].
 \end{equation*}
 Since $\norm{\DD_i-\hDD_i}_K\leq \varepsilon$, $\left\lvert \Pr_{v_i\sim \hDD_i}\left[v_i\geq \theta(v_{-i})\right]-\Pr_{v_i\sim \DD_i}\left[v_i\geq \theta(v_{-i})\right]\right\rvert\leq \varepsilon$, which implies that
\begin{align*}
	&\left\lvert\E_{v\sim \DD^{(i)}}[p_i(v)]-\E_{v\sim \DD^{(i-1)}}[p_i(v)]\right\rvert\\
	\leq &~ \E_{v_{-i}\sim \DD^{(i)}_{-i}}\left[\left\lvert \E_{v_i\sim \hDD_i}[p_i(v_i,v_{-i})]-\E_{v_i\sim \DD_i}[p_i(v_i,v_{-i})]\right\rvert \right]\\
	\leq &~ \E_{v_{-i}\sim \DD^{(i)}_{-i}}\left[\theta(v_{-i})\cdot \varepsilon\right]\\
	\leq &~ H\cdot\varepsilon
\end{align*}
This completes the argument for statement (i). Next, we prove statement (ii).

	\paragraph{Statement (ii):}Since $M$ is a deterministic, DSIC, and IR mechanism for $\DD$, there exists a collection of monotone non-decreasing functions $\{\mu_i(\cdot)\}_{\{i\in[n]\}}$ such that $\mu_i: \supp(\DD_i)\mapsto \mathbb{R}$.  We extend the function $\mu_i(\cdot)$ to the whole interval $[0,H]$. We slightly abuse notation and still called the extended function $\mu_i(\cdot)$. For any $x\in \supp(\DD_i)$, $\mu_i(x)$ remains the same. If $0\notin \supp(\DD_i)$, let $\mu_i(0)=-\infty$. For any $x\in (0,H]$, let $\mu_i(x)=\mu_i(x')$, where $x' = \sup\left\{w\mid w\leq x \text{ and } w\in \supp(\DD_i)\right\}$. 
\footnote{\todo{Making the assumption that the support is compact, when the distribution is continuous.}} It is not hard to see that the extension of $M$ on $\DD^{(j)}$ is defined by the following allocation rule based on the extended $\{\mu_i(\cdot)\}_{\{i\in[n]\}}$. For every profile $v$, give the item to the bidder with the highest positive $\mu_i(v_i)$. If no bidder $i$ has positive $\mu_i(v_i)$, then the item is unallocated. When there are ties, break the tie in alphabetical order. Clearly, the allocation rule is monotone. According to Myerson's payment identity, if a bidder wins the item, she should pay $\inf\{x \mid \text{$x$ is a winning bid}\}$.

For every fixed $v_{-i}$, let $\ell^* = \argmax_{\ell\neq i } \mu_\ell(v_\ell)$. Since there is only one item, $\sum_{\ell\neq i}p_\ell(v)= p_{\ell^*}(v)$ for any $v_i$.  Our goal now is to bound the difference between $\E_{v_i\sim \DD_i}\left[p_{\ell^*}(v)\right]$ and $\E_{v_i\sim \hDD_i}\left[p_{\ell^*}(v)\right]$. Note that \begin{equation*}
 	\E_{v_i\sim \DD_i}\left[p_{\ell^*}(v)\right]=\int_0^H \Pr_{v_i\sim \DD_i}[p_{\ell^*}(v) > t] dt.  
 \end{equation*}
When $\mu_{\ell^*}(t) \geq \mu_{\ell^*}(v_{\ell^*})$, $\Pr_{v_i\sim \DD_i}[p_{\ell^*}(v) > t] =0$, so we only consider the case where $\mu_{\ell^*}(t) <\mu_{\ell^*}(v_{\ell^*})$. Let $\alpha = \max_{\ell\neq i \text{ or } \ell^*}\mu_{\ell}(v_{\ell})$. $p_{\ell^*}(v) > t$ is equivalent to having  $\max\{\alpha, \mu_{i}(v_{i})\} > \mu_{\ell^*}^{(i)}(t)$ and $\mu_{i}(v_{i}) < \mu_{\ell^*}(v_{\ell^*})$ if $\ell^*> i$ (or $\mu_{i}(v_{i}) \leq \mu_{\ell^*}(v_{\ell^*})$ if $\ell^*< i$). Since $\mu_i(\cdot)$ is monotone, it is not hard to observe that this is equivalent to having $v_i$ lie in some interval that only depends on $v_{-i}$. Let the lower bound be $a(v_{-i})$ and the upper bound be $b(v_{-i})$.\footnote{\todo{Whether the interval is open or close depends on the support and $\mu_i$.}} Similarly, we know \begin{equation*}
 	\E_{v_i\sim \hDD_i}\left[p_{\ell^*}(v)\right]=\int_0^H \Pr_{v_i\sim \hDD_i}[p_{\ell^*}(v) > t] dt,
 \end{equation*}
 and $\Pr_{v_i\sim \hDD_i}[p_{\ell^*}(v) > t] $ is also the probability  $v_i$ lie between $a(v_{-i})$ and $b(v_{-i})$. Since $\norm{\DD_i-\hDD_i}_K\leq \varepsilon$, $\left \lvert\Pr_{v_i\sim \DD_i}[p_{\ell^*}(v) > t] -\Pr_{v_i\sim \hDD_i}[p_{\ell^*}(v) > t]\right\rvert\leq \varepsilon$ for all $t\in[0,H]$, and 
 \begin{equation*}
 	\left\lvert \E_{v_i\sim \DD_i}\left[p_{\ell^*}(v)\right]-\E_{v_i\sim \hDD_i}\left[p_{\ell^*}(v)\right]\right\rvert\leq H\cdot \varepsilon.
 	 \end{equation*}

 	Combining statement (i) and (ii), we complete the proof.
 	\end{proof}
 	By Claim~\ref{clm:hybrid-single-kolmogorov}, it is clear that \begin{equation*}
 		\rev(M,\hDD)=\rev(M,\DD^{(n)})\geq \rev(M, \DD^{(0)})-2nH\cdot \varepsilon=\rev(M, \DD)-2nH\cdot \varepsilon
 	\end{equation*}
	\end{proof}
	
	}

	\notshow{With Theorem~\ref{thm:single-kolmogorov}, we are ready to prove Theorem~\ref{thm:single-levy}. We define two auxiliary distributions $\uDD$ and $\lDD$ that will be crucial in our proof. 	\begin{definition}\label{def:best-worst-distributions}
		For every $i$, we define $\uDD_i$ and $\lDD_i$ based on $\DD_i$. $\uDD_i$ is supported on $[0,H+\varepsilon]$, and its CDF is defined as $F_{\uDD_i}(x) = \max\left\{F_{\DD_i}(x-\varepsilon)-\varepsilon,0\right\}$. $\lDD_i$ is supported on $[-\varepsilon,H]$, and its CDF is defined as $F_{\lDD_i}(x) = \min\left\{F_{\DD_i}(x+\varepsilon)+\varepsilon,1\right\}$.
	\end{definition}
We provide a more intuitive interpretation of $\uDD_i$ and $\lDD_i$ here. To obtain $\uDD_i$, we first shift all values in $\DD_i$ to the right by $\varepsilon$, then we move the bottom $\varepsilon$ probability to $H+\varepsilon$. To obtain $\lDD_i$, we first shift all values in $\DD_i$ to the left by $\varepsilon$, then we move the top $\varepsilon$ probability to $-\varepsilon$. It is not hard to see that both $\uDD_i$ and $\lDD_i$ are still in the $\varepsilon$-ball around $\DD$ in L\'{e}vy distance. More importantly, $\uDD_i$ and $\lDD_i$ are the ``best'' and ``worst'' distributions in the $\varepsilon$-ball. To make the statement formal, we need the definition of first-order-stochastic-dominance. 

\begin{definition}[First-Order Stochastic Dominance]
	We say distribution $B$ \emph{first-order stochastically dominates} $A$ iff $F_B(x)\leq F_A(x)$ for all $x\in \mathbb{R}$. We use $A\preccurlyeq B$ to denote that distribution $B$ first-order stochastically dominates distribution $A$. If $\mathbf{A}=\times_{i=1}^n A_i$ and $\mathbf{B}=\times_{i=1}^n B_i$ are two product distributions, and $A_i\preccurlyeq B_i$ for all $i\in[n]$, we slightly abuse the notation $\preccurlyeq$ to write $\mathbf{A}\preccurlyeq \mathbf{B}$.
\end{definition}

\begin{lemma}\label{lem:best-worst-stochastic-dominance}
	For any $\hDD_i$, such that $\norm{\hDD_i-\DD_i}_L\leq \varepsilon$, we have $$\lDD_i \preccurlyeq\hDD_i\preccurlyeq \uDD_i.$$ \end{lemma}
\begin{proof}
	It follows from the definition of L\'{e}vy distance and Definition~\ref{def:best-worst-distributions}. For any $x$, \begin{equation*}
		F_{\hDD_i}(x) \in [F_{\DD_i}(x-\varepsilon)-\varepsilon, F_{\DD_i}(x+\varepsilon)+\varepsilon].
	\end{equation*}
	Clearly, $0\leq F_{\hDD_i}(x) \leq 1$, so we have \begin{equation*}
 	F_{\uDD_i}(x) \leq F_{\hDD_i}(x) \leq F_{\lDD_i}(x)
 \end{equation*}
for all $x$.
\end{proof}}

\begin{prevproof}{Lemma}{lem:LB and UB are close}
	For every $i\in[n]$, we construct two extra distributions $\utDD_i$ and $\ltDD_i$ as follows.  $\utDD_i$ is supported on $[\varepsilon,H+\varepsilon]$, and its CDF is defined as $F_{\utDD_i}(x) = F_{\DD_i}(x-\varepsilon)$. $\ltDD_i$ is supported on $[-\varepsilon,H-\varepsilon]$, and its CDF is defined as $F_{\lDD_i}(x) = F_{\DD_i}(x+\varepsilon)$. In other words, $\utDD_i$ is the distribution by shifting all values in $\DD_i$ to the right by $\varepsilon$, and $\ltDD_i$ is the distribution by shifting all values in $\DD_i$ to the left by $\varepsilon$.
	
	\begin{claim}\label{clm:two shifted distributions have similar revenues}
	Let M be any DSIC and IR mechanism for $\utDD=\bigtimes_{i=1}^n \utDD_i$, there exists a DSIC and IR mechanism $M'$ for $\ltDD=\bigtimes_{i=1}^n \ltDD_i$ such that
	\begin{equation*}
		\rev(M',\ltDD)\geq \rev(M,\utDD) - 2\varepsilon.
	\end{equation*}
	\end{claim}
\begin{proof}
Based on the construction of $\utDD$ and $\ltDD$, we can couple the two distributions so that whenever we draw a value profile $ {v}=(v_1,\ldots, v_n)$ from $\utDD$, we also draw a value profile $ {v-2\varepsilon} = (v_1-2\varepsilon,\ldots, v_n-2\varepsilon)$ from $\ltDD$. Given mechanism $M=(x,p)$, we construct mechanism $M'$ as follows. For every bid profile $ {v}$, we offer bidder $i$ the item with probability $x_i ( {v+2\varepsilon})$ and asks her to pay $p_i( {v+2\varepsilon})-2\varepsilon\cdot x_i ( {v+2\varepsilon})$. Why is $M'$ a DSIC and IR mechanism? For any value profile $ {v}$ and any bidder $i$, her utility for reporting the true value is $$(v_i+2\varepsilon) \cdot x_i( {v+2\varepsilon})-p_i( {v+2\varepsilon}),$$ and her utility for misreporting to $v_i'$ is $$(v_i+2\varepsilon) \cdot x_i\left( (v_i',v_{-i})+ {2\varepsilon}\right)-p_i((v_i',v_{-i})+{2\varepsilon}).$$ Now consider a different scenario, where we run mechanism $M$ and  all the other bidders report $v_{-i}+ {2\varepsilon}$. The former is bidder $i$'s utility in $M$ when her true value is $v_i+2\varepsilon$ and she reports truthfully. The latter is bidder $i$'s utility in $M$ when she lies and reports $v'_i+2\varepsilon$. As $M$ is a DSIC and IR mechanism, $(v_i+2\varepsilon) \cdot x_i( {v+2\varepsilon})-p_i( {v+2\varepsilon})$ is nonnegative and greater than $(v_i+2\varepsilon) \cdot x_i\left( (v_i',v_{-i})+ {2\varepsilon}\right)-p_i((v_i',v_{-i})+ {2\varepsilon})$. Thus, $M'$ is also a DSIC and IR mechanism. Since there is only one item for sale, $\sum_i x_i( {v+2\varepsilon})\leq 1$. For every value profile $ {v}$, the total payment in $M'$ for this profile is at most $2\varepsilon$ smaller than the total payment in $M$ for value profile $ {v+2\varepsilon}$. Therefore, $\rev(M',\ltDD)\geq \rev(M,\utDD) - 2\varepsilon$.
\end{proof}
An easy corollary of Claim~\ref{clm:two shifted distributions have similar revenues} is that \begin{equation}\label{eq:lt and ut}
 	\opt(\ltDD)\geq \opt(\utDD)-2\varepsilon.
 \end{equation}
 Next we will use this corollary and Theorem~\ref{thm:single-kolmogorov} to prove our claim. Note that $\norm{\ltDD_i-\lDD_i}_K\leq \varepsilon$ and $\norm{\utDD_i-\uDD_i}_K\leq \varepsilon$ for all $i\in[n]$. Theorem~\ref{thm:single-kolmogorov} implies that \begin{equation}\label{eq:lt and l}
 	\opt(\lDD)\geq \opt(\ltDD)-3nH\cdot\varepsilon
 \end{equation}
 and \begin{equation}\label{eq:ut and u}
 \opt(\utDD)\geq \opt(\uDD)-3n(H+\varepsilon)\cdot\varepsilon.
  \end{equation}
  Chaining inequalities \eqref{eq:lt and l},~\eqref{eq:lt and ut}, and~\eqref{eq:ut and u}, we have $$\opt(\lDD)\geq \opt(\uDD)-(6nH+3n\varepsilon+2)\cdot \varepsilon.$$
\end{prevproof}

\section{Missing Details from Section~\ref{sec:prokhorov multi}}\label{appx:multi}

\notshow{

\begin{theorem}[TV-Robustness for Multi-item Auctions]\label{thm:multi-item tv-robustness}
	 	Given $\DD=\bigtimes_{i=1}^n \DD_i$, where $\DD_i $ is a $m$-dimensional distribution for all $i\in[n]$, and a BIC and IR mechanism $M$ w.r.t. $\DD$. We use $\hDD=\bigtimes_{i=1}^n \hDD_i$ to denote the true but unknown type distribution satisfying $\norm{\DD_i-\hDD_i}_{TV}\leq \varepsilon$ for all $i\in[n]$. There is a deterministic algorithm, oblivious to $\hDD$, that constructs a mechanism $\hM$ such that:
	 	 \begin{enumerate}
	 		\item $\hM$ is a $O(n m \LL H \varepsilon)$-BIC and IR mechanism for $\hDD$;
	 		\item  the expected revenue of $\hM$ under truthful bidding is $\revT(\hM,\hDD)\geq \rev(M,\DD)-O\left(n^2 m \LL H \varepsilon\right).$
	 	\end{enumerate}
	 	 \end{theorem}
}

\subsection{Proof of Theorem~\ref{lem:transformation under TV}}\label{sec:TV robustness}
	 \begin{prevproof}{Theorem}{lem:transformation under TV}
  
  	 	We first construct a mechanism $M_2$, and we show that $M_2$ is $(2m\LL H\rho+\eta)$-BIC w.r.t. $\hFF$ and IR. We first define a mapping $\tau_i$ for every bidder $i$:\begin{equation}\label{eq:mapping}
    \tau_i(v_i)=
    \begin{cases}
      v_i, & \text{if}\ v_i\in \supp (\FF_i)\\
      \argmax_{z\in \supp (\FF_i)\cup \perp}\E_{b_{-i}\sim \FF_{-i}}\left[ u_i(v_i,M_1(z,b_{-i}))\right], & \text{otherwise.}
    \end{cases}
  \end{equation}

Note that $\E_{b_{-i}\sim \FF_{-i}}\left[ u_i(v_i,M_1(\perp,b_{-i}))\right]=0$.
For any bid profile $v$, we use $\tau(v)$ to denote the vector $(\tau_1(v_1),\ldots,\tau_n(v_n))$. Let $x(\cdot)$ and $p(\cdot)$ be the allocation and payment rule for $M_1$. We now define $M_2$'s allocation rule $x'(\cdot)$ and payment rule $p'(\cdot)$. For any bid profile $v$, $x'(v) = x(\tau(v))$. If $\tau_i(v_i)\neq v_i$ and $\tau_\ell(v_\ell)\neq \perp$ for all bidders $\ell\in[n]$, then $$p_i'(v_i,v_{-i})=v_i(x(\tau(v)))\cdot{\E_{b_{-i}\sim \FF_{-i}}\left[ p_i(\tau_i(v_i),b_{-i})\right]\over \E_{b_{-i}\sim \FF_{-i}}\left[ v_i\left(x\left(\tau_i(v_i),b_{-i}\right)\right)\right]}.$$ Otherwise, $p_i'(v)=p_i(\tau(v))$.

An important property of $p'(\cdot)$ is that $\E_{b_{-i}\sim \FF_{-i}}\left[ p'_i(v_i,b_{-i})\right]=\E_{b_{-i}\sim \FF_{-i}}\left[ p_i(\tau_i(v_i),b_{-i})\right]$ for any $v_i$. We first argue that $M_2$ is IR.

\paragraph{\textbf{$M_2$ is IR:}}
 For any bidder $i$ and any bid profile $v$, if any of $\tau_\ell(v_\ell)=\perp$ bidder $i$'s utility is clearly $0$. So we only need to consider the case where $\tau_\ell(v_\ell)\neq \perp$ for all $\ell \in [n]$. \begin{itemize}
 	\item If $v_i=\tau_i(v_i)$, bidder $i$'s utility is $v_i(x(v_i,\tau_{-i}(v_{-i})))-p_i(v_i,\tau_{-i}(v_{-i}))=u_i\left(v_i, M_1(v_i, \tau_{-i}(v_{-i}))\right)$, which is non-negative as $v_i\in \supp(\FF_i)$ and $M_1$ is IR.
 	\item  If $v_i\neq\tau_i(v_i)$, since $\tau_i(v_i)\neq \perp$ by our assumption, $$\E_{b_{-i}\sim \FF_{-i}}\left[ v_i\left(x\left(\tau_i(v_i),b_{-i}\right)\right)\right]-\E_{b_{-i}\sim \FF_{-i}}\left[ p_i(\tau_i(v_i),b_{-i})\right] =\E_{b_{-i}\sim \FF_{-i}}\left[ u_i(v_i,M_1(\tau_i(v_i),b_{-i}))\right],$$ which is non-negative due to the definition of $\tau_i(\cdot)$. Equivalently, this means that $${\E_{b_{-i}\sim \FF_{-i}}\left[ p_i(\tau_i(v_i),b_{-i})\right]\over \E_{b_{-i}\sim \FF_{-i}}\left[ v_i\left(x\left(\tau_i(v_i),b_{-i}\right)\right)\right]}\leq 1$$ and $p_i'(v_i,v_{-i})\leq v_i(x(\tau(v)))=v_i(x'(v))$.
 \end{itemize}

  Next, we argue that $M_2$ is $(2m\LL H\rho+\eta)$-BIC.

  \paragraph{\textbf{$M_2$ is $(2m\LL H\rho+\eta)$-BIC:}}
Consider any bidder $i$ and any type $v_i$ and $t$, we first bound the difference between $\E_{b_{-i}\sim \FF_{-i}} \left[u_i(v_i,M_1(\tau_i(t),b_{-i}))\right]$ and $\E_{\hb_{-i}\sim \hFF_{-i}} \left[u_i(v_i,M_2(t,\hb_{-i}))\right]$. Note that \begin{equation}\label{eq:same interim utility}
	\E_{b_{-i}\sim \FF_{-i}} \left[u_i(v_i,M_1(\tau_i(t),b_{-i}))\right]=\E_{b_{-i}\sim \FF_{-i}} \left[u_i(v_i,M_2(t,b_{-i}))\right].
\end{equation}
 This is because  $x'(t,b_{-i})=x(\tau_i(t),b_{-i})~\forall b_{-i}\in \supp(\FF_{-i})$ and $\E_{b_{-i}\sim \FF_{-i}}\left[ p'_i(t,b_{-i})\right]=\E_{b_{-i}\sim \FF_{-i}}\left[ p_i(\tau_i(t),b_{-i})\right]$

  Since $\norm{\hFF_j-\FF_j}_{TV}= \varepsilon_j$, we can couple $b_{-i}$ and $\hb_{-i}$ so that $$\Pr[b_{-i}\neq \hb_{-i}]\leq \rho.$$ Clearly, when $b_{-i}=\hb_{-i}$,  ${ u_i(v_i,M_2(t,b_{-i}))=u_i(v_i,M_2(t,\hb_{-i}))}$. When $b_{-i}\neq \hb_{-i}$, $${\left|u_i(v_i,M_2(t,b_{-i}))-u_i(v_i,M_2(t,\hb_{-i}))\right|\leq m\LL H},$$ as $u_i(v_i,M_2(t,b'_{-i}))\in [0,m\LL H]$ for any $b_{-i}'$.  Hence, for any $v_i$ and $t$  
    \begin{equation}\label{eq:under F vs. hF in M2}
 	\left \lvert\E_{b_{-i}\sim \FF_{-i}} \left[u_i(v_i,M_2(t,b_{-i}))\right] - \E_{\hb_{-i}\sim \hFF_{-i}} \left[u_i(v_i,M_2(t,\hb_{-i}))\right]\right\rvert \leq m\LL H\rho.
 \end{equation}
  
  Combining Inequality~\eqref{eq:same interim utility} and~\eqref{eq:under F vs. hF in M2}, we have the following inequality
  \begin{equation}\label{eq:M'_1 vs. M_1}
 	\left \lvert\E_{b_{-i}\sim \FF_{-i}} \left[u_i(v_i,M_1(\tau_i(t),b_{-i}))\right] - \E_{\hb_{-i}\sim \hFF_{-i}} \left[u_i(v_i,M_2(t,\hb_{-i}))\right]\right\rvert \leq m\LL H\rho.
 \end{equation}
  	Suppose bidder $i$ has type $v_i$, how much more utility can she get by misreporting? Since $M_2$ is IR, she clearly cannot gain by reporting a type $t$, whose corresponding $\tau_i(t)=\perp$. Next, we argue that she cannot gain much by reporting any other possible types either. If all other bidders report truthfully, bidder $i$'s interim utility for reporting her true type  \begin{align*}
  	\E_{\hb_{-i}\sim \hFF_{-i}} \left[u_i(v_i,M_2(v_i,\hb_{-i}))\right] \geq &~ \E_{b_{-i}\sim \FF_{-i}} \left[u_i(v_i,M_1(\tau_i(v_i),b_{-i}))\right]- m\LL H\rho \\
  	\geq & ~\max_{x\in\supp(\FF_i)}\E_{b_{-i}\sim \FF_{-i}} \left[u_i(v_i,M_1(x,b_{-i}))\right]- m\LL H\rho-\eta\\
  	  	\geq & ~\max_{t:\tau_i(t)\neq \perp}\E_{b_{-i}\sim \FF_{-i}} \left[u_i(v_i,M_1(\tau_i(t),b_{-i}))\right]- m\LL H\rho-\eta\\
  	\geq & ~\max_{t:\tau_i(t)\neq \perp}\E_{\hb_{-i}\sim \hFF_{-i}} \left[u_i(v_i,M_2(t,\hb_{-i}))\right]- 2m\LL H\rho-\eta
  \end{align*}
	 
The first inequality is due to Inequality~\eqref{eq:M'_1 vs. M_1}. The second inequality is true because (a) if $v_i=\tau_i(v_i)$, then 
$$\E_{b_{-i}\sim \FF_{-i}} \left[u_i(v_i,M_1(\tau_i(v_i),b_{-i}))\right]\geq ~\max_{x\in\supp(\FF_i)}\E_{b_{-i}\sim \FF_{-i}} \left[u_i(v_i,M_1(x,b_{-i}))\right]-\eta$$
 as $M_1$ is $\eta$-BIC; (b) if $v_i\notin \supp(\FF_i)$, then by the definition of $\tau_i(v_i)$, 
 $$\E_{b_{-i}\sim \FF_{-i}} \left[u_i(v_i,M_1(\tau_i(v_i),b_{-i}))\right]\geq ~\max_{x\in\supp(\FF_i)}\E_{b_{-i}\sim \FF_{-i}} \left[u_i(v_i,M_1(x,b_{-i}))\right].$$
 The third inequality is because when $\tau_i(t)\neq \perp$ it must lie in $\supp(\FF_i)$.  The last inequality is again due to  Inequality~\eqref{eq:M'_1 vs. M_1}. 
 
	Finally, we show that $\revT(M_2,\hFF)$ is not much less than $\revT(M_1,\FF)$. Let $b\sim \FF$ and $\hb\sim \hFF$. There exists a coupling of $b$ and $\hb$ so that they are different w.p. less than $\rho$. When $b=\hb$, $M_1(b)=M_2(\hb)$. When $b\neq\hb$, the revenue in $M_1(b)$ is at most $nm\LL H$ more than the revenue in $M_2(\hb)$, as both mechanisms are IR. Therefore, $$\revT(M_2,\hFF)\geq \revT(M_1,\FF)-nm\LL H \rho.$$ 
	  	 \end{prevproof}

\subsection{Proof of Lemma~\ref{lem:prokhorov to TV}}\label{appx:randomized rounding}

\begin{prevproof}{Lemma}{lem:prokhorov to TV}
	According to Theorem~\ref{thm:prokhorov characterization}, there exists a coupling $\gamma$ of $\FF$ and $\hFF$ so that $$\Pr_{(x,y)\sim \gamma}\left[d(x,y)>\varepsilon \right]\leq \varepsilon.$$ Now we bound the probability that $\rfun(x)\neq \rfun(y)$, when $(x,y)$ is drawn from $\gamma$, and $\ell$ is drawn from $U[0,\delta]^k$.
	\begin{align*}
		&\Prob_{\ell\sim U[0,\delta]^k, (x,y)\sim\gamma}\left[\rfun(x)\neq \rfun(y)\right]\\
		=&\Prob_{\ell\sim U[0,\delta]^k, (x,y)\sim\gamma}\left[\rfun(x)\neq \rfun(y)~\land~d(x,y)>\varepsilon \right]\\
		&~~~~~~~~~~~~~~~~~~~~~~~~~~~~~~~~~~~~~~~~~+\Prob_{\ell\sim U[0,\delta]^k, (x,y)\sim\gamma}\left[\rfun(x)\neq \rfun(y)~\land~d(x,y)\leq\varepsilon \right]\\
		\leq & \Prob_{(x,y)\sim\gamma} \left[d(x,y)>\varepsilon \right]+\Prob_{\ell\sim U[0,\delta]^k}\left[\rfun(x)\neq \rfun(y) \mid d(x,y)\leq\varepsilon \right]\cdot \Pr_{(x,y)\sim\gamma}[d(x,y)\leq \varepsilon ]\\
		\leq & \varepsilon + \Prob_{\ell\sim U[0,\delta]^k}\left[\rfun(x)\neq \rfun(y) \mid d(x,y)\leq\varepsilon \right] 
	\end{align*}
	
	Now, we bound the probability that $\rfun(\cdot)$ rounds two points $x$ and $y$ to two different points when $x$ and $y$ are within distance $\varepsilon$. For any fixed $x$ and $y$, we have the following.
	\begin{align*}
		&\Prob_{\ell\sim U[0,\delta]^k}\left[\rfun(x)\neq \rfun(y) \right]\\
		\leq & \sum_{i\in[k]} \Prob_{\ell_i\sim U[0,\delta]}\left[\rfun_i(x)\neq \rfun_i(y)  \right]\\
		\leq& \sum_{i\in[k]} \frac{\lvert x_i-y_i \rvert}{\delta}\\
		=&\frac{d(x,y)}{\delta}
	\end{align*}
The first inequality follows from the union bound. Why is the second inequality true? If $|x_i-y_i|\geq \delta$, the inequality clearly holds, so we only need to consider the case where $|x_i-y_i|<\delta$. W.l.o.g. we assume $y_i\geq x_i$ and we consider the following two cases: (i) $\round{\frac{y_i}{\delta}}=\round{\frac{x_i}{\delta}}$ and (ii) $\round{\frac{y_i}{\delta}}=\round{\frac{x_i}{\delta}}+1$. In case (i), $\rfun_i(x)\neq \rfun_i(y) $ if and only if $\ell \in \left[x_i-\round{\frac{x_i}{\delta}}\cdot \delta, y_i -\round{\frac{y_i}{\delta}}\cdot \delta\right]$. Since $\ell$ is drawn from the uniform distribution over $[0,\delta]$, this happens with probability exactly $\frac{y_i-x_i }{\delta}$. In case (ii), $\rfun_i(x)\neq \rfun_i(y) $ if and only if $\ell \in \left[x_i-\round{\frac{x_i}{\delta}}\cdot \delta,\delta\right]\cup [0, y_i -\round{\frac{y_i}{\delta}}\cdot \delta]$. This again happens with probability $\frac{y_i-x_i }{\delta}$.
	Therefore, $$\Prob_{\ell\sim U[0,\delta]^k}\left[\rfun(x)\neq \rfun(y) \mid d(x,y)\leq\varepsilon \right]\leq \frac{\varepsilon}{\delta}$$
	and 
	\begin{equation}\label{eq:randomized rounding}
			\Prob_{\ell\sim U[0,\delta]^k, (x,y)\sim\gamma}\left[\rfun(x)\neq \rfun(y)\right]\leq \left(1+\frac{1}{\delta}\right)\varepsilon.
	\end{equation}

Clearly, for any choice of $\ell$,  $\norm{\round{\FF}_{\ell,\delta}-\round{\hFF}_{\ell,\delta}}_{TV}\leq\Pr_{(x,y)\sim \gamma}\left[\rfun(x)\neq \rfun(y)\right]$. Combining this inequality with Inequality~\eqref{eq:randomized rounding}, we have

\begin{align*}
	&\E_{\ell\sim U[0,\delta]^k}\left[\norm{\round{\FF}_{\ell,\delta}-\round{\hFF}_{\ell,\delta}}_{TV}\right]\\
	\leq  & \E_{\ell\sim U[0,\delta]^k}\left[\Pr_{(x,y)\sim \gamma}\left[\rfun(x)\neq \rfun(y)\right]\right]\\
	= & \Prob_{\ell\sim U[0,\delta]^k, (x,y)\sim\gamma}\left[\rfun(x)\neq \rfun(y)\right]\\
	\leq &\left(1+\frac{1}{\delta}\right)\varepsilon
\end{align*}

\notshow{	 From Inequality~\eqref{eq:randomized rounding}, we can further show that with high probability over the randomness of $\ell$,  the event that  $\Pr_{(x,y)\sim \gamma}\left[\rfun(x)\neq \rfun(y)\right]$ is large happens with low probability. In particular, for any constant $\CC\geq 1$ \begin{align*}
		&\Prob_{\ell\sim U[0,\delta]^k}\left[ \Prob_{(x,y)\sim\gamma}\left[\rfun(x)\neq \rfun(y) 
\right]> \CC \left(1+\frac{1}{\delta}\right)\varepsilon \right]\\
\leq &\frac{\E_{\ell\sim U[0,\delta]^k}\left[ \Prob_{(x,y)\sim\gamma}\left[\rfun(x)\neq \rfun(y) 
\right] \right]} {\CC \left(1+\frac{1}{\delta}\right)\varepsilon}\\
\leq & {1\over \CC}.
	\end{align*}
	
	The first inequality is due to Markov's inequality, and the second inequality follows from Equation~\eqref{eq:randomized rounding}. Hence, \begin{equation}\label{eq:UB on bad rounding}
		\Prob_{\ell\sim U[0,\delta]^k}\left[ \Prob_{(x,y)\sim\gamma}\left[\rfun(x)\neq \rfun(y) 
\right]\leq \CC \left(1+\frac{1}{\delta}\right)\varepsilon \right]\geq 1-  {1\over \CC}.
	\end{equation}
	 Note that for any fixed $\ell$ the random variables $\rfun(x)$ and $\rfun(y)$ when $(x,y)$ is drawn from $\gamma$ have distributions  $\round{\FF}_{\ell,\delta}$ and $\round{\hFF}_{\ell,\delta}$. In particular, $\norm{\round{\FF}_{\ell,\delta}-\round{\hFF}_{\ell,\delta}}_{TV}\leq \Prob_{(x,y)\sim\gamma}\left[\rfun(x)\neq \rfun(y)\right]$ for any fixed $\ell$. By Equation~\eqref{eq:UB on bad rounding}, 
	 \begin{equation*}
	 			\Pr_{\ell\sim U[0,\delta]^k}\left[\norm{\round{\FF}_{\ell,\delta}-\round{\hFF}_{\ell,\delta}}_{TV}\leq \CC\left(1+\frac{1}{\delta}\right)\varepsilon \right]\geq 1- {1\over \CC}.
	 \end{equation*}}
	 \end{prevproof}
	 
	 \notshow{
	 \begin{prevproof}{Lemma}{lem:prokhorov to TV}
	According to Theorem~\ref{thm:prokhorov characterization}, there exists a coupling $\gamma$ of $\FF$ and $\hFF$ so that $$\Pr_{(x,y)\sim \gamma}\left[d(x,y)>\varepsilon \right]\leq \varepsilon.$$ Now we bound the probability that $\rfun(x)\neq \rfun(y)$, when $(x,y)$ is drawn from $\gamma$, and $\ell$ is drawn from $U[0,\delta]^k$.
	\begin{align*}
		&\Prob_{\ell\sim U[0,\delta]^k, (x,y)\sim\gamma}\left[\rfun(x)\neq \rfun(y)\right]\\
		=&\Prob_{\ell\sim U[0,\delta]^k, (x,y)\sim\gamma}\left[\rfun(x)\neq \rfun(y)~\land~d(x,y)>\varepsilon \right]\\
		&~~~~~~~~~~~~~~~~~~~~~~~~~~~~~~~~~~~~~~~~~+\Prob_{\ell\sim U[0,\delta]^k, (x,y)\sim\gamma}\left[\rfun(x)\neq \rfun(y)~\land~d(x,y)\leq\varepsilon \right]\\
		\leq & \Prob_{(x,y)\sim\gamma} \left[d(x,y)>\varepsilon \right]+\Prob_{\ell\sim U[0,\delta]^k}\left[\rfun(x)\neq \rfun(y) \mid d(x,y)\leq\varepsilon \right]\cdot \Pr_{(x,y)\sim\gamma}[d(x,y)\leq \varepsilon ]\\
		\leq & \varepsilon + \Prob_{\ell\sim U[0,\delta]^k}\left[\rfun(x)\neq \rfun(y) \mid d(x,y)\leq\varepsilon \right] 
	\end{align*}
	
	Now, we bound the probability that $\rfun(\cdot)$ rounds two points $x$ and $y$ to two different points when $x$ and $y$ are within distance $\varepsilon$. For any fixed $x$ and $y$, we have the following.
	\begin{align*}
		&\Prob_{\ell\sim U[0,\delta]^k}\left[\rfun(x)\neq \rfun(y) \right]\\
		\leq & \sum_{i\in[k]} \Prob_{\ell_i\sim U[0,\delta]}\left[\rfun_i(x)\neq \rfun_i(y)  \right]\\
		\leq& \sum_{i\in[k]} \frac{\lvert x_i-y_i \rvert}{\delta}\\
		=&\frac{d(x,y)}{\delta}
	\end{align*}
The first inequality follows from the union bound. Why is the second inequality true? If $|x_i-y_i|\geq \delta$, the inequality clearly holds, so we only need to consider the case where $|x_i-y_i|<\delta$. W.l.o.g. we assume $y_i\geq x_i$ and we consider the following two cases: (i) $\round{\frac{y_i}{\delta}}=\round{\frac{x_i}{\delta}}$ and (ii) $\round{\frac{y_i}{\delta}}=\round{\frac{x_i}{\delta}}+1$. In case (i), $\rfun_i(x)\neq \rfun_i(y) $ if and only if $\ell \in \left[x_i-\round{\frac{x_i}{\delta}}\cdot \delta, y_i -\round{\frac{y_i}{\delta}}\cdot \delta\right]$. Since $\ell$ is drawn from the uniform distribution over $[0,\delta]$, this happens with probability exactly $\frac{y_i-x_i }{\delta}$. In case (ii), $\rfun_i(x)\neq \rfun_i(y) $ if and only if $\ell \in \left[x_i-\round{\frac{x_i}{\delta}}\cdot \delta,\delta\right]\cup [0, y_i -\round{\frac{y_i}{\delta}}\cdot \delta]$. This again happens with probability $\frac{y_i-x_i }{\delta}$.
	Therefore, $$\Prob_{\ell\sim U[0,\delta]^k}\left[\rfun(x)\neq \rfun(y) \mid d(x,y)\leq\varepsilon \right]\leq \frac{\varepsilon}{\delta},$$
	and 
	\begin{equation}\label{eq:randomized rounding}
			\Prob_{\ell\sim U[0,\delta]^k, (x,y)\sim\gamma}\left[\rfun(x)\neq \rfun(y)\right]\leq \left(1+\frac{1}{\delta}\right)\varepsilon.
	\end{equation}

	 From Equation~\eqref{eq:randomized rounding}, we can further show that with high probability over the randomness of $\ell$,  the event $\rfun(x)\neq \rfun(y)$ happens with low probability. In particular, for any constant $\CC\geq 1$ \begin{align*}
		&\Prob_{\ell\sim U[0,\delta]^k}\left[ \Prob_{(x,y)\sim\gamma}\left[\rfun(x)\neq \rfun(y) 
\right]> \CC \left(1+\frac{1}{\delta}\right)\varepsilon \right]\\
\leq &\frac{\E_{\ell\sim U[0,\delta]^k}\left[ \Prob_{(x,y)\sim\gamma}\left[\rfun(x)\neq \rfun(y) 
\right] \right]} {\CC \left(1+\frac{1}{\delta}\right)\varepsilon}\\
\leq & {1\over \CC}.
	\end{align*}
	
	The first inequality is due to Markov's inequality, and the second inequality follows from Equation~\eqref{eq:randomized rounding}. Hence, \begin{equation}\label{eq:UB on bad rounding}
		\Prob_{\ell\sim U[0,\delta]^k}\left[ \Prob_{(x,y)\sim\gamma}\left[\rfun(x)\neq \rfun(y) 
\right]\leq \CC \left(1+\frac{1}{\delta}\right)\varepsilon \right]\geq 1-  {1\over \CC}.
	\end{equation}
	 Note that for any fixed $\ell$ the random variables $\rfun(x)$ and $\rfun(y)$ when $(x,y)$ is drawn from $\gamma$ have distributions  $\round{\FF}_{\ell,\delta}$ and $\round{\hFF}_{\ell,\delta}$. In particular, $\norm{\round{\FF}_{\ell,\delta}-\round{\hFF}_{\ell,\delta}}_{TV}\leq \Prob_{(x,y)\sim\gamma}\left[\rfun(x)\neq \rfun(y)\right]$ for any fixed $\ell$. By Equation~\eqref{eq:UB on bad rounding}, 
	 \begin{equation*}
	 			\Pr_{\ell\sim U[0,\delta]^k}\left[\norm{\round{\FF}_{\ell,\delta}-\round{\hFF}_{\ell,\delta}}_{TV}\leq \CC\left(1+\frac{1}{\delta}\right)\varepsilon \right]\geq 1- {1\over \CC}.
	 \end{equation*}
	 \end{prevproof}

	 }

	 \subsection{Missing Proofs from Section~\ref{sec:robust multi}}\label{appx:proof_multi}
	 
	\begin{prevproof}{Lemma}{lem:original to rounded}
		We first define $\Mli$. If the bid profile $w\notin \supp(\lDD)$, the mechanism allocates nothing and charges no one. If the bid profile $w\in \supp(\lDD)$, for each bidder $i$  sample $w'_i$ independently from the distribution $\DD_i \mid \bigtimes_{j=1}^m \beta(w_{ij})$, where $\beta(w_{ij})$ is defined to be $[0,\ell_j)$ if $w_{ij}=0$ and $[w_{ij},w_{ij}+\delta)$ otherwise. Bidder $i$ receives allocation $x_{M,i}(w')$ and pays $(p_{M,i}(w')-m\LL\delta)^+=\max\{0, p_{M,i}(w')-m\LL\delta\}$. Note that,  for any $i\in [n]$, if $w_i$ is drawn from $\round{\DD_i}_{\ell,\delta}$ then $w'_{i}$ is drawn from $\DD_i$. If all bidders bid truthfully in $\Mli$, the revenue is at least $\rev(M,\DD)-nm\LL \delta$. Next, we argue that $\Mli$ is IR and $\xi_1$-BIC with $\xi_1=O(m\LL\delta)$.
		
		Note that for every bidder $i$ and $w_i\in \supp(\round{\DD_i}_{\ell,\delta})$ her interim utility in $\Mli$ when all other bidders bid truthfully is at least 
		$\E_{w_i'\sim \DD_i \mid \bigtimes_{j=1}^m \beta(w_{ij}), w_{-i}'\sim \DD_{-i}}\left [u_i(w_i, M(w_i',w_{-i}'))\right] $ due to the definition of $\Mli$. Now consider every realization of $w_i'$, it must hold that
\begin{align*}
 	&\E_{w_{-i}'\sim \DD_{-i}}\left [u_i(w_i, M(w_i',w_{-i}'))\right] \\
 	\geq & \E_{w_{-i}'\sim \DD_{-i}}\left [u_i(w'_i, M(w_i',w_{-i}'))\right]-m\L \delta\\
 	\geq & \max_{x\in \supp(\DD_i)}\E_{w_{-i}'\sim \DD_{-i}}\left [u_i(w'_i, M(x,w_{-i}'))\right]-m\L \delta\\
 	\geq &  \max_{x\in \supp(\DD_i)}\E_{w_{-i}'\sim \DD_{-i}}\left [u_i(w_i, M(x,w_{-i}'))\right] - 2m\LL\delta 
 \end{align*}
		The first and the last inequalities are both due to the fact that the valuation is $\LL$-Lipschitz and $\norm{w_i-w_i'}_1\leq m\delta$. The second inequality is because $M$ is BIC w.r.t. $\DD$. 
		 Hence, bidder $i$'s interim utility in $\Mli$ is at least $\max_{x\in \supp(\DD_i)}\E_{w_{-i}'\sim \DD_{-i}}\left [u_i(w_i, M(x,w_{-i}'))\right] - 2m\LL\delta $.
		
		If bidder $i$ misreports, her utility is no more than $$\max_{x\in \supp(\DD_i)} \E_{w_{-i}'\sim \DD_{-i}}\left [u_i(w_i, M(x,w_{-i}'))\right]+m\LL\delta,$$ due to the definition of $\Mli$. Therefore, misreporting can increase bidder $i$'s utility by at most $3m\LL\delta$, and $\Mli$ is $3m\LL\delta$-BIC.

		
		Next, we argue that $\Mli$ is IR. If the $w_{-i}\notin \supp(\lDD_{-i})$, bidder $i$'s utility is $0$. So we focus on the case where $w_{-i}\in \supp(\lDD_{-i})$. We will show that for any realization of $w'_i$ and $w'_{-i}$, bidder $i$'s utility is non-negative. If the payment is $0$, the claim is trivially true. If the payment is nonzero, bidder $i$ pays $p_{M,i}(w')-m\LL \delta$ and has utility $u_i(w_i, M(w_i',w_{-i}')))+m\L\delta$ which is at least $u_i(w'_i, M(w_i',w_{-i}')))$, since the valuation is $\LL$-Lipschitz and $\norm{w_i-w_i'}_1\leq m\delta$. As $M$ is IR, $u_i(w'_i, M(w_i',w_{-i}')))\geq 0$. Thus, bidder $i$'s utility is non-negative and $\Mli$ is IR.
		\end{prevproof}

\begin{prevproof}{Lemma}{lem:discrete to continuous}
	We first construct $\hMl$. For any bid profile $w$, construct $w'=(\rfun(w_1),\ldots,\rfun(w_n))$, and run $\Mlii$ on $w'$. Bidder $i$ receives allocation $x_{\Mlii, i}(w')$ and pays $\max\{0,p_{\Mlii,i}(w')-m\LL\delta\}$. Note that if $w_i\sim \hDD_i$, then $w_i'\sim \round{\hDD_i}_{\ell,\delta}$. Assuming all other bidders bid truthfully and bidder $i$'s type is $w_i$, bidder $i$'s interim utility for bidding truthfully is \begin{align*}
 	\E_{b_{-i}\sim \hDD_{-i}}\left[u_i(w_i,\hMl(w_i,b_{-i}))\right]\geq &~\E_{b'_{-i}\sim \hlDD_{-i}}\left[u_i(w_i,\Mlii(w'_i,b'_{-i}))\right]\\
 	\geq &~ \E_{b_{-i}\sim \hlDD_{-i}}\left[u_i(w'_i,\Mlii(w'_i,b'_{-i}))\right]-m\L\delta\\
 	 	\geq &~ \max_{x\in \supp(\round{\hDD_i}_{\ell,\delta})} \E_{b'_{-i}\sim \hlDD_{-i}}\left[u_i(w'_i,\Mlii(x,b'_{-i}))\right]-\xi_2-m\L\delta\\
 	\geq &~ \max_{x\in \supp(\round{\hDD_i}_{\ell,\delta})} \E_{b'_{-i}\sim \hlDD_{-i}}\left[u_i(w_i,\Mlii(x,b'_{-i}))\right]-\xi_2-2m\LL \delta\\
 	\geq & ~ \max_{y\in \supp(\hDD_i)} \E_{b_{-i}\sim \hDD_{-i}}\left[u_i(w_i,\hMl(y,b_{-i}))\right]-\xi_2-3m\LL \delta
 \end{align*}
 
 The first inequality and the last equality are due to the definition of $\hMl$. The second and the fourth inequalities are due to the $\LL$-Lipschitzness of the valuation function and $\norm{w_i-w_i'}_1\leq m\delta$. The third inequality is because $\Mlii$ is a $\xi_2$-BIC mechanism w.r.t. $\hlDD$. By this chain of inequalities, we know that $\hMl$ is a $(\xi_2+3m\LL \delta)$-BIC mechanism w.r.t. $\hDD$. 
 
 Next, we argue that $\hMl$ is also IR. Consider any bidder $i$ and type profile $w$, $\hMl(w)$ has the same allocation as $\Mlii(w')$. When bidder $i$'s payment is $0$, her utility is clearly non-negative. When bidder $i$'s payment is $p_{\Mlii,i}(w')-m\LL\delta$, her utility is at least $u_i(w'_i, \Mlii(w'))$ due to the $\L$-Lipschitzness of the valuation function and $\norm{w_i-w_i'}_1\leq m\delta$. Since $\Mlii$ is IR, bidder $i$'s utility in $\hMl$ is also non-negative.
 
 Finally, if all bidders bid truthfully in $\hMl$ when their types are drawn from $\hDD$, its revenue under truthful bidding is $$\revT\left(\hMl,\hDD\right)\geq \revT\left(\Mlii,\hlDD\right)-nm\LL\delta.$$
\end{prevproof}

\notshow{
\begin{prevproof}{Lemma}{lem:discrete to continuous}
	We first construct $\hM$. For any bid profile $w$, construct $w'=(\rfun(w_1),\ldots,\rfun(w_n))$, and run $M_2$ on $w'$. Bidder $i$ receives allocation $x_{M_2}(w')$ and pays $\max\{0,p_{M_2}(w')-m\LL\delta\}$. Note that if $w_i\sim \hDD_i$, $w_i'\sim \round{\hDD_i}_{\ell,\delta}$. Assuming all other bidders bid truthfully and bidder $i$'s type is $w_i$, bidder $i$'s interim utility for bidding truthfully is \begin{align}
 	\E_{b_{-i}\sim \hDD_{-i}}\left[w_i,\hM(w_i,b_{-i})\right]\geq &~\E_{b'_{-i}\sim \hlDD_{-i}}\left[w_i,M_2(w'_i,b'_{-i})\right]+m\LL \delta\\
 	\geq &~ \E_{b_{-i}\sim \hlDD_{-i}}\left[w'_i,M_2(w'_i,b'_{-i})\right]\label{eq:IR discrete to continuous}\\
 	 	\geq &~ \max_{x\in \supp(\round{\hDD_i}_{\ell,\delta})} \E_{b'_{-i}\sim \hlDD_{-i}}\left[w'_i,M_2(x,b'_{-i})\right]-\xi_2\\
 	\geq &~ \max_{x\in \supp(\round{\hDD_i}_{\ell,\delta})} \E_{b'_{-i}\sim \hlDD_{-i}}\left[w_i,M_2(x,b'_{-i})\right]-\xi_2-m\LL \delta\\
 	= & ~ \max_{y\in \supp(\hDD_i)} \E_{b_{-i}\sim \hDD_{-i}}\left[w_i,\hM(y,b_{-i})\right]-\xi_2-2m\LL \delta
 \end{align}
 
 The first and last equalities are due to the definition of $\hM$. The first and the third inequalities are due to the definition of $w_i'$ and the $\LL$-Lipschitzness of the valuation function. The second inequality is because $M_2$ is a $\xi_2$-BIC mechanism w.r.t. $\hlDD$. By this chain of inequalities, we know that $\hM$ is a $(\xi_2+2m\LL \delta)$-BIC mechanism w.r.t. $\hDD$. It is not hard to see that $\hM$ is also IR by Inequality~\eqref{eq:IR discrete to continuous}. Finally, if all bidders bid truthfully in $\hM$ when their types are drawn from $\hDD$, its revenue $$\revT(\hM,\hDD)=\revT(M_2,\hlDD)-nm\LL\delta.$$
\end{prevproof}

}

	 \begin{prevproof}{Theorem}{thm:multi-item p-robustness}
	 	First, sample $\ell$ uniformly from $[0,\delta]^m$, and construct $\round{\DD_i}_{\ell,\delta}$ for all $i\in[n]$. According to Lemma~\ref{lem:original to rounded}, we can construct a mechanism $\Mli$ based on $M$ that is $\xi_1=O(m\LL \delta)$-BIC w.r.t. $\bigtimes_{i=1}^n \round{\DD_i}_{\ell,\delta}$, IR, and has revenue $\revT\left(\Mli, \bigtimes_{i=1}^n \round{\DD_i}_{\ell,\delta}\right)\geq \rev(M,\DD)-nm\LL\delta$. 
	 	
	 	Next, we transform $\Mli$ to $\Mlii$ using Lemma~\ref{lem:transformation under TV}. We use $\el_i$ to denote $\norm{\round{\DD_i}_{\ell,\delta}-\round{\hDD_i}_{\ell,\delta}}_{TV}$ for our sample $\ell$ and every $i\in [n]$, and $\rho^{(\ell)}$ to denote $\sum_{i\in [n]}\el_i$. For every realization of $\ell$, $\Mlii$ is $\xi_2=\left ( 2m\LL H\rho^{(\ell)}+\xi_1\right)$-BIC w.r.t. $\bigtimes_{i=1}^n \round{\hDD_i}_{\ell,\delta}$ and IR. Its revenue under truthful bidding satisfies $$\revT\left(\Mlii, \bigtimes_{i=1}^n \round{\hDD_i}_{\ell,\delta}\right)\geq \revT\left(\Mli, \bigtimes_{i=1}^n \round{\DD_i}_{\ell,\delta}\right)-nm\LL H\rho^{(\ell)}.$$
	 	
	 Lemma~\ref{lem:discrete to continuous} shows that we can construct $\hMl$ using $\Mlii$, such that $\hMl$ is a $(\xi_2+3m\LL\delta)$-BIC w.r.t. $\hDD$ and IR mechanism with revenue $$\revT\left(\hMl,\hDD\right)\geq\revT\left(\Mlii, \bigtimes_{i=1}^n \round{\hDD_i}_{\ell,\delta}\right)-nm\LL\delta.$$

	 Since $\hMl$ is $O(m\LL\delta+m\LL H\rho^{(\ell)})$-BIC w.r.t. $\hDD$ and IR for every realization of $\ell$, our mechanism $\hM$ is clearly $O\left (m\LL\delta+m\LL H\cdot \E_{\ell\sim U[0,\delta]^m} \left[\rho^{(\ell)}\right]\right )$-BIC w.r.t. $\hDD$ and IR. Moreover, its expected revenue under truthful bidding satisfies
	 
	 $$\revT\left(\hM,\hDD\right)\geq\rev\left(M, \DD\right)-O\left(nm\LL\delta+nm\LL H\cdot \E_{\ell\sim U[0,\delta]^m} \left[\rho^{(\ell)}\right]\right).$$
	 
	 	According to Lemma~\ref{lem:prokhorov to TV}, $$\E_{\ell\sim U[0,\delta]^m} \left[\rho^{(\ell)}\right]\leq n\left(1+{1\over \delta}\right)\varepsilon.$$

\notshow{
	 	\begin{equation}\label{eq:randomized rounding success rate}
	 		\Pr_{\ell\sim U[0,\delta]^m}\left[\norm{\round{\hDD_i}_{\ell,\delta}-\round{\DD_i}_{\ell,\delta}}_{TV}>{n\over \alpha} \left(1+\frac{1}{\delta}\right)\varepsilon\right]\leq \alpha/ n
	 	\end{equation} 
	 	for any $\alpha>0$. By the union bound, \begin{equation}\label{eq:good l}
 	\norm{\round{\hDD_i}_{\ell,\delta}-\round{\DD_i}_{\ell,\delta}}_{TV}\leq {n\over \alpha} \left(1+\frac{1}{\delta}\right)\varepsilon \text{ for all $i\in[n]$}
 \end{equation}
happens with probability at least $1-\alpha$. From now on, we will only consider the case where Equation~\eqref{eq:good l} is true. Choose $\FF_i$ to be $\round{\DD_i}_{\ell,\delta}$, $\hFF_i$ to be $\round{\hDD_i}_{\ell,\delta}$, and $\norm{\FF_i-\hFF_i}_{TV}\leq {n\over \alpha} \left(1+\frac{1}{\delta}\right)\varepsilon = \mu$. We use Lemma~\ref{lem:transformation under TV} to construct mechanism $M_2$ based on only $\Mli$ but not any information about $\hDD$ or $\bigtimes_{i=1}^n \round{\hDD_i}_{\ell,\delta}$, such that $M_2$ is $(\xi_2=2nm\LL H\mu+\xi_1)$-BIC w.r.t. $\bigtimes_{i=1}^n \round{\hDD_i}_{\ell,\delta}$, IR, and has revenue $$\revT\left(M_2,\bigtimes_{i=1}^n \round{\hDD_i}_{\ell,\delta}\right)\geq \revT\left(\Mli, \bigtimes_{i=1}^n \round{\DD_i}_{\ell,\delta}\right)-O\left(n^2m\LL H\mu\right).$$
}

We choose $\delta$ to be $\sqrt{nH\varepsilon}$, and $\hM$ becomes $\kappa$-BIC w.r.t. $\hDD$, where $\kappa=O\left( n m \LL H \varepsilon+   m \LL \sqrt{nH\varepsilon}\right)$, and IR. Furthermore, 
$$\revT\left(\hM,\hDD\right)\geq\rev\left(M, \DD\right)-O\left(n\kappa\right).$$


	 \end{prevproof}

\subsection{Lipschitz Continuity of the Optimal Revenue in Multi-item Auctions}\label{sec:Lipschitz Continuity of Opt}

	Using Theorem~\ref{thm:multi-item p-robustness}, we can easily prove that the optimal BIC revenue w.r.t. $\DD$ and the optimal BIC revenue w.r.t. $\hDD$ are close as long as $\DD_i$ and $\hDD_i$ are close in either the total variation distance or the Prokhorov distance for all $i\in[n]$.
	
	\begin{theorem}[Lipschitz Continuity of the Optimal Revenue]\label{thm:continuous of OPT under Prokhorov}
		Consider the general mechanism design setting of Section 2. Recall that $\L$ is the Lipschitz constant of the valuations. For any distributions $\DD=\bigtimes_{i=1}^n \DD_i$ and $\hDD=\bigtimes_{i=1}^n \hDD_i$, where $\DD_i$ and $\hDD_i$ are supported on $[0,H]^m$ for every $i\in[n]$ 
		\begin{itemize}
					\item  If $\norm{\DD_i-\hDD_i}_{TV}\leq \varepsilon$ for all $i\in[n]$, then $$\left \lvert\opt(\DD)-\opt(\hDD)\right\rvert\leq O\left(nm\LL H\left(n\varepsilon+\sqrt{n\varepsilon}\right)\right);$$
\item if $\norm{\DD_i-\hDD_i}_P\leq \varepsilon$ for all $i\in[n]$, then $$\left \lvert\opt(\DD)-\opt(\hDD)\right\rvert\leq O\left(n\kappa+n\sqrt{m\LL H\kappa}\right),$$
		where $\kappa=O\left(n m \LL H \varepsilon+  m \LL \sqrt{nH\varepsilon}\right)$.
		\end{itemize}
	\end{theorem}
		
		\begin{prevproof}{Theorem}{thm:continuous of OPT under Prokhorov}
		Let $M^*$ be the optimal BIC mechanism for $\DD$. We first prove the Prokorov case. According to Theorem~\ref{thm:multi-item p-robustness}, there exists a mechanism $\hM^*$ such that it is $\kappa$-BIC w.r.t. $\hDD$ and IR. Moreover, $$\revT(\hM^*,\hDD)\geq \rev(M^*,\DD)-O(n\kappa).$$
		By Lemma~\ref{lem:eps-BIC to BIC}, $\revT\left(\hM^*,\hDD\right)\leq \opt\left(\hDD\right)
		+2n\sqrt{m\LL H\kappa}$. Combining the two inequalities, we have $$\opt\left(\hDD\right)\geq \opt(\DD)-O\left(n\kappa+n\sqrt{m\LL H\kappa}\right).$$
		
		By symmetry, we can also argue that $$\opt(\DD)\geq \opt\left(\hDD\right)-O\left(n\kappa+n\sqrt{m\LL H\kappa}\right).$$
		
		In the TV case, $\revT\left(\hM^*,\hDD\right)\geq \rev(M^*,\DD)-O(n^2m\LL H\varepsilon)$. Since $\hM^*$ is $O(mn\L H\varepsilon)$-BIC, $\opt\left(\hDD\right)\geq \revT\left(\hM^*,\hDD\right)-O(nm\LL H\sqrt{n\varepsilon})$ due to Lemma~\ref{lem:eps-BIC to BIC}. By symmetry and the inequalities above, we have $\left \lvert\opt(\DD)-\opt\left(\hDD\right)\right\rvert\leq O\left(nm\LL H(n\varepsilon+\sqrt{n\varepsilon})\right)$.
	\end{prevproof}
	
	\subsection{Approximation Preserving Transformation}\label{sec:approximation preserving guarantee}
\notshow{

\begin{theorem}\label{thm:approximation preserving}
	Given $\DD=\bigtimes_{i=1}^n \DD_i$, where $\DD_i $ is a $m$-dimensional distribution for all $i\in[n]$, and a BIC mechanism $M$ w.r.t. $\DD$. Let $\alpha$ be an arbitrary real number in $(0,1)$. If $M$ is a $c$-approximation to the optimal BIC revenue for $\DD$, we have a randomized algorithm that constructs a mechanism $\hM$ \yangnote{in polynomial time given access to the optimizer that provides the best bid} such that for any possible true type distribution $\hDD=\bigtimes_{i=1}^n \hDD_i$ satisfying $\norm{\DD_i-\hDD_i}_P\leq \varepsilon$ for all $i\in[n]$, with probability at least $1-\alpha$ over the randomness of the construction of $\hM$, $\hM$ is a $\kappa\over \alpha$-BIC and IR mechanism for $\hDD$, where  $\kappa=O\left( n^2 m \LL H \varepsilon+  n m \LL \sqrt{H\varepsilon}\right)$. Moreover, the expected revenue of $\hM$ under truthful bidding is {$$\revT(\hM,\hDD)\geq c\cdot\opt_{\kappa\over \alpha}(\hDD)-O\left({ n\kappa\over \alpha}+n\sqrt{{m\LL H\kappa\over \alpha}}\right).$$} 
	\end{theorem}
	
	}

	\notshow{ 
	\begin{theorem}\label{thm:multi-item p-robustness}
	 	Given $\DD=\bigtimes_{i=1}^n \DD_i$, where $\DD_i $ is a $m$-dimensional distribution for all $i\in[n]$, and a BIC mechanism $M$ w.r.t. $\DD$. We use $\hDD=\bigtimes_{i=1}^n \hDD_i$ to denote the true but unknown type distribution, which could be any distribution that satisfies $\norm{\DD_i-\hDD_i}_P\leq \varepsilon$ for all $i\in[n]$. For any $\alpha\in (0,1)$, our $\hDD$-oblivious randomized algorithm constructs a mechanism $\hM$ such that with probability at least $1-\alpha$ the following holds for any fixed $\hDD$:
	 	\begin{enumerate}
	 		\item $\hM$ is a $\frac{\kappa}{\alpha}$-BIC and IR mechanism for $\hDD$, where  $\kappa=O\left( n^2 m \LL H \varepsilon+  n m \LL \sqrt{H\varepsilon}\right)$;
	 		\item  the expected revenue of $\hM$ under truthful bidding is $$\revT(\hM,\hDD)\geq \rev(M,\DD)-O\left(\frac{n\kappa}{\alpha}\right).$$
	 	\end{enumerate}
	 	
	 	If $\norm{\DD_i-\hDD_i}_{TV}\leq \varepsilon$ for all $i\in[n]$, our algorithm becomes deterministic and provides stronger guarantees. More specifically, \begin{enumerate}
	 		\item $\hM$ is a $O(n m \LL H \varepsilon)$-BIC and IR mechanism for $\hDD$;
	 		\item  the expected revenue of $\hM$ under truthful bidding is $$\revT(\hM,\hDD)\geq \rev(M,\DD)-O\left(n^2 m \LL H \varepsilon\right).$$
	 	\end{enumerate}\yangnote{Finally, our algorithm runs in polynomial time if we are given an oracle OptimizerBIC.} 	 \end{theorem}
	 	
	 	}

	\begin{theorem}[Approximation Preserving Transformation]\label{thm:approximation preserving}
	Consider the general mechanism design setting of Section 2. Recall that $\L$ is the Lipschitz constant of the valuations. Given $\DD=\bigtimes_{i=1}^n \DD_i$, where $\DD_i $ is a $m$-dimensional distribution supported on $[0,H]^m$ for all $i\in[n]$, and a BIC w.r.t. $\DD$ and IR mechanism $M$. We use $\hDD=\bigtimes_{i=1}^n \hDD_i$ to denote the true but unknown type distribution, and $\hDD_i$ is supported on $[0,H]^m$ for all $i\in[n]$. 
	
	If $\norm{\DD_i-\hDD_i}_{TV}\leq \varepsilon$ for all $i\in[n]$, we can construct a mechanism $\hM$, in a way that is completely oblivious to the true distribution $\hDD$, such that 
	\begin{enumerate}
	 		\item $\hM$ is $\eta$-BIC w.r.t. $\hDD$ and IR, where $\eta=O(n m \LL H \varepsilon)$;
	 		\item if $M$ is a $c$-approximation to the optimal BIC revenue for $\DD$, then $$\revT\left(\hM,\hDD\right)\geq c\cdot\opt_{\eta}\left(\hDD\right)-O\left(nm\LL H\left(n\varepsilon+\sqrt{n\varepsilon}\right)\right).$$
	 		\end{enumerate}
If $\norm{\DD_i-\hDD_i}_{P}\leq \varepsilon$ for all $i\in[n]$, we can again construct a mechanism $\hM$, in a way that is completely oblivious to the true distribution $\hDD$, such that
\begin{enumerate}
	 		\item $\hM$ is $\kappa$-BIC w.r.t. $\hDD$ and IR, where  $\kappa=O\left( n m \LL H \varepsilon+   m \LL \sqrt{nH\varepsilon}\right)$;
	 		\item if $M$ is a $c$-approximation to the optimal BIC revenue for $\DD$, then $\hM$ is almost a $c$-approximation to the optimal ${\kappa}$-BIC revenue for $\hDD$, that is, $$\revT\left(\hM,\hDD\right)\geq c\cdot\opt_{\kappa}\left(\hDD\right)-O\left({ n\kappa}+n\sqrt{{m\LL H\kappa}}\right).$$
	 	\end{enumerate}
		\end{theorem}
\begin{prevproof}{Theorem}{thm:approximation preserving}
For the TV case, by Theorem~\ref{lem:transformation under TV}, we can construct a $\eta$-BIC w.r.t. $\hDD$ and IR mechanism $\hM$ such that $\revT\left(\hM,\hDD\right)\geq \rev(M,\DD)-O\left(n^2m\L H\varepsilon\right )\geq c\cdot \opt(\DD)-O\left(n^2m\L H\varepsilon\right )$. By Theorem~\ref{thm:continuous of OPT under Prokhorov}, $\opt(\DD)$ is at least  $\opt(\hDD)-O\left(nm\L H(n\varepsilon+\sqrt{n\varepsilon})\right)$. Finally, $\opt(\hDD)\geq \opt_\eta(\hDD) -2n\sqrt{m\LL H\eta}$ due to Lemma~\ref{lem:eps-BIC to BIC}, so 
$$\revT\left(\hM,\hDD\right)\geq c\cdot\opt_\eta\left(\hDD\right)-O\left(nm\L H(n\varepsilon+\sqrt{n\varepsilon})\right).$$

	For the Prokhorov case, according to Theorem~\ref{thm:multi-item p-robustness}, we can construct a $\kappa$-BIC w.r.t. $\hDD$ and IR mechanism $\hM$ such that  $\revT\left(\hM,\hDD\right)\geq \rev(M,\DD)-O\left({n\kappa}\right)\geq c\cdot \opt(\DD)-O\left({n\kappa}\right)$.	By Theorem~\ref{thm:continuous of OPT under Prokhorov} and Lemma~\ref{lem:eps-BIC to BIC},  $\opt(\DD)\geq \opt\left(\hDD\right)-O\left(n\kappa+n\sqrt{m\LL H\kappa}\right)\geq \opt_{\kappa}(\hDD)-O\left(n\kappa+n\sqrt{{m\LL H\kappa}}\right)$	Chaining all the inequalities above, we have $$\revT\left(\hM,\hDD\right)\geq c\cdot\opt_{\kappa}(\hDD)-O\left({ n\kappa}+n\sqrt{{m\LL H\kappa}}\right).$$
	\end{prevproof}

If there is a single bidder, we can strengthen Theorem~\ref{thm:approximation preserving} and make constructed mechanism $\hM$ exactly IC with essentially the same guarantees.

\begin{theorem}[Single-Bidder Approximation Preserving Transformation]\label{thm:single approximation preserving}
	Consider the general mechanism design setting of Section 2. Recall that $\L$ is the Lipschitz constant of the valuations. Given a $m$-dimensional distribution $\DD$ supported on $[0,H]^m$, and a IC and IR mechanism $M$. We use $\hDD$ to denote the true but unknown type distribution, and $\hDD$ is also supported on $[0,H]^m$. 
	\begin{itemize}
		\item If $\norm{\DD-\hDD}_{TV}\leq \varepsilon$, we can construct an IC and IR mechanism $\hM$, in a way that is completely oblivious to the true distribution $\hDD$, such that if $M$ is a $c$-approximation to the optimal BIC revenue for $\DD$, then $$\rev\left(\hM,\hDD\right)\geq c\cdot \left(1-O\left(\sqrt{m\LL H\varepsilon}\right)\right)\cdot\opt\left(\hDD\right)-O\left(\left(m\LL H+\sqrt{m\LL H}\right)\cdot\sqrt{\varepsilon}\right).$$
		\item If $\norm{\DD-\hDD}_{P}\leq \varepsilon$, we can again construct an IC and IR mechanism $\hM$, in a way that is completely oblivious to the true distribution $\hDD$, such that if $M$ is a $c$-approximation to the optimal BIC revenue for $\DD$,  $$\rev\left(\hM,\hDD\right)\geq c\cdot\left(1-\sqrt{\kappa}\right)\cdot\opt\left(\hDD\right)-O\left({ \kappa}+\left(\sqrt{m\LL H}+1\right)\cdot\sqrt{\kappa}\right),$$
	 where  $\kappa=O\left( m \LL H \varepsilon+   m \LL \sqrt{H\varepsilon}\right)$.
	\end{itemize}	 	
		\end{theorem}

\begin{prevproof}{Theorem}{thm:single approximation preserving}
	We only sketch the proof here. Let $M'$ be the mechanism constructed using Theorem~\ref{thm:approximation preserving}, and we construct another mechanism $\hM$ by modifying $M'$ using Lemma~\ref{lem:eps-IC to IC}. Clearly, $\hM$ is IC and IR. It is not hard to verify that $\rev\left(\hM,\hDD\right)$ satisfies the guarantees in the statement by combining the revenue guarantees for $\revT\left(M',\hDD\right)$ as provided by Theorem~\ref{thm:approximation preserving} and the relation between $\rev\left(\hM,\hDD\right)$ and $\revT\left(M',\hDD\right)$ as stated in Lemma~\ref{lem:eps-IC to IC}. 
	\end{prevproof}


\notshow{

\subsection{From Robustness to Learnability}
In this section, we show the connection between learnability and robustness by rephrasing Theorem~\ref{thm:approximation preserving}. In particular, we show that as long as we can learn each bidder's type distribution within Prokhorov distance $\Theta\left({\varepsilon^4 \over\poly(n,m,\L,H)}\right)$ or TV distance $\Theta\left({\varepsilon^2 \over n^3m^2\L^2 H^2}\right)$, we can construct a nearly BIC mechanism that has revenue at most $\varepsilon$ less than the optimal revenue w.r.t. $\hDD$. 

\notshow{

\begin{theorem}\label{thm:how close do we need to learn}
Let $\hDD=\bigtimes_{i=1}^n \hDD_i$ be the true but unknown type distribution, and $\alpha$ be an arbitrary real number in $(0,1)$. For any $\varepsilon\in (0,1)$, there exists a $\delta = \Theta\left({\alpha^4\varepsilon^4 \over\poly(n,m,\L,H)}\right)$ such that if we can learn an approximate distribution $\DD=\bigtimes_{i=1}^n \DD_i$ satisfying $\norm{ \DD_i- \hDD_i}_P\leq \delta$, there is a randomized algorithm that converts any BIC and IR mechanism $M$ w.r.t. $\DD$ to a mechanism $\hM$ that have the following properties with probability at least $1-\alpha$ over the randomness of the algorithm:\begin{enumerate}
	\item $\hM$ is $\kappa\over \alpha$-BIC and IR w.r.t. $\hDD$, where $\kappa=O\left(\left( \alpha\varepsilon\over nm\L H \right)^2 \right)$;
	\item if $M$ is a $c$-approximation to the optimal BIC revenue w.r.t. $\DD$, $\revT(\hM,\hDD)\geq  c\cdot\opt_{\kappa\over \alpha}(\hDD)-\varepsilon$.
\end{enumerate}\end{theorem}

}

\begin{theorem}\label{thm:how close do we need to learn}
Let $\hDD=\bigtimes_{i=1}^n \hDD_i$ be the true but unknown type distribution, and we use $\DD=\bigtimes_{i=1}^n \DD_i$ to denote the approximate type distribution that we learned.
\begin{itemize}
	\item Prokhorov distance: For any $\alpha\in (0,1)$ and any $\varepsilon\in (0,1)$, there exists a $\delta_1 = \Theta\left({\alpha^4\varepsilon^4 \over\poly(n,m,\L,H)}\right)$ such that if we can learn each $\DD_i$ satisfying $\norm{ \DD_i- \hDD_i}_P\leq \delta_1$ for each $i\in[n]$, there is a randomized algorithm that converts any BIC and IR mechanism $M$ w.r.t. $\DD$ to a mechanism $\hM$ that have the following properties with probability at least $1-\alpha$ over the randomness of the algorithm:\begin{enumerate}
	\item $\hM$ is $\kappa\over \alpha$-BIC and IR w.r.t. $\hDD$, where $\kappa=O\left(\left( \alpha\varepsilon\over nm\L H \right)^2 \right)$;
	\item if $M$ is a $c$-approximation to the optimal BIC revenue w.r.t. $\DD$, $\revT(\hM,\hDD)\geq  c\cdot\opt_{\kappa\over \alpha}(\hDD)-\varepsilon$.
\end{enumerate}
\item {TV distance} : For any $\varepsilon\in (0,1)$, there exists a $\delta_2 = \Theta\left({\varepsilon^2 \over n^3m^2\L^2 H^2}\right)$ such that if we can learn each $\DD_i$ satisfying $\norm{ \DD_i- \hDD_i}_{TV}\leq \delta_2$ for each $i\in[n]$, there is a deterministic algorithm that converts any BIC and IR mechanism $M$ w.r.t. $\DD$ to a mechanism $\hM$ that have the following properties:
\begin{enumerate}
	 		\item $\hM$ is a $\eta$-BIC and IR mechanism for $\hDD$, where $\eta=O({\varepsilon^2\over n m\L H })$;
	 		\item if $M$ is a $c$-approximation to the optimal BIC revenue for $\DD$, then $$\revT(\hM,\hDD)\geq c\cdot\opt_{\eta}(\hDD)-\varepsilon.$$
	 		\end{enumerate}

\end{itemize} \end{theorem}
If we have sample access to distribution $\hDD$, we can use Theorem~\ref{thm:how close do we need to learn} to determine the sample complexity for learning a close to optimal mechanism.

}

\section{Missing Proofs from Section~\ref{sec:old sampling}}\label{appx:sample independent}

We first show that for any product distribution $\FF$, we can learn the rounded distribution of $\FF$ within small TV distance with polynomially many samples.

\begin{lemma}\label{lem:discretized product TV}
	Let $\FF=\bigtimes_{j=1}^m \FF_j$, where $\FF_j$ is an arbitrary distribution supported on $[0,H]$ for every $j\in[m]$. Given $N=O\left({m^3H\over \eta^3}\cdot(\log 1/\delta+\log m)\right)$ samples, we can learn a product distribution $\hFF=\bigtimes_{j=1}^m \hFF_j$ such that $$\norm{{\FF}-\hFF}_{P}\leq \eta$$ with probability at least $1-\delta$.\end{lemma}
	
\begin{proof}
	We denote the samples as $s^1,\ldots,s^N$. Round each sample to multiples of $\eta'=\eta/m$. More specifically, let $\hat{s}^i= \left(\round{s^i_1/\eta'}\cdot \eta',\ldots, \round{s^i_m/\eta'}\cdot \eta' \right )$ for every sample $i\in[N]$. Let $\hFF_j$ be the uniform distribution over $\hat{s}^1_j,\ldots, \hat{s}^N_j$. Let $\uFF_j=\round{\FF_j}_{0,\eta'}$. Note that $\hFF_j$ is the empirical distribution of $N$ samples from $\uFF_j$. As $\left\lvert \supp(\uFF_j)\right\rvert = \round{\frac{H}{\eta'}}={mH\over \eta}$, with $N=O\left({ \lvert \supp(\uFF_j)\rvert\over \eta'^2}\cdot(\log 1/\delta+\log m)\right)$ samples, the empirical distribution $\hFF_j$ should satisfy $\norm{\hFF_j-\uFF_j}_{TV}\leq \eta'$ with probability at least $1-\delta/m$. By the union bound $\norm{\hFF_j-\uFF_j}_{TV}\leq \eta'$ for all $j\in[m]$ with probability at least $1-\delta$, which implies {$\norm{\hFF-\uFF}_{TV}\leq \eta$} with probability at least $1-\delta$. Observe that $\uFF$ and $\FF$ can be coupled so that the two samples are always within $\eta$ in $\ell_1$ distance. When {$\norm{\hFF-\uFF}_{TV}\leq \eta$}, consider the coupling between $\hFF$ and $\FF$ by composing the optimal coupling between $\hFF$ and $\uFF$ and the coupling between $\uFF$ and $\FF$. Clearly, the two samples from $\hFF$ and $\FF$ are within $\ell_1$ distance $\eta$ with probability at least $1-\eta$. Due to Theorem~\ref{thm:prokhorov characterization}, the existence of this coupling implies that {$\norm{\hFF-\FF}_{P}\leq \eta$}. 
\end{proof}

\begin{prevproof}{Theorem}{thm:multi-item auction sample}
We only consider the case, where $\eta\leq \alpha\cdot \min\left\{{\varepsilon\over n}, {\varepsilon^2\over n^2m\LL H}\right\}$. $\alpha$ is an absolute constant and we will specify its choice in the end of the proof.

In light of Lemma~\ref{lem:discretized product TV}, we take $N=O\left({m^3H\over \sigma^3}\cdot(\log {n\over \delta}+\log m)\right)$ from $\hDD$ and learn a distribution $\DD$ so that, with probability at least $1-\delta$, $\norm{\DD_i-\hDD_i}_P\leq \sigma$ for all $i\in[n]$. According to Theorem~\ref{thm:approximation preserving}, we can transform $M$ into mechanism $\hM$ that is $O\left( nm\LL H \sigma+m\LL\sqrt{n H\sigma} \right)$-BIC w.r.t. $\hDD$ and IR. Choose $\sigma$ in a way so that $\hM$ is $\eta$-BIC w.r.t. $\hDD$. Moreover, $\hM$'s revenue under truthful bidding satisfies 	
$$\revT\left(\hM,\hDD\right)\geq c\cdot \opt_{\eta}\left(\hDD\right)-O\left(n\eta+n\sqrt{m\LL H\eta}\right).$$ If we choose $\alpha$ to be sufficiently small, then $$\revT\left(\hM,\hDD\right)\geq c\cdot \opt_{\eta}\left(\hDD\right)-\varepsilon.$$

When there is only a single-bidder, we can apply Lemma~\ref{lem:eps-IC to IC} to transform $\hM$ to an IC and IR mechanism, whose revenue satisfies the guarantee in the statement.

\notshow{
	We take $N=O\left({m^3H\over \sigma^3}\cdot(\log 1/\delta+\log m)\right)$ samples from $\round{\hDD_i}_{0,\sigma/m}$, and let $\DD_i$ be the empirical distribution over the samples. By Lemma~\ref{lem:discretized product TV}, $\norm{\DD_i-\round{\hDD_i}_{0,\sigma/m}}_{TV} \leq \sigma$ with probability $1-\delta$. We use $\hlDD$ to denote $\bigtimes_{i=1}^n \round{\hDD_i}_{0,\sigma/m}$ and  choose $\sigma = \frac{\tau^2}{n^3m^2L^2H^2}$. By Theorem~\ref{thm:how close do we need to learn}, we can construct a mechanism $M_1$ based on $M$ that is $\alpha=O(\frac{\tau^2}{nmLH})$-BIC and IR for $\hlDD$. Since $M$ is 
	c-approximate to $\text{OPT}(\DD)$, we also have revenue bound $$\revT(M_1,\hlDD ) \geq c\cdot\opt_{\alpha}(\hlDD) - \tau .$$
	
Next, we construct $\hM$ based Lemma~\ref{lem:discrete to continuous}. Since $M_1$ is $\alpha$-BIC and IR with respect to $\hlDD$, $\hM$ that is $(\alpha+3\L{\sigma})$-BIC and IR with respect to $\hDD$ and 
$$\revT(\hM,\hDD) \geq \revT(M_1,\hlDD)-n\L{\sigma} \geq c \cdot \opt_\alpha(\hlDD) - \tau-n\L{\sigma}.$$

Let $M^*$ be the optimal BIC and IR mechanism w.r.t. $\hDD$. According to Lemma~\ref{lem:original to rounded}, 
there exists a $O(\L \sigma)$-BIC and IR mechanism $M'$ such that $$\revT(M',\hlDD)\geq \rev(M^*,\hDD)-n\L\sigma=\opt(\hDD)-n\L\sigma.$$ Note that $\alpha>\L\sigma$, so $\revT(M',\hlDD)\leq \opt_\alpha(\hlDD)$. Hence, $$\revT(\hM,\hDD) \geq c \cdot \opt_\alpha(\hlDD) - \tau-n\L{\sigma}\geq  c \cdot \opt(\hDD) - \tau-(c+1)n\L{\sigma}.$$ 

Finally, $\opt(\hDD)\geq \opt_{\alpha+3\L\sigma}(\hDD)-2n\sqrt{m\L H(\alpha+3\L\sigma)}$ by Lemma~\ref{lem:eps-BIC to BIC}. Let $\eta=\alpha+3\L\sigma$. To sum up, $\hM$ is a $\eta$-BIC mechanism and 
$$\revT(\hM,\hDD)\geq c\cdot \opt_{\eta}(\hDD)-\tau-(c+1)n\L{\sigma}-2n\sqrt{m\L H\eta}.$$

Choose $\tau$ appropriately so that $\varepsilon=\tau+(c+1)n\L{\sigma}+2n\sqrt{m\L H\eta}$. In particular, $\eta=\frac{\varepsilon^2}{\poly(n,m,\L,H)}$ and $N=O(\frac{\poly(n,m,\L,H,\log(1/\delta))}{\varepsilon^6})$.
}
\end{prevproof}

\section{Missing Proofs from Section~\ref{sec:new sampling}}\label{appx:sample}

\subsection{Proof of Theorem~\ref{thm:sample complexity MRF finite alphabet}} \label{sec:MRFs}

\begin{prevproof}{Theorem}{thm:sample complexity MRF finite alphabet}
For the purposes of this proof we take $n=|V|$. We first prove the finite alphabet case, we then extend the result to the infinite alphabet case, and finally we discuss how to accommodate latent variables.

\paragraph{Finite alphabet $\Sigma$:}We will prove our first sample complexity bound by constructing an $\varepsilon$-cover, in total variation distance, of the set $\cal P$ of all MRFs with hyperedges of size at most $d$. We can assume that all $p \in {\cal P}$ satisfy the following:
\begin{itemize}

\item[$(A1):$] $p$ is defined on the hypergraph $G=(V,E)$, whose edge set is $E={V \choose d}$, and all its node potential functions are constant and equal  $1$.
\end{itemize}

\noindent The reason we can assume $(A1)$ for all $p\in {\cal P}$ is that potentials of nodes and smaller-size hyperedges can always be incorporated into the potentials of some size-$d$ hyperedge that contains them, and the potentials of size-$d$ hyperedges that are not present can always be taken to be constant $1$ functions.\\ 

Moreover, we can assume the following property for all MRFs $p \in {\cal P}$:
\begin{itemize}
\item[$(A2)$:] $\max_{\sigma \in \Sigma^e}{\psi_e(\sigma)}=1, \forall e\in E$.
\end{itemize}
The reason we can assume $(A2)$ for all $p\in {\cal P}$ is that the density of an MRF is invariant to multiplying any single potential function by some scalar.

\bigskip Now, given some MRF $p \in {\cal P}$, satisfying $(A1)$ and $(A2)$, which we can assume without loss of generality, we will make a sequence of transformations to arrive at some MRF $p'' \in {\cal P}$ such that $\norm{p-p''}_{TV} \le\varepsilon$ and~$p''$ can be described using $B={\rm poly}\left(|E|, |\Sigma|^d, \log({1\over \varepsilon})\right)$ bits. This, in turn, will imply that there exists an $\varepsilon$-cover ${\cal P}' \subset {\cal P}$ that has size $2^B$, and the existence of an $\varepsilon$-cover of this size implies that $O(B/\varepsilon^2)$-many samples from any $p \in {\cal P}$ suffice to learn some $q \in {\cal P}$ such that $\norm{p-q}_{TV} \le O(\varepsilon)$, using a tournament-style  density estimation algorithm; see~e.g.~\cite{devroye2012combinatorial,DaskalakisK14,AcharyaJOS14} and their references. 

\medskip Here are the steps to transform an arbitrary $p \in {\cal P}$ into some $p'' \in {\cal P}$ of low bit complexity:
\begin{itemize}

\item ({\bf Notation}:) From now on we will use $\hat{p}$ to denote unnormalized densities. I.e. if $p$ is defined in terms of potential functions $(\psi^p_e(\cdot))_{e\in E}$, then $\hat{p}(x)= \prod_{e\in E}\psi^p_e(x_e), \forall x \in \Sigma^V$.

\item ({\bf Step 1:}) Given some arbitrary $p \in {\cal P}$, we  construct some $p' \in {\cal P}$ such that $\norm{p-p'}_{TV}\le \varepsilon$, $p'$  satisfies $(A1)$ and $(A2)$ and, moreover, the unnormalized density of $p'$ satisfies that, for all $x \in \Sigma^V$, $\hat{p}'(x) = \left(1+{\varepsilon \over 2 n^d}\right)^{i_x}$, for some integer $i_x$. The existence of such $p'$ follows from the invariance of MRFs with respect to multiplying their potential functions by scalars, and the following.

\begin{claim} \label{claim:can round potentials}
Suppose $p, p' \in {\cal P}$ satisfy $(A1)$ and are defined in terms of potential functions $(\psi^p_e)_e$ and  $(\psi^{p'}_e)_e$ respectively. Moreover, suppose that $\forall e, \sigma \in \Sigma^e:$ $$\psi^{p'}_e(\sigma) \le \psi^{p}_e(\sigma) \le \left(1+{\varepsilon\over 2 n^d}\right)\psi^{p'}_e(\sigma).$$ Then $\norm{p-p'}_{TV}\le \varepsilon$.
\end{claim}
\begin{prevproof}{Claim}{claim:can round potentials} It follows from the condition in the statement of the claim that, for all $x \in \Sigma^V$:
$$\hat{p}'(x) \le \hat{p}(x) \le \left(1+{\varepsilon \over 2 n^d}\right)^{n \choose d} \hat{p}'(x) \le e^{\varepsilon/2}\hat{p}'(x) \le (1+\varepsilon) \hat{p}'(x).$$

Using the above, let us compare the normalized densities. For all $x\in \Sigma^V$:
$$p(x)={\hat{p}(x) \over \sum_{y} \hat{p}(y)} \le {\hat{p}'(x) (1+\varepsilon) \over \sum_{y} \hat{p}'(y)} \le p'(x) (1+\varepsilon).$$
Moreover,
$$p(x)={\hat{p}(x) \over \sum_{y} \hat{p}(y)} \ge {\hat{p}'(x) \over \sum_{y} \hat{p}'(y){ (1+\varepsilon)}} \ge p'(x) /(1+\varepsilon).$$
Using the above, let us bound the total variation distance between $p$ and $p'$:
\begin{align*}
\norm{p-p'}_{TV} &= {1 \over 2} \sum_x |p(x)-p'(x)|\\
 &={1 \over 2} \sum_{x: p(x) \ge p'(x)} (p(x)-p'(x))+ {1 \over 2} \sum_{x: p(x) < p'(x)} (p'(x)-p(x))\\
&\le{1 \over 2} \sum_{x: p(x) \ge p'(x)} \varepsilon p'(x) + {1 \over 2} \sum_{x: p(x) < p'(x)} \varepsilon p(x) \le \varepsilon.
\end{align*}
\end{prevproof}
\item ({\bf New Notation:}) We introduce some further notation. Let $\left(\psi_e^{p'}\right)_e$ be the potential functions defining distribution $p' \in {\cal P}$ from Step 1. We reparametrize these potential functions as follows:
$$\forall e, x \in \Sigma^e: \xi^{p'}_e(x) \equiv \log\left(\psi^{p'}_e(x)\right)/\log\left(1+{\varepsilon \over 2 n^d}\right).$$
Given the definition of $p'$ in Step 1, our new potential functions satisfy the following linear equations:
\begin{align}
\forall x\in \Sigma^V: \sum_{e \in E}\xi^{p'}_e(x_e) = i_x, \label{eq: LP1}
\end{align}
where, because of Assumption $(A2)$, satisfied by $p'$, the integers $i_x \le 0$, for all $x$. 

\item ({\bf Step 2:}) {We define $p''$ by setting up a linear program with variables $\xi_e^{p''}(x_e), \forall e, x_e\in\Sigma^E$. In particular, the number of variables of the linear program we are about to write is $L = |E|\cdot |\Sigma|^d$. To define our linear program, we first define $x^*=\argmax_{x} i_x$, and partition $\Sigma^V$ into two sets $\Sigma^V = {\cal G} \sqcup {\cal B}$, by taking ${\cal G}=\{x~|~i_x \ge i_{x^*}-T\}$, and $\cal B$ the complement of $\cal G$, for $T= {4 n^{d} \over \varepsilon}(n \log|\Sigma|+\log({1 \over \varepsilon}))$. In particular, all configurations in $\cal B$ have probability $p'(x) \le \varepsilon/|\Sigma|^n$. Our goal is to exhibit that there exists $p'' \in {\cal P}$ that (i) satisfies properties $(A1)$ and $(A2)$; (ii) can be described with ${\rm poly}\left(|E|, |\Sigma|^d, \log({1\over \varepsilon})\right)$ bits; 
 and (iii) satisfies $\sum_{x\in \cal B} p
''(x)\leq \varepsilon$ and $p''(x)=p'(x)\cdot (1+\delta)~ \forall x \in {\cal G}$, where $ \delta \in \left[-\varepsilon, {\varepsilon\over 1-\varepsilon}\right]$. We note that (iii) implies that $\norm{p'-p''}_{TV} \le \varepsilon$, as either $p''(x)\geq p'(x)$ for all $x\in {\cal G}$ simultaneously or $p''(x)< p'(x)$ for all $x\in {\cal G}$ simultaneously, and the total mass in $\cal B$ under both $p'$ and $p''$ are at most $\varepsilon$. Combining (iii) and Claim~\ref{claim:can round potentials}, we have (iv) $\norm{p-p''}_{TV} \le 2\varepsilon$. To exhibit the existence of $p''$ we write the following linear program:
\begin{align}
\forall x\in {\cal G}\setminus \{ x^* \}: &\sum_{e \in E}\xi^{p''}_e(x_e) - \sum_{e \in E}\xi^{p''}_e(x^*_e)= i_x-i_{x^*} \label{eq: LP2}\\
\forall x\in {\cal B}: &\sum_{e \in E}\xi^{p''}_e(x_e) - \sum_{e \in E}\xi^{p''}_e(x^*_e) \le -T \notag
\end{align}}
Note that, because LP~\eqref{eq: LP1} is feasible, it follows that LP~\eqref{eq: LP2} is feasible as well. Moreover, the coefficients and constants of LP~\eqref{eq: LP2} have {absolute value less than $T$ and} bit complexity polynomial in $d$, $\log n$, $\log({1 \over \varepsilon})$ and $\log \log |\Sigma|$, and the number of variables of this LP is $L=|E|\cdot |\Sigma|^d$. From the theory of linear programming it follows that there exists a solution to  LP~\eqref{eq: LP2} of bit complexity polynomial in $|E|$, $|\Sigma|^d$, $\log n$, and $\log({1 \over \varepsilon})$. {Why is (iii) true? It is not hard to see that for any $x\in {\cal B}$, $p''(x)\leq \varepsilon / |\Sigma|^n$ due to the second type of constraints in LP~\eqref{eq: LP2}. For any $x\in {\cal G}\setminus \{ x^* \}$, ${p''(x)\over p''(x^*)} = {p'(x)\over p'(x^*)} $ due to the first type of constraints in LP~\eqref{eq: LP2}, so $p''(x)=p'(x)\cdot (1+\delta)~ \forall x \in {\cal G}$ for some constant $\delta$. Since both $\sum_{x\in{\cal G}} p'(x)$ and $\sum_{x\in{\cal G}} p''(x)$ lie in $[1-\varepsilon,1]$,   $ \delta$ lies in $\left[-\varepsilon, {\varepsilon\over 1-\varepsilon}\right]$.}
\end{itemize}

To summarize the above (setting $\varepsilon \leftarrow \varepsilon/2$ in the above derivation), given an arbitrary $p \in {\cal P}$ we can construct $p'' \in {\cal P}$ such that: $p''$ can be described using $B={\rm poly}\left(|E|, |\Sigma|^d, \log({1\over \varepsilon})\right)$ bits---by specifying the low complexity solution $\left(\xi^{p''}_e\right)_e$ to LP~\eqref{eq: LP2}, and $p''$ satisfies $\norm{p-p''}_{TV} \le \varepsilon$. As we have noted above, the existence of such $p''$ for every $p \in {\cal P}$ implies the existence of an $\varepsilon$-cover, in total variation distance, of ${\cal P}$ that has size $2^B$, and tournament-style arguments imply then that any $p \in {\cal P}$ can be learned to within $O(\varepsilon)$ in total variation distance from $O({B\over \varepsilon^2})$-many samples, i.e.~from ${{\rm poly}\left(|E|, |\Sigma|^d, \log({1\over \varepsilon})\right) \over \varepsilon^2}$-many samples.

We now prove the second part of the statement. If the hypergraph $(V,E_p)$ with respect to which $p$ is defined is known, we redo the above argument, except we take ${\cal P}$ to be all MRFs defined on the graph $G=(V,E)$, where $E$ is the union of $E_p$ and all singleton sets corresponding to the nodes $V$.

For the third part of the statement, we note that an arbitrary distribution $p$ on $d$ variables, each taking values in $\Sigma$, can be expressed as a MRF with maximum hyperedge-size $d$. As such, it is folklore (see e.g.~\cite{devroye2012combinatorial}) that $\Omega(|\Sigma|^d/\varepsilon^2)$ samples are necessary to learn $p$ to within $\varepsilon$ in total variation distance. This completes the proof for the finite alphabet case.

Next, we show how to extend our sample complexity to the case where the alphabet $\Sigma=[0,H]$.

\paragraph{Alphabet $\Sigma=[0,H]$:} Let $\delta= {\varepsilon\over 8d\CC (n+1)^d}$, and $\Sigma_\delta$ be the set of all multiples of $\delta$ between $0$ and $H$.\footnote{We further assume that $H$ is a multiple of $\delta$. If not, let $k$ be the integer such that $\delta \in \left [{H\over 2^{k}}, {H\over 2^{k-1}}\right]$, and change $\delta$ to be ${H\over 2^{k}}$.} We first define distribution $\tilde{p}$ to be the rounded version of $p$ using the following coupling. For any sample $x$ drawn from $p$, create a sample $\tilde{x}$ drawn from $\tilde{p}$ such that $\tilde{x}_v=\left\lfloor{x_v\over \delta }\right\rfloor \cdot \delta$ for every $v\in V$. Note that (i) this coupling makes sure that the two samples from $p$ and $\tilde{p}$ are always within $\varepsilon$ of each other in $\ell_1$-distance. Our plan is to show that we can (ii) learn an MRF $q$ with polynomially many samples from distribution $\tilde{p}$ such that  $\norm{q-\tilde{p}}_{TV}= O(\varepsilon)$. Why does this imply our statement? First, we can generate a sample from $\tilde{p}$ using a sample from $p$ due to the coupling between the two distributions. Second, $\norm{q-\tilde{p}}_{TV}= O(\varepsilon)$ means that we can couple $q$ and $\tilde{p}$ in a way that the two samples are the same with probability at least $1-O(\varepsilon)$. Composing this coupling with the coupling between $\tilde{p}$ and $p$, we have a coupling between $p$ and $q$ so that the two samples are within $\varepsilon$ of each other in $\ell_1$-distance with probability at least $1-O(\varepsilon)$. According to Theorem~\ref{thm:prokhorov characterization}, $\norm{p-q}_P= O(\varepsilon)$. Now, we focus on proving (ii).

	We separate the proof into two steps. In the first step, we show that for any $\tilde{p}$, there is a discretized MRF $q'$ supported on $\Sigma_\delta^V$ with hyperedges of size at most $d$ such that $\norm{\tilde{p}-q'}_{TV}\leq \varepsilon$ and $q'$ can be described with $B={\rm poly}\left(|E|, |\Sigma_\delta|^d, \log({1\over \varepsilon})\right)$ bits. In other words, there is a $2^B$-sized $\varepsilon$-cover over all  possible distributions~$\tilde{p}$. In the second step, we show how to learn an MRF $q$ with $O(B/\varepsilon^2)$ samples from $\tilde{p}$ using a tournament-style density estimation algorithm; see~e.g.~\cite{devroye2012combinatorial,DaskalakisK14,AcharyaJOS14} and their references. Before we present the two steps of our proof, and in order to simplify our notation and avoid carrying around node potentials, let us introduce into the edge set $E$ of our hypergraph a singleton edge for every node $v$, and take the potential of every such edge $e=\{v\}$ to equal the node potential of node $v$.
	
	\begin{itemize}
		\item \textbf{(Step 1:)} We first define a discrete MRF $p'$ on the same graph $G= (V,E)$ as $p$ with alphabet $\Sigma_\delta$.  Distribution $p'$ is defined by choosing its log-potential $\phi_e^{p'}(x_e)$ to be exactly $\phi_e^p(x_e)$ for every hyperedge $e\in E$ and every possible value $x_e\in \Sigma_\delta^e$. 
		Next, we show that (iii) $\norm{p'-\tilde{p}}_{TV}\leq \varepsilon/2$. 
	
We use $A_x$ to denote the $n$-dimensional cube $\bigtimes_{v\in V}[x_v,x_v+\delta)$ for any $x\in \Sigma_\delta^V$. Note that $$\tilde{p}(x)={\int_{A_x} \exp\left ({\sum_e \phi^p_e(y_e)}\right) dy \over \int_{[0,H]^n} \exp\left({\sum_e \phi^p_e(y_e)}\right) dy}\leq {\delta^n \exp\left ({\sum_e \phi^p_e(x_e)}\right)\cdot \exp(d|E| \CC\delta)  \over \delta^n  \sum_{y\in \Sigma_\delta^V}\exp\left ({\sum_e \phi^p_e(y_e)}\right)\cdot \exp(-d|E| \CC\delta)}\leq p'(x) (1+\varepsilon/2).$$ The first inequality is due the $\CC$-Lipschitzness of the log potential functions and the second inequality is due to the definition of $\delta$.
	Similarly,
	 $$\tilde{p}(x)={\int_{A_x} \exp\left ({\sum_e \phi^p_e(y_e)}\right) dy \over \int_{[0,H]^n} \exp\left({\sum_e \phi^p_e(y_e)}\right) dy}\geq {\delta^n \exp\left ({\sum_e \phi^p_e(x_e)}\right)\cdot \exp(-d|E| \CC \delta)  \over \delta^n  \sum_{y\in \Sigma_\delta^V}\exp\left ({\sum_e \phi^p_e(y_e)}\right)\cdot \exp(d|E| \CC\delta)}\geq {p'(x) \over 1+\varepsilon/2}.$$
	We complete the proof of (iii) by combining the two inequalities.
	\begin{align*}
\norm{\tilde{p}-p'}_{TV} &= {1 \over 2} \sum_{x\in\Sigma_\delta^V} |\tilde{p}(x)-p'(x)|\\
 &={1 \over 2} \sum_{x: \tilde{p}(x) \ge p'(x)} (\tilde{p}(x)-p'(x))+ {1 \over 2} \sum_{x: \tilde{p}(x) < p'(x)} (p'(x)-\tilde{p}(x))\\
&\le{1 \over 2} \sum_{x: \tilde{p}(x) \ge p'(x)} {\varepsilon\over 2} p'(x) + {1 \over 2} \sum_{x: \tilde{p}(x) < p'(x)} {\varepsilon\over 2} \tilde{p}(x) \le \varepsilon/2.
\end{align*}
  Let $\cal P$ be the set of all MRFs  with hyperedges of size at most $d$ and alphabet $\Sigma_\delta$. By redoing Step 1 and 2 of the proof for the finite alphabet case, we can show that (iv) for any $\hat{p}\in {\cal P}$, there exists another $\hat{p}' \in {\cal P}$ describable with $B={\rm poly}\left(|E|, |\Sigma_\delta|^d, \log({1\over \varepsilon})\right)$ bits such that $\norm{\hat{p}-\hat{p}'}_{TV}\leq \varepsilon/2$. Since $p'\in {\cal P}$, there exists a $q'\in {\cal P}$ describable with $B$ bits such that $\norm{p'-q'}_{TV}\leq \varepsilon/2$. Combining this inequality with (iii), we have $\norm{\tilde{p}-q'}_{TV}\leq \varepsilon$.

 \item \textbf{(Step 2:)} Let ${\cal P'}\subset {\cal P}$ be the set of all MRFs in $\cal P$ with bit complexity at most $B$ from Step 1. Since $\min_{\tilde{q} \in {\cal P'}} \norm{\tilde{q}-\tilde{p}}_{TV}\leq \varepsilon$, we can learn an MRF $q \in \cal P'$ such that $\norm{q-\tilde{p}}_{TV}\leq O(\varepsilon)$ with $O(B/\varepsilon^2)$ samples from $\tilde{p}$ using a tournament-style density estimation algorithm~\cite{devroye2012combinatorial,DaskalakisK14,AcharyaJOS14}.	\end{itemize}
To sum up, we can learn an MRF $q$ such that $\norm{q-{p}}_{P}\leq \varepsilon$ with ${{\rm poly}\left(|V|^{d^2}, {\left ( H\over \varepsilon\right)}^d,\CC^d \right) }$ many samples from~$p$. If the graph $G$ on which $p$ is defined is known, we can choose $\delta$ to be $O\left( {\varepsilon\over 8d\CC |E|}\right)$ and improve the sample complexity to ${{\rm poly}\left(|V|, |E|^{d}, {\left ( H\over \varepsilon\right)}^d,\CC^d \right)}$.

\paragraph{Latent Variable Models:} Finally, we consider the case where only $k$ out of the $n$ variables of the MRF are observable. Let $S$ be the set of observable variables, and  use $p_S$ to denote the marginal  of $p$ on these variables. We will first consider the finite alphabet case. Consider the $\varepsilon$-cover we constructed earlier. We argued that for any MRF $p$ there exists an MRF $q$ in the cover  such that $\norm{p-q}_{TV}\leq \varepsilon$. For that $q$ we clearly also have $\norm{p_S-q_S}_{TV}\leq \varepsilon$. The issue is that we do not know for a given $q$ in the cover which subset of its variables set $S$ might correspond to. But this is not a big deal. We can  use our cover to generate an $\varepsilon$-cover  of all possible marginals  $p_S$ of all possible MRFs $p$ as follows. Indeed, for any $q'$ in the original $\varepsilon$-cover, we include in the new cover the marginal distribution $q'_{S'}$ of every possible subset $S'$ of its variables of size $k$. This increases the size of our original cover by a multiplicative factor of at most $n^k$. As a result, the number of samples required for the tournament-style density estimation algorithm to learn a good distribution  increases by a multiplicative factor of $k\log n$. For the infinite alphabet case, our statement follows from applying the same modification to the $\varepsilon$-cover of $\tilde{p}$.
\end{prevproof}

\notshow{
\begin{theorem}[Learnability of Continuous MRFs in Prokhorov Distance] \label{thm:sample complexity MRF continuous}
Suppose we are given sample access to an MRF $p$, as in Definition~\ref{def:MRF}, defined on an unknown graph with hyperedges of size at most $d$. Suppose that $\Sigma=[0,H]$ and the log potential $\phi_e^p(\cdot) \equiv \log \left(\psi_e^p(\cdot) \right)$ for every edge $e$ is $\CC$-Lipschitz in the $\ell_1$-norm. Given ${{\rm poly}\left(|V|^{d^2}, {\left ( H\over \varepsilon\right)}^d,\CC^d \right) \over \varepsilon^2}$  samples from $p$ we can learn some MRF~$q$ whose hyper-edges also have size at most $d$ such that $\norm{p-q}_{P} \le \varepsilon$. If the graph on which $p$ is defined is known, then${{\rm poly}\left(|V|, |E|^{d}, {\left ( H\over \varepsilon\right)}^d,\CC^d \right) \over \varepsilon^2}$-many samples suffice. 
\end{theorem}

\begin{prevproof}{Theorem}{thm:sample complexity MRF continuous}
	Let $n=|V|$, $\delta= {\varepsilon\over 8d\CC n^d}$, and $\Sigma_\delta$ be the set of all multiples of $\delta$ between $0$ and $H$.~\footnote{We further assume that $H$ is a multiple of $\delta$. If not, let $k$ be the integer such that $\delta \in \left [{H\over 2^{k}}, {H\over 2^{k-1}}\right]$, and change $\delta$ to be ${H\over 2^{k}}$.} We first define distribution $\tilde{p}$ to be the rounded version of $p$ using the following coupling. For any sample $x$ drawn from $p$, create a sample $\tilde{x}$ drawn from $\tilde{p}$ such that $\tilde{x}_v=\left\lfloor{x_v\over \delta }\right\rfloor \cdot \delta$ for every $v\in V$. Note that (i) this coupling makes sure that the two samples from $p$ and $\tilde{p}$ are always within $\varepsilon$ of each other in $\ell_1$-distance. Our plan is to show that we can (ii) learn an MRF $q$ with polynomially many samples from distribution $\tilde{p}$ such that  $\norm{q-\tilde{p}}_{TV}= O(\varepsilon)$. Why does this imply our claim? First, we can generate a sample from $\tilde{p}$ using a sample from $p$ due to the coupling between the two distributions. Second, $\norm{q-\tilde{p}}_{TV}= O(\varepsilon)$ means that we can couple $q$ and $\tilde{p}$ in a way that the two samples are the same with probability at least $1-O(\varepsilon)$. Composing this coupling with the coupling between $\tilde{p}$ and $p$, we have a coupling between $p$ and $q$ so that the two samples are at most $\varepsilon$ away from each other in$\ell_1$-distance with probability at least $1-O(\varepsilon)$. According to Theorem~\ref{thm:prokhorov characterization}, $\norm{p-q}_P= O(\varepsilon)$. Now, we focus on proving (ii).

	We separate the proof into two steps. In the first step, we show that for any $\tilde{p}$, there is a discretized MRF $q'$ supported on $\Sigma_\delta^V$ with hyperedges of size at most $d$ such that $\norm{\tilde{p}-q'}_{TV}\leq \varepsilon$ and $q'$ can be described with $B={\rm poly}\left(|E|, |\Sigma_\delta|^d, \log({1\over \varepsilon})\right)$ bits. In other words, there is a size $2^B$ $\varepsilon$-cover over all the possible distribution $\tilde{p}$. In the second step, we show how to learn an MRF $q$ with $O(B/\varepsilon^2)$ samples from $\tilde{p}$ using a tournament-style density estimation algorithm; see~e.g.~\cite{devroye2012combinatorial,DaskalakisK14,AcharyaJOS14} and their references. 
	
	\begin{itemize}
		\item \textbf{(Step 1:)} We first define a discrete MRF $p'$ on the same graph $G= (V,E)$ as $p$ with alphabet $\Sigma_\delta$. More specifically, we define the potential $\phi_e^{p'}(x_e)$ to be exactly $\phi_e^p(x_e)$ for every hyperedge $e\in E$ and every possible value $x_e\in \Sigma_\delta^e$. 
		Next, we show that (iii) $\norm{p'-\tilde{p}}_{TV}\leq \varepsilon/2$. 
	
We use $A_x$ to denote the $n$-dimensional cube $\bigtimes_{v\in V}[x_v,x_v+\delta)$ for any $x\in \Sigma_\delta^V$. Note that $$\tilde{p}(x)={\int_{A_x} \exp\left ({\sum_e \phi^p_e(y_e)}\right) dy \over \int_{[0,H]^n} \exp\left({\sum_e \phi^p_e(y_e)}\right) dy}\leq {\delta^n \exp\left ({\sum_e \phi^p_e(x_e)}\right)\cdot \exp(d|E| \CC\delta)  \over \delta^n  \sum_{y\in \Sigma_\delta^V}\exp\left ({\sum_e \phi^p_e(y_e)}\right)\cdot \exp(-d|E| \CC\delta)}\leq p'(x) (1+\varepsilon/2).$$ The first inequality is due the $\CC$-Lipschitzness of the potential functions and the second inequality is due to the definition of $\delta$.
	Similarly,
	 $$\tilde{p}(x)={\int_{A_x} \exp\left ({\sum_e \phi^p_e(y_e)}\right) dy \over \int_{[0,H]^n} \exp\left({\sum_e \phi^p_e(y_e)}\right) dy}\geq {\delta^n \exp\left ({\sum_e \phi^p_e(x_e)}\right)\cdot \exp(-d|E| \CC \delta)  \over \delta^n  \sum_{y\in \Sigma_\delta^V}\exp\left ({\sum_e \phi^p_e(y_e)}\right)\cdot \exp(d|E| \CC\delta)}\geq p'(x) (1-\varepsilon/2).$$
	We complete the proof of (ii) by combining the two inequalities.
	\begin{align*}
\norm{\tilde{p}-p'}_{TV} &= {1 \over 2} \sum_{x\in\Sigma_\delta^V} |\tilde{p}(x)-p'(x)|\\
 &={1 \over 2} \sum_{x: \tilde{p}(x) \ge p'(x)} (\tilde{p}(x)-p'(x))+ {1 \over 2} \sum_{x: \tilde{p}(x) < p'(x)} (p'(x)-\tilde{p}(x))\\
&\le{1 \over 2} \sum_{x: \tilde{p}(x) \ge p'(x)} {\varepsilon\over 2} p'(x) + {1 \over 2} \sum_{x: \tilde{p}(x) < p'(x)} {\varepsilon\over 2} p(x) \le \varepsilon/2.
\end{align*}
  Let $\cal P$ be the set of all MRFs of with hyperedges of size at most $d$ and alphabet $\Sigma_\delta$. By redoing step 1 and 2 of the proof for Theorem~\ref{thm:sample complexity MRF finite alphabet}, we can show that (iv) for any $\hat{p}\in {\cal P}$, there exists another $\hat{p}' \in {\cal P}$ describable with $B={\rm poly}\left(|E|, |\Sigma_\delta|^d, \log({1\over \varepsilon})\right)$ bits such that $\norm{\hat{p}-\hat{p}'}_{TV}\leq \varepsilon/2$. Since $p'\in {\cal P}$, there exists a $q'\in {\cal P}$ describable with $B$ bits such that $\norm{p'-q'}_{TV}\leq \varepsilon/2$. Combining this inequality with (iii), we have $\norm{\tilde{p}-q'}_{TV}\leq \varepsilon$.

 \item \textbf{(Step 2:)} Let ${\cal P'}\subset {\cal P}$ be the set of all MRFs in $\cal P$ with bit complexity at most $B$. Since $\min_{\tilde{q} \in {\cal P'}} \norm{\tilde{q}-\tilde{p}}_{TV}\leq \varepsilon$, we can learn an MRF $q \in \cal P'$ such that $\norm{q-\tilde{p}}_{TV}\leq O(\varepsilon)$ with $O(B/\varepsilon^2)$ samples from $\tilde{p}$ using a tournament-style density estimation algorithm; see~e.g.~\cite{devroye2012combinatorial,DaskalakisK14,AcharyaJOS14} and their references. 
	\end{itemize}
To sum up, we can learn an MRF $q$ such that $\norm{q-{p}}_{P}\leq \varepsilon$ with ${{\rm poly}\left(|V|^{d^2}, {\left ( H\over \varepsilon\right)}^d,\CC^d \right) \over \varepsilon^2}$ many samples from $p$. If the graph $G$ on which $p$ is defined is known, the sample complexity can be improved to ${{\rm poly}\left(|V|, |E|^{d}, {\left ( H\over \varepsilon\right)}^d,\CC^d \right) \over \varepsilon^2}$.
\end{prevproof}
}

\subsection{Proof of Theorem~\ref{thm:sample complexity Bayesnets finite alphabet}} \label{sec:Bayesnets}

\begin{prevproof}{Theorem}{thm:sample complexity Bayesnets finite alphabet}
We first prove the theorem statement for the finite alphabet case, we then extend it to the infinite alphabet case, and finally show how we can accommodate latent variables as well.
\paragraph{Finite alphabet $\Sigma$:} We prove the claims in the theorem statement in reverse order.

\medskip For the third part of the statement, we note that an arbitrary distribution $p$ on $d+1$ variables, each taking values in $\Sigma$, can be expressed as a Bayesnet with maximum indegree $d$. As such, it is folklore (see e.g.~\cite{devroye2012combinatorial}) that $\Omega(|\Sigma|^{d+1}/\varepsilon^2)$ samples are necessary to learn $p$ to within $\varepsilon$ in total variation distance.

\medskip To prove the second part of the statement, we show that there is an $\varepsilon$-cover, in total variation distance, of all Bayesnets ${\cal P}$ on a given DAG $G$ of indegree at most $d$, which has size $B=\left({n |\Sigma| \over \varepsilon}\right)^{n|\Sigma|^{d+1}}$, where $n =|V|$. The existence of an $\varepsilon$-cover of this size implies that $O(\log(B)/\varepsilon^2)$-many samples from any $p \in {\cal P}$ suffice to learn some $q \in {\cal P}$ such that $\norm{p-q}_{TV} \le O(\varepsilon)$, using a tournament-style  density estimation algorithm; see~e.g.~\cite{devroye2012combinatorial,DaskalakisK14,AcharyaJOS14} and their references. Thus, to prove the second part of the theorem statement it suffices to argue that an $\varepsilon$-cover of size $B$ exists. We prove the existence of this cover by exploiting the following lemma.


\begin{lemma} \label{lemma:discretizing a Bayesnet}
Suppose $p$ and $q$ are Bayesenets on the same DAG $G=(V,E)$. Suppose that, for all $v \in V$, for all $\sigma \in \Sigma^{\Pi(v)}$, where $\Pi(v)$ are the parents of $v$ in $G$ (using the same notation as in Definition~\ref{def:Bayesnet}), it holds that $$\norm{p_{X_v | X_{\Pi_v}=\sigma} - q_{X_v | X_{\Pi_v}=\sigma}}_{TV} \le {\varepsilon \over |V|}.$$
Then $\norm{p-q}_{TV} \le \varepsilon$.
\end{lemma}
\begin{prevproof}{Lemma}{lemma:discretizing a Bayesnet}
We employ a hybrid argument. First, let us denote $n=|V|$ and label the nodes in $V$ with labels $1,\ldots,n$ according to some topological sorting of $G$. In particular, the parents (if any) of any node $i$ have indices $<i$. Now, for our hybrid argument we construct the following auxiliary distributions, for $i=0,\ldots,n$:
$$h^{i}(x)  = \prod_{v=1}^{i} p_{X_v | X_{\Pi_v}} (x_v | x_{\Pi_v}) \prod_{v=i+1}^{n} q_{X_v | X_{\Pi_v}} (x_v | x_{\Pi_v}), \text{for all }x \in \Sigma^V.$$
In particular, $h^0 \equiv q$ and $h^n \equiv p$, and the rest are fictional distributions. By triangle inequality, we have that:
$$\norm{p-q}_{TV} \le \sum_{i=1}^n \norm{h^i-h^{i-1}}_{TV}.$$ 
We will bound each term on the RHS by $\varepsilon/n$ to conclude the proof of the lemma. Indeed,
\begin{align*}
&\norm{h^i-h^{i-1}}_{TV}\\
 =& \sum_{x} ~\vline\prod_{v=1}^{i} p_{X_v | X_{\Pi_v}} (x_v | x_{\Pi_v}) \cdot\prod_{v=i+1}^{n} q_{X_v | X_{\Pi_v}} (x_v | x_{\Pi_v}) -\prod_{v=1}^{i-1} p_{X_v | X_{\Pi_v}} (x_v | x_{\Pi_v}) \cdot \prod_{v=i}^{n} q_{X_v | X_{\Pi_v}} (x_v | x_{\Pi_v})~\vline\\
=& \sum_{x} ~\vline \prod_{v=1}^{i-1} p_{X_v | X_{\Pi_v}}(x_v | x_{\Pi_v}) \cdot \left(p_{X_i | X_{\Pi_i}}(x_i | x_{\Pi_i}) -q_{X_i | X_{\Pi_i}}(x_i | x_{\Pi_i}) \right )  \cdot \prod_{v=i+1}^{n} q_{X_v | X_{\Pi_v}} (x_v | x_{\Pi_v})~\vline\\
=& \sum_{x}  \prod_{v=1}^{i-1} p_{X_v | X_{\Pi_v}}(x_v | x_{\Pi_v}) \cdot ~\vline~p_{X_i | X_{\Pi_i}}(x_i | x_{\Pi_i}) -q_{X_i | X_{\Pi_i}}(x_i | x_{\Pi_i}) ~\vline  \cdot \prod_{v=i+1}^{n} q_{X_v | X_{\Pi_v}} (x_v | x_{\Pi_v})\\
=& \sum_{x_{1\ldots i-1}} \left( \prod_{v=1}^{i-1} p_{X_v | X_{\Pi_v}}(x_v | x_{\Pi_v}) \cdot ~\sum_{x_i}~\left(~\vline~p_{X_i | X_{\Pi_i}}(x_i | x_{\Pi_i}) -q_{X_i | X_{\Pi_i}}(x_i | x_{\Pi_i}) ~\vline  \cdot \sum_{x_{i+1\ldots n}} \left(\prod_{v=i+1}^{n} q_{X_v | X_{\Pi_v}} (x_v | x_{\Pi_v}) \right)\right)\right)\\
=& \sum_{x_{1\ldots i-1}} \left( \prod_{v=1}^{i-1} p_{X_v | X_{\Pi_v}}(x_v | x_{\Pi_v}) \cdot ~\sum_{x_i}~\left(~\vline~p_{X_i | X_{\Pi_i}}(x_i | x_{\Pi_i}) -q_{X_i | X_{\Pi_i}}(x_i | x_{\Pi_i}) ~\vline  ~ \right)\right)\\
\le & \sum_{x_{1\ldots i-1}} \left( \prod_{v=1}^{i-1} p_{X_v | X_{\Pi_v}}(x_v | x_{\Pi_v}) \cdot ~\varepsilon/n\right)\\
= & \varepsilon/n,
\end{align*}
where for the inequality we used the hypothesis in the statement of the lemma.\end{prevproof}

Now suppose $p\in {\cal P}$ is an arbitrary Bayesnet defined on $G$. It follows from Lemma~\ref{lemma:discretizing a Bayesnet} that $p$ lies $\varepsilon$-close in total variation distance to a Bayesnet $q$ such that, for all $v\in V$, and all $\sigma \in \Sigma^{\Pi_v}$, the conditional distribution $q_{X_v | X_{\Pi_v}=\sigma}$ is a discretized version of $p_{X_v | X_{\Pi_v}=\sigma}$ that is ${\varepsilon \over n}$-close in total variation distance. Note that $p_{X_v | X_{\Pi_v}=\sigma}$ is an element of the simplex over $|\Sigma|$ elements,  and it is easy to see that this simplex can be ${\varepsilon \over n}$-covered, in total variation distance, using a discrete set of at most $\left({n |\Sigma| \over \varepsilon}\right)^{|\Sigma|}$-many distributions. As there are at most $n \cdot |\Sigma|^d$ conditional distributions to discretize, a total number of  
$$B=\left({n |\Sigma| \over \varepsilon}\right)^{n|\Sigma|^{d+1}}$$
  discretized distributions suffice to cover all ${\cal P}$.
	
	\medskip To prove the first part of the theorem statement, we proceed in the same way, except that now that we do not know the DAG our cover will be larger. Since there are at most $n^{dn}$  DAGs of indegree at most $d$ on $n$ labeled vertices, and for each DAG there is a cover of all Bayesnets defined on that DAG of size at most $B$, as above, it follows that there is an $\varepsilon$-cover, in total variation distance, of all Bayesnets of indegree at most $d$ of size:
	$$n^{dn}\cdot B.$$
Given the bound on the cover size, the proof concludes by appealing to tournament-style density estimation algorithms, as we did earlier. This completes our proof for the finite alphabet case.


\paragraph{\textbf{Alphabet $\Sigma=[0,H]$:} }Let $\delta= {\varepsilon\over d\CC n}$, and $\Sigma_\delta$ be the set of all multiples of $\delta$ between $0$ and $H$.\footnote{We further assume that $H$ is a multiple of $\delta$. If not, let $k$ be the integer such that $\delta \in \left [{H\over 2^{k}}, {H\over 2^{k-1}}\right]$, and change $\delta$ to be ${H\over 2^{k}}$.} For any set of nodes $S$ and $x=(x_v)_{v\in S}$, we use $\round{x}_\delta$ to denote the corresponding rounded vector $\left(\lfloor {x_v\over \delta}\rfloor\cdot \delta\right)_{v\in S}$. We first define distribution $\tilde{p}$ to be the rounded version of $p$ using the following coupling. For any sample $x$ drawn from $p$, create a sample $\tilde{x}=\round{x}_\delta$ drawn from $\tilde{p}$. Note that (i) this coupling makes sure that the two samples from $p$ and $\tilde{p}$ are always within $\varepsilon$ of each other in $\ell_1$-distance. Our plan is to show that we can (ii) learn a Bayesnet $q$ with in-degree at most $d$ using polynomially many samples from distribution $\tilde{p}$ such that  $\norm{q-\tilde{p}}_{TV}= O(\varepsilon)$. Why does this imply our claim? First, we can generate a sample from $\tilde{p}$ using a sample from $p$ due to the coupling between the two distributions. Second, $\norm{q-\tilde{p}}_{TV}= O(\varepsilon)$ means that we can couple $q$ and $\tilde{p}$ in a way that the two samples are the same with probability at least $1-O(\varepsilon)$. Composing this coupling with the coupling between $\tilde{p}$ and $p$, we have a coupling between $p$ and $q$ such that the two samples are at most $\varepsilon$ away from each other in $\ell_1$-distance with probability at least $1-O(\varepsilon)$. This implies, according to Theorem~\ref{thm:prokhorov characterization}, that $\norm{p-q}_P= O(\varepsilon)$. Now, we focus on proving (ii) and separate the proof into three steps.

\begin{itemize}
	\item \textbf{(Step 1:)} We first prove that there is a Bayesnet $p''$ with in-degree at most $d$ and alphabet $\Sigma_\delta$ such that $\norm{\tilde{p}-p''}_{TV}\leq \varepsilon$. 
We first construct a Bayesnet $p'$ on the same DAG as $p$, where the conditional probability distribution for every node $v$, and $\sigma\in \Sigma^{\Pi_v}$ is defined as 
$$p'_{X_v|X_{\Pi(v)=\sigma}}\equiv p_{X_v|X_{\Pi(v)=\lfloor \sigma\rfloor_\delta}}.$$ Clearly, for any node $v$, and $\sigma\in \Sigma^{\Pi_v}$, 
$$\norm{p_{X_v|X_{\Pi(v)=\sigma}}-p'_{X_v|X_{\Pi(v)=\sigma}}}_{TV}=\norm{p_{X_v|X_{\Pi(v)=\sigma}}-p_{X_v|X_{\Pi(v)=\lfloor \sigma\rfloor_\delta }}}_{TV}\leq \CC\cdot\norm{\sigma-\round{\sigma}_\delta}_1\leq {\CC d\delta \leq {\varepsilon\over |V|}}.$$

Hence, Lemma~\ref{lemma:discretizing a Bayesnet} implies that:  (iii) $\norm{p-p'}_{TV}\leq \varepsilon$.\footnote{Even though Lemma~\ref{lemma:discretizing a Bayesnet} was only proved earlier for a finite alphabet, the same proof extends to when the alphabet is infinite.}

Next, we construct the rounded distribution $p''$ of $p'$ via the following coupling. For any sample $x'$ drawn from $p'$, create a sample $x''=\round{x'}_\delta$ from $p''$. It is not hard to verify that $p''$ can also be captured by a Bayesnet defined on the same DAG as $p$ and $p'$. In particular, for every node $v$, every $x_v\in \Sigma_\delta$, and $x_{\Pi_v}\in \Sigma_\delta^{\Pi_v}$, the conditional probability is $$p''_{X_v|X_{\Pi_v}}\left(x_v|x_{\Pi_v}\right )=\int_{x_v}^{x_v+\delta} p'_{X_v|X_{\Pi_v}}\left(z|x_{\Pi_v}\right ) dz.$$ As $p''$ is the rounded distribution of $p'$, $\tilde{p}$ is the rounded distribution of $p$, and $\norm{p-p'}_{TV}\leq \varepsilon$, it must be the case  that $\norm{p''-\tilde{p}}_{TV}\leq \varepsilon$.

\item \textbf{(Step 2:)} Let $\cal P$ be the set of all Bayesnets defined on a DAG with $n$ nodes and  in-degree at most $d$, and which 	have alphabet $\Sigma_\delta$. We argue that there is a size $A=n^{dn}\cdot\left({n |\Sigma_\delta| \over \varepsilon}\right)^{n|\Sigma_\delta|^{d+1}}$ $\varepsilon$-cover  $\cal P'$, in total variation distance, of $\cal P$, and $\cal P'\subset \cal P$. This follows from the same argument we did in the proof for the finite alphabet case. First, there are $n^{dn}$ different DAGs with $n$ nodes and in-degree at most $d$. Second, for each DAG there are at most $n \cdot |\Sigma_\delta|^d$ conditional distributions. Finally, it suffices to ${\varepsilon \over n}$-cover each conditional distribution, in total variation distance, 
 which can be accomplished by a discrete set of at most $\left({n |\Sigma_\delta| \over \varepsilon}\right)^{|\Sigma_\delta|}$-many distributions. Since $p''\in {\cal P}$ and $\norm{p''-\tilde{p}}_{TV}\leq \varepsilon$, there exists a Bayesnet $\tilde{q}$ from the $\varepsilon$-cover $\cal P'$ such that $\norm{\tilde{q}-\tilde{p}}_{TV}\leq 2\varepsilon$.

\item \textbf{(Step 3:)} Since $\min_{q'\in \cal P'} \norm{q'-\tilde{p}}_{TV} \leq 2\varepsilon$, we can use a tournament-style  density estimation algorithm (see~e.g.~\cite{devroye2012combinatorial,DaskalakisK14,AcharyaJOS14} and their references) to learn a Bayesnet $q\in \cal P'$ such that $\norm{q-\tilde{p}}_{TV}=O(\varepsilon)$ given $O\left( {\log A \over \varepsilon^2}\right)$ samples from $\tilde{p}$ . 
\end{itemize}

\noindent To sum up, we can learn a Bayesnet $q$ defined on a DAG with in-degree at most $d$ using $$O\left({ d |V|\log |V|  + |V| \cdot \left({H|V|d\CC\over \varepsilon}\right)^{d+1} \log\left({|V|{Hd\CC} \over \varepsilon}\right) \over \varepsilon^2}\right)$$ samples from $p$ such that $\norm{q-p}_P\leq \varepsilon$. If the DAG that $p$ is defined on is known, the sample complexity improves to $O\left({ |V| \cdot \left({H|V|d\CC\over \varepsilon}\right)^{d+1} \log\left({|V|{Hd\CC} \over \varepsilon}\right)\over \varepsilon^2}\right)$.

\paragraph{Latent Variable Model:} Finally, we consider the case where only $k$ out of the $n$ variables of the Bayesnet $p$ are observable. Let $S$ be the set of observable variables, and  use $p_S$ to denote the marginal  of $p$ on these variables. We will first consider the finite alphabet case. Consider the $\varepsilon$-cover we constructed earlier. We argued that for any Bayesnet $p$ there exists an Bayesnet $q$ in the cover  such that $\norm{p-q}_{TV}\leq \varepsilon$. For that $q$ we clearly also have $\norm{p_S-q_S}_{TV}\leq \varepsilon$. The issue is that we do not know for a given $q$ in the cover which subset of its variables set $S$ might correspond to. But this is not a big deal. We can  use our cover to generate an $\varepsilon$-cover  of all possible marginals  $p_S$ of all possible Bayesnets $p$ as follows. Indeed, for any $q'$ in the original $\varepsilon$-cover, we include in the new cover the marginal distribution $q'_{S'}$ of every possible subset $S'$ of its variables of size $k$. This increases the size of our original cover by a multiplicative factor of at most $n^k$. As a result, the number of samples required for the tournament-style density estimation algorithm to learn a good distribution  increases by a multiplicative factor of $k\log n$. For the infinite alphabet case, our statement follows from applying the same modification to the $\varepsilon$-cover of $\tilde{p}$.
\end{prevproof}

\notshow{\begin{theorem}[Learnability of Continuous Bayesnets in Prokhorov Distance] \label{thm:sample complexity Bayesnets infinite alphabet}
Suppose we are given sample access to a Bayesnet $p$, as in Definition~\ref{def:Bayesnet}, defined on an unknown DAG with in-degree at most $d$. Suppose also that $\Sigma=[0,H]$ and the conditional probability for every node $v$ is $\CC$-Lipschitz in the $\ell_1$-norm, that is, $\norm{p_{X_v | X_{\Pi_v}=\sigma} - p_{X_v | X_{\Pi_v}=\sigma'}}_{TV} \le\CC\cdot\norm{\sigma-\sigma'}_1$ for any $v$ and $\sigma,\sigma'\in \Sigma^{\Pi_v}$. Given $O\left({ d |V|\log |V| + |V| \cdot \left({H|V|d\CC\over \varepsilon}\right)^{d+1} \log\left({|V|{H|V|d\CC} \over \varepsilon^2}\right)\over \varepsilon^2}\right)$ -many samples from $p$ we can learn some Bayesnet~$q$ defined on a DAG whose in-degree is also bounded by $d$ such that $\norm{p-q}_{P} \le \varepsilon$. If the graph on which $p$ is defined is known, then $O\left({|V| \cdot \left({H|V|d\CC\over \varepsilon}\right)^{d+1} \log\left({|V|{H|V|d\CC} \over \varepsilon^2}\right)\over \varepsilon^2}\right)$ -many samples suffice. 
\end{theorem}

\begin{prevproof}{Theorem}{thm:sample complexity Bayesnets infinite alphabet}
	Let $n=|V|$, $\delta= {\varepsilon\over d\CC n}$, and $\Sigma_\delta$ be the set of all multiples of $\delta$ between $0$ and $H$. For any set of nodes $S$ and $x=(x_v)_{v\in S}$, we use $\round{x}_\delta$ to denote the corresponding rounded vector $\left(\lfloor {x_v\over \delta}\rfloor\cdot \delta\right)_{v\in S}$. We first define distribution $\tilde{p}$ to be the rounded version of $p$ using the following coupling. For any sample $x$ drawn from $p$, create a sample $\tilde{x}=\round{x}_\delta$ drawn from $\tilde{p}$. Note that (i) this coupling makes sure that the two samples from $p$ and $\tilde{p}$ are always within $\varepsilon$ of each other in $\ell_1$-distance. Our plan is to show that we can (ii) learn an Bayesnet $q$ with in-degree at most $d$ using polynomially many samples from distribution $\tilde{p}$ such that  $\norm{q-\tilde{p}}_{TV}= O(\varepsilon)$. Why does this imply our claim? First, we can generate a sample from $\tilde{p}$ using a sample from $p$ due to the coupling between the two distributions. Second, $\norm{q-\tilde{p}}_{TV}= O(\varepsilon)$ means that we can couple $q$ and $\tilde{p}$ in a way that the two samples are the same with probability at least $1-O(\varepsilon)$. Composing this coupling with the coupling between $\tilde{p}$ and $p$, we have a coupling between $p$ and $q$ such that the two samples are at most $\varepsilon$ away from each other in $\ell_1$-distance with probability at least $1-O(\varepsilon)$. According to Theorem~\ref{thm:prokhorov characterization}, $\norm{p-q}_P= O(\varepsilon)$. Now, we focus on proving (ii) and separate the proof into three steps.

\begin{itemize}
	\item \textbf{(Step 1:)} Prove that there is a Bayesnet $p''$ with in-degree at most $d$ and alphabet $\Sigma_\delta$ such that $\norm{p-p'}_{TV}\leq \varepsilon$. 
We first construct a Bayesnet $p'$ on the same DAG as $p$, where the conditional probability distribution for every node $v$, and $\sigma\in \Sigma^{\Pi_v}$ is defined as 
$$p'_{X_v|X_{\Pi(v)=\sigma}}\equiv p_{X_v|X_{\Pi(v)=\lfloor \sigma\rfloor_\delta}}.$$ Clearly, for any node $v$, and $\sigma\in \Sigma^{\Pi_v}$, 
$$\norm{p_{X_v|X_{\Pi(v)=\sigma}}-p'_{X_v|X_{\Pi(v)=\sigma}}}_{TV}=\norm{p_{X_v|X_{\Pi(v)=\sigma}}-p_{X_v|X_{\Pi(v)=\lfloor \sigma\rfloor_\delta }}}_{TV}\leq \CC\cdot\norm{\sigma-\round{\sigma}_\delta}_1\leq \yangnote{\CC d\delta \leq {\varepsilon\over |V|}}.$$

Hence, Lemma~\ref{lemma:discretizing a Bayesnet}~\footnote{Even though Lemma~\ref{lemma:discretizing a Bayesnet} is only proved for finite alphabet, the same proof holds when the alphabet is infinite.} implies that  (iii) $\norm{p-p'}_{TV}\leq \varepsilon$.

Next, we construct the rounded distribution $p''$ of $p'$ by the following coupling. For any sample $x'$ drawn from $p'$, create a sample $x''=\round{x'}_\delta$ from $p''$. It is not hard to verify that $p''$ can also be captured by a Bayesnet defined on the same DAG as $p$ and $p'$. In particular, for every node $v$, every $x_v\in \Sigma_\delta$, and $x_{\Pi_v}\in \Sigma_\delta^{\Pi_v}$, the conditional probability is $$p''_{X_v|X_{\Pi_v}}\left(x_v|x_{\Pi_v}\right )=\int_{x_v}^{x_v+\delta} p'_{X_v|X_{\Pi_v}}\left(z|x_{\Pi_v}\right ) dz.$$ As $p''$ is the rounded distribution of $p'$, $\tilde{p}$ is the rounded distribution of $p$, and $\norm{p-p'}_{TV}\leq \varepsilon$, it must be the case  that $\norm{p''-\tilde{p}}_{TV}\leq \varepsilon$.

\item \textbf{(Step 2:)} Let $\cal P$ be the set of all Bayesnets defined on a DAG with $n$ nodes and has in-degree at most $d$, and have alphabet $\Sigma_\delta$. We argue that there is a size $B=n^{dn}\cdot\left({n |\Sigma_\delta| \over \varepsilon}\right)^{n|\Sigma_\delta|^{d+1}}$ $\varepsilon$-cover  $\cal P'$, in total variation distance, over $\cal P$, and $\cal P'\subset \cal P$. This follows from the same argument in the proof of Theorem~\ref{thm:sample complexity Bayesnets finite alphabet}. First, there are $n^{dn}$ different DAGs with $n$ nodes and in-degree at most $d$. Second, for each DAG there are at most $n \cdot |\Sigma_\delta|^d$ conditional distributions. Finally, it suffices to ${\varepsilon \over n}$-cover each conditional distribution, in total variation distance, according to Lemma~\ref{lemma:discretizing a Bayesnet}. This can be accomplished by a discrete set of at most $\left({n |\Sigma_\delta| \over \varepsilon}\right)^{|\Sigma_\delta|}$-many distributions. Since $p''\in {\cal P}$ and $\norm{p''-\tilde{p}}_{TV}\leq \varepsilon$, there exists a Bayesnet $\tilde{q}$ from the $\varepsilon$-cover for $\cal P$ such that $\norm{\tilde{q}-\tilde{p}}_{TV}\leq 2\varepsilon$.

\item \textbf{(Step 3:)} Since $\min_{q'\in \cal P'} \norm{q'-\tilde{p}}_{TV} \leq 2\varepsilon$. Use a tournament-style  density estimation algorithm; see~e.g.~\cite{devroye2012combinatorial,DaskalakisK14,AcharyaJOS14} and their references, we can learn a Bayesnet $q$ from the $\cal P'$ with $O\left( {\log B \over \varepsilon^2}\right)$ samples from $\tilde{p}$ such that $\norm{q-\tilde{p}}_{TV}=O(\varepsilon)$. 
\end{itemize}

To sum up, we can learn a Bayesnet $q$ defined on a DAG with in-degree at most $d$ using $O\left({ d |V|\log |V| + |V| \cdot \left({H|V|d\CC\over \varepsilon}\right)^{d+1} \log\left({|V|{H|V|d\CC} \over \varepsilon^2}\right)\over \varepsilon^2}\right)$ samples from $p$ such that $\norm{q-p}_P\leq \varepsilon$. If the DAG where $p$ is defined on is known, the sample complexity improves to $O\left({ |V| \cdot \left({H|V|d\CC\over \varepsilon}\right)^{d+1} \log\left({|V|{H|V|d\CC} \over \varepsilon^2}\right)\over \varepsilon^2}\right)$.
\end{prevproof}}

\end{document}